\documentclass[ba,preprint]{imsart}

\RequirePackage{amsthm,amsmath,amsfonts,amssymb}
\RequirePackage[authoryear]{natbib}
\RequirePackage[colorlinks,citecolor=blue,urlcolor=blue,backref=page,backref=page]{hyperref}
\RequirePackage{graphicx}

\pubyear{2024}
\arxiv{2311.05532}
\volume{TBA}
\issue{TBA}
\firstpage{1}
\lastpage{1}

\usepackage{amsmath}
\usepackage{mathtools}
\usepackage{graphicx}
\usepackage{enumerate}
\usepackage{enumitem}
\usepackage{url}
\usepackage{subfigure}
\usepackage{algorithm}
\usepackage{algpseudocode}
\usepackage{stfloats}
\usepackage{lineno}
\usepackage{hyperref}
\usepackage{bm}
\usepackage{bbm}
\usepackage{amssymb}
\usepackage{enumerate}
\usepackage{xcolor}
\usepackage{float}
\usepackage{tabularx}
\usepackage{tabu}
\usepackage{multicol}
\usepackage{multirow}
\usepackage{colortbl,booktabs,threeparttable}
\usepackage{dcolumn}
\usepackage{flushend}
\usepackage{soul}
\usepackage{ragged2e}

\graphicspath{{Figures/}}

\startlocaldefs
    \theoremstyle{plain}

    \newtheorem{theorem}{Theorem}
    \newtheorem{lemma}{Lemma}
    \newtheorem{corollary}{Corollary}
    \theoremstyle{definition}
    \newtheorem{definition}{Definition}
    \newtheorem{example}{Example}
    
    \newtheorem{remark}{{Remark}}

    \newtheorem{philosophy}{{Philosophy}}
    \newtheorem{assumption}{{Assumption}}
    \newtheorem{insight}{{Insight}}

    \newcommand{\rscl}[1]{\mathrm{#1}}
    \renewcommand{\vec}[1]{\bm #1}
    \newcommand{\rvec}[1]{\mathbf{#1}}
    \newcommand{\vech}[1]{\hat{\bm{#1}}}
    \newcommand{\mat}[1]{\bm #1}
    
    \newcommand{\bb}[1]{\mathbb{#1}}
    \renewcommand{\cal}[1]{\mathcal{#1}}
    
    \newcommand{\Tr}{\operatorname{Tr}}
    
    \newcommand{\R}{\mathbb{R}}
    \newcommand{\E}{\mathbb{E}}

    \renewcommand{\d}{\mathrm{d}}
    \renewcommand{\th}{\text{th}}

    \newcommand{\KL}[2]{\operatorname{KL}\left[{#1}\left\|{#2}\right.\right]}
    \newcommand{\Div}[2]{\operatorname{D}\left[{#1}\left\|{#2}\right.\right]}
    \DeclareMathOperator*{\argmax}{argmax}
    \DeclareMathOperator*{\argmin}{argmin}
    \DeclareMathOperator*{\Ent}{{Ent}}

    \newcommand{\defeq}{\coloneqq}
    \newcommand{\stp}{\hfill $\square$}
    \newcommand{\bfit}[1]{\textbf{\textit{#1}}}
    \newcommand{\captext}[1]{\texorpdfstring{#1}{}}
    \newcommand{\tabincell}[2]{\begin{tabular}{@{}#1@{}}#2\end{tabular}}
    \newcommand{\quotemark}[1]{“#1”}
    
    \definecolor{hl-bg-color}{RGB}{255,255,215}
    \sethlcolor{hl-bg-color} 
    \definecolor{new-magenta}{RGB}{255,0,255}
    \soulregister\cite7
    \soulregister\citep7
    \soulregister\citet7
    \soulregister\ref7
    \soulregister\eqref7

    \usepackage{enumitem}
    \setlist[enumerate]{leftmargin=.3in}
    \setlist[itemize]{leftmargin=.3in}
\endlocaldefs

\begin{document}

\begin{frontmatter}
\title{Uncertainty-Aware Bayes' Rule and Its Applications}
\runtitle{Uncertainty-Aware Bayes' Rule}

\begin{aug}
\author[A]{\fnms{Shixiong}~\snm{Wang}\ead[label=e1]{s.wang@u.nus.edu}\orcid{0000-0001-7928-6918}}
\address[A]{Department of Electrical and Electronic Engineering, Imperial College London \printead[presep={,\ }]{e1}.}

\runauthor{S. Wang}
\end{aug}

\begin{abstract}
Bayes' rule has enabled innumerable powerful algorithms of statistical signal processing and statistical machine learning. However, when model misspecifications exist in prior and/or data distributions, the direct application of Bayes' rule is questionable. 
Philosophically, the key is to balance the relative importance between the prior information and the data evidence when calculating posterior distributions: If prior distributions are overly conservative (i.e., exceedingly spread), we upweight the prior belief; if prior distributions are overly aggressive (i.e., exceedingly concentrated), we downweight the prior belief. The same operation also applies to likelihood distributions, which are defined as normalized likelihoods if the normalization exists. 
This paper studies a generalized Bayes' rule, called uncertainty-aware (UA) Bayes' rule, to technically realize the above philosophy, thus combating model uncertainties in prior and/or data distributions. In particular, the advantage of the proposed UA Bayes' rule over the existing power posterior (i.e., $\alpha$-posterior) is investigated. Applications of the UA Bayes' rule on classification and estimation are discussed: Specifically, the UA naive Bayes classifier, the UA Kalman filter, the UA particle filter, and the UA interactive-multiple-model filter are suggested and experimentally validated.
\end{abstract}



\end{frontmatter}

\section{Introduction} \label{sec:introdction}
Bayes' rule has countless successful applications in statistical signal processing and statistical machine learning, such as Kalman filter, particle filter, Bayesian hypothesis testing, naive Bayes classifier, Bayesian optimization, Bayesian bandits, and Bayesian networks. It provides a law for calculating the posterior distribution $p(\vec \theta|\vec y)$ after observing the sample $\vec y$ using the likelihood function $\vec \theta \mapsto p(\vec y|\vec \theta) \defeq p_{\vec \theta}(\vec y)$ and the prior distribution $p(\vec \theta)$:
\begin{equation}\label{eq:bayes-rule}
p(\vec \theta|\vec y) \propto p(\vec \theta) \cdot p(\vec y|\vec \theta),~~~\forall \vec y \in \cal Y,~\forall \vec \theta \in \Theta,
\end{equation}
where $\cal Y \subseteq \R^l$ is the domain of the data-generating distribution $p(\vec y | \vec \theta)$, for every $\vec \theta \in \Theta$, and $\Theta \subseteq \R^d$ is the parameter space. Depending on different applications, two interpretations of Bayes' rule \eqref{eq:bayes-rule} are standard; see \citet{gelman2013bayesian}. In the first interpretation (i.e., subjective Bayes), the measurement $\vec y$ is drawn from a data-generating distribution $p_{\vec \theta_0}(\vec y)$, which is parametrized by a deterministic but unknown parameter $\vec \theta_0 \in \Theta$. In this case, we aim to estimate the true value $\vec \theta_0$, and the prior distribution $p(\vec \theta)$ serves as a normalized epistemic belief that $\vec \theta \in \Theta$ is the true value. Namely, $\vec \theta_0$ is not actually a realization of the distribution $p(\vec \theta)$; instead, $p(\vec \theta)$ represents a human's normalized subjective probabilistic guess about $\vec \theta_0$ and $p(\vec \theta)$ may have no actual physical connection to $\vec \theta_0$. Therefore, we do not assume the physical existence of a joint data-generating distribution $p(\vec y, \vec \theta)$, although we mathematically adopt it. Consequently, with the likelihood function $p(\vec y|\vec \theta)$ indicating our $\vec y$-data evidence that $\vec \theta$ is the true value, the posterior distribution $p(\vec \theta|\vec y)$ denotes our integrated posterior belief that $\vec \theta$ is the true value. In the second interpretation (i.e., latent-variable Bayes), the sample $\vec y$ is drawn from the data-generating distribution $p_{\vec \theta_0}(\vec y)$ where $\vec \theta_0$ is a realization of the distribution $p(\vec \theta)$. That is, physically, there exists a joint data-generating distribution $p(\vec y, \vec \theta)$, although the random variable $\vec \theta$ is hidden and not directly observable. In this case, Bayes' rule is used to estimate the hidden random variable $\vec \theta$. For real-world engineering examples of the two interpretations, see Appendix \ref{append:interpretation-bayes-rule} in supplementary materials. In what follows, we refer to the two interpretations as deterministic and stochastic ones, respectively, and do not mention the specific connotation if it is clear from the contexts.

From the optimization-centric perspective, the Bayesian posterior \eqref{eq:bayes-rule} solves the following problem
\begin{equation}\label{eq:bayesian-rule-opt}
   p(\vec \theta|\vec y) = \displaystyle \argmin_{q(\bm \theta)} \E_{\vec \theta \sim q(\vec \theta)}\left[-\ln p(\bm y |\bm \theta) \right] + \KL{q(\bm \theta)}{p(\bm \theta)},
\end{equation}
where $\KL{q(\bm \theta)}{p(\bm \theta)}$ denotes the Kullback–Leibler (KL) divergence of $q(\bm \theta)$ from $p(\bm \theta)$; see, e.g., \citet{blei2017variational,knoblauch2022optimization}. This perspective intuitively supports the validity of the Bayesian posterior \eqref{eq:bayes-rule}: it maximizes the $\bm y$-data evidence conveyed in $p(\bm y |\bm \theta)$ and simultaneously minimizes the deviation from the prior knowledge $p(\bm \theta)$.

In using the conventional Bayes' rule \eqref{eq:bayes-rule}, the fundamental assumption is that the data distribution $\vec y \mapsto p(\vec y | \vec \theta)$ for every $\vec \theta \in \Theta$ and the prior distribution $p(\vec \theta)$ are accurate. 
However, this assumption is highly untenable in practice, and as a result, the performance of algorithms relying on the conventional Bayesian posterior \eqref{eq:bayes-rule} significantly degrades \citep[Chapter~15]{huber2009robust}; \citep{knoblauch2022optimization}. 
For specific examples and illustrations from the general Bayesian statistics community, see, e.g., \citet{berger1994overview,insua2000bayesian,ruggeri2005robust,gelman2013bayesian,medina2022robustness}; for those from the Bayesian statistical signal processing community, see, e.g., \citet{petersen1999robust,zoubir2018robust,shmaliy2018comparing,wang2022thesisdistributionally}. In practice, both the prior distribution $p(\vec \theta)$ and the data-generating distribution $p(\vec y | \vec \theta)$ can be {independently} uncertain, compared to their unknown ground truths $p_0(\vec \theta)$ and $p_0(\vec y | \vec \theta)$, respectively;\footnote{Note that $p_0(\vec y | \vec \theta)$ and $p(\vec y | \vec \theta_0)$ denote different concepts: the former models the true data-generating law, while the later means the nominal (i.e., assumed) data-generating law at $\vec \theta_0$.} for motivating examples, see Appendix \ref{append:interpretation-bayes-rule} in supplementary materials. Therefore, the uncertainties in the two distributions should be independently treated in calculating the posterior distribution. Note that, in the deterministic interpretation, the true prior distribution $p_0(\vec \theta)$ should be understood as the delta distribution $\delta_{\vec \theta_0}(\vec \theta)$ concentrated at the true value $\vec \theta_0$, so the resulting true posterior distribution would be $\delta_{\vec \theta_0}(\vec \theta)$ as well.

Facing the model uncertainties before applying Bayes' rule, one way is to improve the modeling accuracy and reduce such uncertainties whenever new domain knowledge arrives, for example, by replacing Gaussian distributions with Student's $t$ distributions when outliers exist \citep[Chapter~17]{gelman2013bayesian}; approaches in line with this principle are frequently referred to as adaptive methods. When improving the modeling accuracy is practically challenging due to, e.g., limited domain knowledge, the other way that tolerates or withstands the model uncertainties can be attractive to practitioners; methods aligning with this principle are called robust methods. The connotation of robustness in Bayesian inference is as follows.
\begin{philosophy}[Robustness]\label{phi:robustness}
    Let the true Bayesian posterior $p_0(\vec \theta | \vec y)$ be defined by the true prior $p_0(\vec \theta)$ and the true data distribution $p_0(\vec y | \vec \theta)$. Similarly, let the nominal Bayesian posterior $p(\vec \theta | \vec y)$ be defined by the nominal prior $p(\vec \theta)$ and the nominal data distribution $p(\vec y | \vec \theta)$. The robustness of a generalized posterior $p_g(\vec \theta | \vec y)$ refers to its ability to outperform the nominal Bayesian posterior $p(\vec \theta | \vec y)$, when at least one of $p(\vec \theta)$ and $p(\vec y | \vec \theta)$ is uncertain \citep{medina2022robustness,knoblauch2022optimization}. Note that $p_g(\vec \theta | \vec y)$ is defined by nominal distributions $p(\vec \theta)$ and $p(\vec y | \vec \theta)$. Here, given the sample $\vec y$, outperformance means, \bfit{among many others}: 1) $p_g(\vec \theta | \vec y)$ is closer to the true Bayesian posterior  $p_0(\vec \theta | \vec y)$ than $p(\vec \theta | \vec y)$ under some statistical distances or divergences; or 2) in the deterministic interpretation, the mean-squared estimation error of $p_g(\vec \theta | \vec y)$ for $\vec \theta_0$ is smaller than that of $p(\vec \theta | \vec y)$, i.e., $\int_{\Theta} (\vec \theta - \vec \theta_0)^\top(\vec \theta - \vec \theta_0) p_g(\vec \theta | \vec y) \d \vec \theta < \int_{\Theta} (\vec \theta - \vec \theta_0)^\top(\vec \theta - \vec \theta_0) p(\vec \theta | \vec y) \d \vec \theta$; or 3) the posterior cost $\int_{\Theta} u(\vec \theta, \vec y) p_g(\vec \theta | \vec y) \d \vec \theta$ is smaller than $\int_{\Theta} u(\vec \theta, \vec y) p(\vec \theta | \vec y) \d \vec \theta$, for some cost functions $u$; or 4) the population risk\footnote{This Bayesian decision analysis framework is standard for statistical machine learning and signal processing, e.g., classification and estimation tasks \citep[Chapter~5]{murphy2022probabilistic}. For extended investigations, see Subsection \ref{subsec:applications}. All real-world experiments in Subsection \ref{sec:applications} are based on this paradigm.} $\int_{\cal Y} \int_{\Theta} c[\vech \theta_g(\vec y), \vec \theta] p_0(\vec y, \vec \theta) \d \vec \theta \d \vec y$ is smaller than $\int_{\cal Y} \int_{\Theta} c[\vech \theta(\vec y), \vec \theta] p_0(\vec y, \vec \theta) \d \vec \theta \d \vec y$, for some risk functions $c$, where the estimate $\vech \theta_g(\vec y)$ of $\vec \theta$ is defined by $p_g(\vec \theta | \vec y)$ and $\vech \theta(\vec y)$ is by $p(\vec \theta | \vec y)$. 
    Thus, robustness is application-specific.
    \stp
\end{philosophy}


In what follows, when we do not specify whether a distribution is true or nominal, it is understood to be nominal. According to \citet[Section~15.2]{huber2009robust}, the robust design can be achieved by two operations: a) modifying the prior distribution (i.e., prior belief) and/or the data distribution (i.e., data evidence); or b) balancing the relative importance between them. This proposal is motivationally valid because the key to information fusion indeed lies in properly weighing the information from independent sources \citep{koliander2022fusion}, specifically between prior beliefs and data evidence in Bayesian inference. The boundary between the two operations, however, is not sharp because to make one distribution noninformative (e.g., to widen its spread on $\Theta$) is to downweight its relative importance, compared to the other, in calculating the posterior; cf. \eqref{eq:bayes-rule}. Four technical approaches that realize the above philosophy of Bayesian robustness include the following.
\begin{enumerate}[label=A\arabic*)]
    \item Employing noninformative priors (e.g., uniform or large-entropy distributions\footnote{In the deterministic interpretation of Bayes' rule, the most informative priors are delta distributions $\delta_{\vec \theta_0}(\vec \theta)$, whereas in the stochastic interpretation, the most informative priors are the true ones $p_0(\vec \theta)$. For both cases, the most noninformative priors are uniform distributions.}) to admit the lack of prior knowledge  \citep{berger1994overview}. In this case, the data distribution $p(\bm y | \bm \theta)$ is assumed to be accurate, and Bayes' rule \eqref{eq:bayes-rule} calculating the posterior remains unchanged.

    \item Modifying the data distribution to better accommodate data evidence under the likelihood uncertainty \citep{nguyen2019optimistic}; \citep[Chapters~4 and 5]{huber2009robust}. For a given $\bm \theta$, this amounts to adjusting its $\bm y$-data evidence $p(\bm y | \bm \theta)$. In this case, the prior distribution $p(\bm \theta)$ is assumed to be accurate and the rule generating the posterior stays the same as \eqref{eq:bayes-rule}.

    \item Leveraging the maximum entropy scheme to adjust, and hence balance the relative importance between, the prior distribution $p(\bm \theta)$ and the data distribution $p(\bm y |\bm \theta)$ \citep{wang2022distributionally}. To be specific, if $p(\bm \theta)$ is uncertain, we replace it with the maximum entropy distribution within its neighbor (e.g., a distributional ball), thus downweighting the impact of the prior belief on the posterior distribution. That is, in calculating the posterior distribution using \eqref{eq:bayes-rule}, the prior distribution $p(\bm \theta)$ is replaced with the following maximum entropy distribution $f^*(\bm \theta)$:
    \begin{equation}\label{eq:max-ent-intro}
        \begin{array}{cc}
        f^*(\bm \theta) = \displaystyle \argmax_{f(\bm \theta)} & \displaystyle \int_{\Theta} -f(\bm \theta) \ln f(\bm \theta) \d \bm \theta \\
        \text { s.t. } & 
        \left\{
            \begin{array}{l}
            S(f(\bm \theta), p(\bm \theta)) \leq \epsilon \\
            \int_{\Theta} f(\bm \theta) \d \bm \theta = 1.
            \end{array}
        \right.
        \end{array}
    \end{equation}
    where $S(f(\bm \theta), p(\bm \theta))$ defines a similarity measure between two distributions $f(\bm \theta)$ and $p(\bm \theta)$, and $\epsilon \ge 0$ is the uncertainty quantification level of $p(\bm \theta)$. The larger the radius $\epsilon$ is, the larger the entropy of $f^*(\bm \theta)$ is, and therefore, the more uniformly $f^*(\bm \theta)$ spreads on $\Theta$ and the more the prior distribution $p(\bm \theta)$ is downweighted. The same operation applies to the data distribution $p(\bm y |\bm \theta)$ for every $\vec \theta$, if it is uncertain, to reduce the impact of the data distribution on the posterior distribution. As a result, the relative importance of the prior distribution and the data distribution is determined by their corresponding uncertainty radii. Note that, in this case, Bayes' rule \eqref{eq:bayes-rule} is not altered as well.

    \item Modifying Bayes' rule \eqref{eq:bayes-rule}. One trending and computationally efficient representative in this category is the $\alpha$-posterior\footnote{It is also known as the power posterior, tempered posterior, or coarsened posterior.} by introducing a nonnegative parameter $\alpha$ in calculating the posterior \citep{alquier2020concentration,medina2022robustness}; \citep[Section~6.8.5~and~Chapter~8]{ghosal2017fundamentals}  
    \begin{equation}\label{eq:alpha-posterior-intro}
        p_{\alpha}(\vec \theta|\vec y) \propto p(\vec \theta) \cdot p^\alpha(\vec y|\vec \theta),~~~~~~ \alpha \ge 0,
    \end{equation}
    which solves an $\alpha$-weighted version of \eqref{eq:bayesian-rule-opt}, i.e.,
    \begin{equation}\label{eq:alpha-posterior-intro-opt}
       p_{\alpha}(\vec \theta|\vec y) = \displaystyle \argmin_{q(\bm \theta)}~~ \alpha \cdot \E_{\vec \theta \sim q(\vec \theta)}\left[-\ln p(\bm y |\bm \theta) \right] + \KL{q(\bm \theta)}{p(\bm \theta)},~~~~~~ \alpha \ge 0.
    \end{equation}
    Through the above formula, the prior distribution is assumed to be accurate and the importance of the data distribution is reduced if $0 \le \alpha < 1$, remains unchanged if $\alpha = 1$, and is increased if $\alpha > 1$. 
    Other generalizations to \eqref{eq:bayesian-rule-opt} propose to replace the negative logarithm cost function $-\ln p(\bm y |\bm \theta)$ with a general user-designed version $\rho: \cal Y \times \Theta \to \R_+$, and/or substitute the Kullback–Leibler divergence with other user-specified divergences $\operatorname{D}$ such as the Rényi divergence \citep[Table~1]{knoblauch2022optimization}, giving rise to the Rule-of-Three framework \citep{knoblauch2022optimization} 
    \begin{equation}\label{eq:bayesian-rule-of-three}
       \displaystyle \min_{q(\bm \theta)} \E_{\vec \theta \sim q(\vec \theta)} \left[ \rho(\bm y, \bm \theta) \right] + \Div{q(\bm \theta)}{p(\bm \theta)}.
    \end{equation}
    Note that \eqref{eq:bayesian-rule-of-three} encompasses \eqref{eq:alpha-posterior-intro-opt} as a special case. The $\alpha$-posterior \eqref{eq:alpha-posterior-intro} is well-known to be robust against the misspecified likelihood function (i.e., data distribution) \citep{wu2023comparison,medina2022robustness}, while the posteriors under the Rule-of-Three \eqref{eq:bayesian-rule-of-three} can combat uncertainties in both the prior distribution and the data distribution \citep{knoblauch2022optimization}.
\end{enumerate}

Despite the remarkable capabilities of the existing approaches in addressing uncertainties in both the prior and data distributions, four issues naturally arise:
\begin{enumerate}[label=I\arabic*)]
    \item The $\alpha$-posterior, although computationally attractive due to its closed-form feature, is not reliable in addressing uncertainties in the prior distribution \citep{knoblauch2022optimization}. To be specific, if $\alpha > 1$, the $\alpha$-posterior tends to be robust against the uncertainties in the prior distribution because the prior belief is downweighted relative to the data distribution, but the $\alpha$-posterior simultaneously tends to become concentrated (e.g., a delta distribution) as $\alpha$ increases; see Remark \ref{rem:deficiency-alpha-overconcentration}. This concentration property reduces the spread of the posterior distribution so that in the deterministic interpretation of Bayes' rule, the posterior distribution has no sufficient ability for uncertainty quantification, i.e., the ability to cover the true parameter $\bm \theta_0$ with high probability; see \citet{wu2023comparison}. In contrast, if $0 \le \alpha < 1$, the posterior's over-concentration can be avoided, but at the expense of not being robust against the uncertainties in the prior distribution. For justifications on the deficiency of the $\alpha$-posterior, see \citet[Appendix~B.3.1]{knoblauch2022optimization}. For a motivational example, see below.
    \begin{example}[Deficiency of $\alpha$-Posterior: Over-Concentration]\label{ex:alpha-post-deficiency}
    We consider a deterministic interpretation of Bayes' rule. Suppose the true data-generating distribution is a one-dimensional unit-variance Gaussian density $\cal N(y; 0, 1)$, that is, $\theta_0 = 0$. At measurement $y = 0.1$, the data evidence of $\theta$ (i.e., normalized likelihood function) is $l_{y = 0.1}(\theta) = \cal N(\theta; 0.1, 1)$. If we adopt one expert's prior belief of $\theta_0$ as $p(\theta) \defeq \cal N(\theta; 1, 1)$, then $p(\theta)$ is an unreliable prior because it is significantly miscentered. In this case, a good strategy to downweight the misleading information in the prior is to let the prior $p(\theta)$ be noninformative by inflating its variance, that is, $p(\theta) \defeq \cal N(\theta; 1, 1/\beta)$ where $0 < \beta < 1$. As a result, the corresponding posterior would be $\cal N(\theta; \frac{0.1 + \beta}{1 + \beta}, \frac{1}{1 + \beta})$. By using a small value for $\beta$ (e.g., 0.01), the error of the posterior mean $\frac{0.1 + \beta}{1 + \beta}$ can be greatly reduced without overly concentrating the posterior. In contrast, if the $\alpha$-posterior is used, the corresponding posterior would be $\cal N(\theta; \frac{1 + 0.1\alpha}{1 + \alpha}, \frac{1}{1 + \alpha})$, which implies that reducing the error of the posterior mean $\frac{1 + 0.1 \alpha}{1 + \alpha}$ by a large $\alpha > 1$ also significantly concentrates the posterior. 
    \stp
    \end{example}

    \item Both the prior distribution $p(\vec \theta)$ and the data distribution $p(\vec y | \vec \theta)$ can be independently uncertain, compared to their ground truths $p_0(\vec \theta)$ and $p_0(\vec y | \vec \theta)$, respectively; see explanations in Appendix \ref{append:interpretation-bayes-rule} in supplementary materials. Therefore, the uncertainties in the two distributions should be independently treated in calculating the posterior distribution, which cannot be directly achieved by the $\alpha$-posterior because in the $\alpha$-posterior, only the uncertainty in the data distribution $p(\vec y | \vec \theta)$ is explicitly handled. For a motivational example, see below.
    \begin{example}[Deficiency of $\alpha$-Posterior: Misleading Delta Prior]\label{ex:alpha-post-deficiency-2}
    As in Example \ref{ex:alpha-post-deficiency}, we let $\theta_0 = 0$ and the true data evidence be $l_{y = 0.1}(\theta) = \cal N(\theta; 0.1, 1)$. Suppose we have a stubborn expert who strongly believes that the true value of $\theta_0$ must be $-0.1$, so the associated prior distribution is a delta distribution $\delta_{-0.1}(\theta)$. Consequently, no matter what value of $\alpha \in [0, \infty)$ is used, the resulting $\alpha$-posterior is $\delta_{-0.1}(\theta)$ as well, which cannot cover the true value $\theta_0 = 0$ and indicates the deficiency of the $\alpha$-posterior. In this case, the only mitigating strategy is to reduce the confidence of the expert, for example, changing $\delta_{-0.1}(\theta)$ into $\cal N(\theta; -0.1, 1)$. Technically, we treat $\delta_{-0.1}(\theta)$ as a Gaussian distribution $\cal N(\theta; -0.1, \epsilon)$ with variance $\epsilon \to 0$. By using a proper $\beta$ to inflate the variance of the prior (i.e., reduce the confidence of the expert), we can modify the prior $\cal N(\theta; -0.1, \epsilon)$ into $\cal N(\theta; -0.1, \epsilon/\beta)$ where $\epsilon/\beta = 1$. Then, fusing with the exact likelihood information $\cal N(\theta; 0.1, 1)$, the resulting posterior becomes $\cal N(\theta; 0, 0.5)$, which is centered at the true value $\theta_0 = 0$.
    \stp
    \end{example}

    \item The maximum entropy scheme and the Rule-of-Three framework, although capable of handling uncertainties in both the prior distribution and the data distribution, tend to be computationally prohibitive in real-time data analytics, such as signal processing, because the existence of closed-form solutions is not necessarily guaranteed when complex discrepancy measures for distributions are used; recall $S$ in \eqref{eq:max-ent-intro} and $\operatorname{D}$ in \eqref{eq:bayesian-rule-of-three}. As a result, sufficient computational resources must be employed to accommodate numerical methods, which excludes the applicability of these methods in computation-limited scenarios, e.g., low-cost embedded signal processors. 
    
    \item Designing proper $S$, $\rho$, and $\operatorname{D}$ in \eqref{eq:max-ent-intro} and \eqref{eq:bayesian-rule-of-three} requires non-trivial domain knowledge, which is application-specific and user-unfriendly. 
\end{enumerate}

To address the four issues above, this paper proposes a new optimization-centric interpretation of Bayes' rule, differently from \eqref{eq:bayesian-rule-opt}, and formulates a new framework for generating uncertainty-aware posteriors.
\begin{itemize}[label=\scriptsize$\bullet$]   
    \item \textbf{New Optimization-Centric Interpretation.} By introducing the concept of \textit{likelihood distribution} $l_{\bm y}(\vec \theta) \defeq \frac{p(\bm y| \vec \theta)}{\int_{\Theta} p(\bm y| \vec \theta) \d \vec \theta}$, i.e., a $\bm y$-parametric distribution of $\bm \theta$, we show that the Bayesian posterior \eqref{eq:bayes-rule} solves the following optimization problem
    \begin{equation}\label{eq:new-opt-intro}
        p(\vec \theta|\vec y) = \argmin_{q(\vec \theta)} \KL{q(\vec \theta)}{p(\vec \theta)} + \KL{q(\vec \theta)}{l_{\bm y}(\vec \theta)} + \Ent q(\vec \theta),
    \end{equation}
    where $\Ent q(\vec \theta) \defeq - \int q(\vec \theta) \ln{q(\vec \theta)} \d \vec \theta$ denotes the differential entropy of $q(\vec \theta)$; see Definition \ref{def:likelihood-dist} and Lemma \ref{lem:posterior-dist}. Model \eqref{eq:new-opt-intro} reveals that the posterior distribution $p(\vec \theta|\vec y)$ is a minimum-entropy (i.e., minimum spread) distribution that is simultaneously close to the prior belief $p(\bm \theta)$ and the $\bm y$-data evidence $l_{\bm y}(\bm \theta)$; NB: the evidence from the data $\bm y$ is encoded in the distribution $l_{\bm y}(\bm \theta)$. We notice that the new interpretation \eqref{eq:new-opt-intro} links the Bayesian inference problem to the probability density fusion problem \citep{koliander2022fusion}. To be specific, in Bayesian inference, the posterior distribution $p(\vec \theta|\vec y)$ is an entropy-regularized KL-barycenter of the prior distribution $p(\vec \theta)$ and the likelihood distribution $l_{\bm y}(\vec \theta)$, while in probability density fusion, there is no entropy-regularizer presented when computing the KL-barycenter of different distributions. 

    \item \textbf{New Uncertainty-Aware Framework.} Aligning with the philosophy of Bayesian robustness discussed before (i.e., to balance the relative importance between $p(\bm \theta)$ and $l_{\bm y}(\bm \theta)$ and to control the spread of the resultant posterior distribution), motivated by \eqref{eq:new-opt-intro}, this paper proposes to study the following optimization problem 
    \begin{equation}\label{eq:opt-gen-posterior-dist-intro}
        \min_{q(\vec \theta)} \alpha_1 \KL{q(\vec \theta)}{p(\vec \theta)} + \alpha_2 \KL{q(\vec \theta)}{l_{\bm y}(\vec \theta)} + \alpha_3 \Ent q(\vec \theta),
    \end{equation}
    where $\alpha_1$, $\alpha_2$, and $\alpha_3$ are weights to emphasize or discount the corresponding terms, respectively; the weights can be negative to let the posterior be away from $p(\vec \theta)$ or $l_{\bm y}(\vec \theta)$, or to let the posterior entail large entropy (i.e., large spread). The posterior distribution solving the generalized uncertainty-aware formulation \eqref{eq:opt-gen-posterior-dist-intro} is given by
    \begin{equation}\label{eq:gen-bayes-rule-intro}
        p_g(\vec \theta | \vec y) \propto p^{\frac{\alpha_1}{\alpha_1 + \alpha_2 -\alpha_3}}(\vec \theta) \cdot l^{\frac{\alpha_2}{\alpha_1 + \alpha_2 - \alpha_3}}_{\bm y}(\vec \theta),
    \end{equation}
    if $\alpha_1 + \alpha_2 > \alpha_3$; see Theorem \ref{thm:gen-posterior-dist}. The existing $\alpha$-posterior can be recovered by setting $\alpha_1 = 1$ and $\alpha_2 = \alpha_3 = \alpha$. For notational simplicity, we rewrite \eqref{eq:gen-bayes-rule-intro} by employing the parameter pair $(\alpha, \beta)$ and suggest the uncertainty-aware Bayes' rule in Definition \ref{def:alpha-beta-posterior}:
    \begin{equation}\label{eq:generalized-bayes-rule}
        p_g(\vec \theta | \vec y) \propto p^{\beta}(\vec \theta) \cdot p^{\alpha}(\vec y | \vec \theta),~~~~~~ 0 \le \alpha, \beta < \infty,
    \end{equation}
    which independently handles the uncertainties in the prior and data distributions. 
    This paper names the uncertainty-aware (UA) posterior distribution in \eqref{eq:generalized-bayes-rule} as the $(\alpha,\beta)$-posterior, generalizing the existing $\alpha$-posterior from a new perspective rather than the scheme \eqref{eq:alpha-posterior-intro-opt} and the Rule-of-Three framework \eqref{eq:bayesian-rule-of-three}. The key and nontrivial distinction between \eqref{eq:gen-bayes-rule-intro} and \eqref{eq:generalized-bayes-rule} is that different parameter pairs $(\alpha_1, \alpha_2, \alpha_3)$ determine whether $\alpha$ is independent of $\beta$; see Table \ref{tab:examples}.
\end{itemize}

Compared to the $\alpha$-posterior strategy, the maximum entropy scheme, and the Rule-of-Three framework, the practical benefits of the proposed UA framework \eqref{eq:opt-gen-posterior-dist-intro} and the UA posterior \eqref{eq:generalized-bayes-rule} are four-fold:
\begin{itemize}[label=\scriptsize$\bullet$]    
    \item The proposed framework \eqref{eq:opt-gen-posterior-dist-intro} can \bfit{independently} achieve robustness against uncertainties in the prior distribution by tuning $\alpha_1 \ge 0$ (i.e., to be large if the prior is reliable and small otherwise) and that in the likelihood distribution by tuning $\alpha_2 \ge 0$. Simultaneously, the framework can restrict the posterior distribution from over-concentration by using small values of $\alpha_3 \ge 0$. Therefore, when $\alpha_1$, $\alpha_2$, and $\alpha_3$ are suitably balanced, Issues I1) and I2) faced by the existing $\alpha$-posterior can be avoided by the $(\alpha,\beta)$-posterior \eqref{eq:generalized-bayes-rule}; NB: $\alpha$ and $\beta$ are determined by $\alpha_1$, $\alpha_2$, and $\alpha_3$; recall Examples \ref{ex:alpha-post-deficiency} and \ref{ex:alpha-post-deficiency-2}. Intuitively, this benefit comes from the extra freedom introduced by the new parameter $\beta$, that is, by \bfit{not} forcing $\alpha_2 = \alpha_3$; see Table \ref{tab:examples}.

    \item The $(\alpha,\beta)$-posterior is as computationally lightweight as the $\alpha$-posterior. This computational benefit cannot be always shared by the existing Rule-of-Three framework under generic complex cost functions $\rho$ and divergences $\operatorname{D}$, as well as the existing maximum entropy scheme under generic complex distributional balls. In addition, no domain knowledge and efforts are required to properly design $S$, $\rho$, and $\operatorname{D}$. Therefore, Issues I3) and I4) can also be attacked by the $(\alpha,\beta)$-posterior. 
    
    \item In the $(\alpha,\beta)$-posterior, the maximum entropy transform in \eqref{eq:max-ent-intro} on the prior and data distributions can be achieved by adjusting parameters $\alpha$ and $\beta$. To be specific, for example, $p^\beta(\bm \theta)$ (after normalization) can have larger or smaller entropy than $p(\bm \theta)$; see Corollary \ref{cor:alpha-scale}. Recall that the entropy value of a distribution quantifies its spread on $\Theta$; the more the prior or likelihood distribution spreads on $\Theta$, the less influence it has on the posterior.

    \item  The $\beta$-scaled prior distribution $p^\beta({\vec \theta})$, after normalization, can be closer in the KL sense to the underlying true prior distribution $p_0({\vec \theta})$ than the original $p({\vec \theta})$, for some $\beta \ge 0$; see Theorem \ref{thm:benefits-alpha-distributions}. The similar statement applies to $\alpha$-scaled $l^{\alpha}_{\vec y}(\vec \theta)$ and $l_{\vec y}(\vec \theta)$. This scaling-for-closeness property justifies the superior robustness of the proposed $(\alpha,\beta)$-posterior over the conventional Bayesian posterior [i.e., $(1,1)$-posterior] and the existing $\alpha$-posterior [i.e., $(\alpha,1)$-posterior]. Specifically, as special cases of the $(\alpha,\beta)$-posterior, in the sense of Philosophy \ref{phi:robustness}, the Bayesian posterior and $\alpha$-posterior are less robust.
\end{itemize}

Philosophically, the uncertainty awareness of the $(\alpha,\beta)$-posterior includes two aspects: conservativeness and aggressiveness. When an information source (i.e., prior belief or data evidence) is believed to be overly aggressive (i.e., exceedingly concentrated), we pursue conservativeness and downweight the information by increasing the spread. Conversely, when the information source is believed to be overly conservative (i.e., exceedingly spread), we pursue aggressiveness and upweight the information by increasing the concentration. Motivated by \eqref{eq:new-opt-intro}, this paper employs the entropy value of a distribution as the measure of its spread. We show that for $\alpha, \beta < 1$, the $(\alpha,\beta)$-posterior can technically reflect the conservativeness philosophy to downweight the prior belief and/or data evidence, as the maximum entropy scheme does; in contrast, for $\alpha, \beta > 1$, the $(\alpha,\beta)$-posterior can reflect the aggressiveness philosophy and upweight the prior belief and/or data evidence; see Corollary \ref{cor:alpha-scale}. In addition, from \eqref{eq:opt-gen-posterior-dist-intro}, we see that the entropy term controls the spread of the resulting posterior: If $\alpha_3$ is large, then the posterior has a small entropy value (i.e., small spread and strong concentration); if $\alpha_3$ is small, then the posterior has a large entropy value (i.e., large spread and weak concentration).

Although the two parameters $\alpha$ and $\beta$ bring excellent freedom for uncertainty awareness in calculating the posterior distribution, their optimal values cannot be determined in a universal way, especially considering different performance measures in diverse real-world applications. To be specific, the theoretical optimality of the parameters relies on the underlying true but unknown data-generating distribution [cf. \citet[Lemma~1]{medina2022robustness}, \citet[Section~3]{wu2023comparison}, and Theorem \ref{thm:benefits-alpha-distributions}], while the practical optimality depends on the specific applications and the employed empirical performance measures (cf. Philosophy \ref{phi:robustness}). 
This dilemma also arises in tuning the coefficient $\alpha$ of the $\alpha$-posterior \eqref{eq:alpha-posterior-intro} \citep[p.~26]{medina2022robustness}; \citep{wu2023comparison}, the radii $\epsilon$'s of the maximum entropy scheme \eqref{eq:max-ent-intro} \citep[Section~V]{wang2022distributionally}, the parameters involved in $\rho$ and $\operatorname{D}$ of the Rule-of-Three framework \eqref{eq:bayesian-rule-of-three} \citep[p.~36,~57]{knoblauch2022optimization}, and the fusion weights of probability density functions \citep[Section~VII]{koliander2022fusion}. Hence, data-driven empirical tuning methods for $\alpha$ and $\beta$ are unavoidable to achieve satisfactory performances for specific real-world problems (see Subsection \ref{subsec:parameter-tuning} for details); a similar fact is also acknowledged by \citet[p.~26]{medina2022robustness}; \citet{wu2023comparison} in the practical use of the $\alpha$-posterior. Key findings of this paper are as follows:
\begin{itemize}[label=\scriptsize$\bullet$]  
    \item In some applications, e.g., maximum \textit{a posteriori} (MAP) classification, only the ratio $\alpha/\beta$ matters, while in others, e.g., Bayesian posterior estimation and non-MAP classification, $\alpha$ and $\beta$ independently contribute to performance. Hence, for MAP classification problems, using the $\alpha$-posterior is sufficient. However, for the general estimation and non-MAP classification problems, the $(\alpha, \beta)$-posterior is necessary. For motivations, recall Examples \ref{ex:alpha-post-deficiency} and \ref{ex:alpha-post-deficiency-2}; for technical details, see Remarks \ref{rem:deficiency-alpha-overconcentration} and  \ref{rem:deficiency-alpha-posterior}, and Subsection \ref{subsec:applications}; for experimental justifications, see Subsection \ref{sec:applications}.

    \item In the literature of the $\alpha$-posterior, to establish desired properties, such as posterior consistency, asymptotic normality, and robustness, $\alpha \le 1$ is usually required; see, e.g., \citet{miller2019robust,medina2022robustness}, as well as the popular SafeBayes algorithm \citep{grunwald2012safe}. However, in this paper, to achieve good empirical performances, it is found that $\alpha$ may be either larger or smaller than one, which is also reported in \citet[Tables~1-4]{wu2023comparison}. The same observation happens for $\beta$ as well. Hence, building favored theoretical properties in one aspect does not necessarily benefit real-world performances in other aspects; see Example \ref{ex:benefits-alpha-distributions} for motivation. For experimental justifications, see Subsection \ref{sec:applications}.
\end{itemize}


\textbf{Notation}: Random and deterministic quantities are denoted by upright and italic symbols (e.g., $\rvec y$ and $\vec y$), respectively. Let $\R^d$ denote the $d$-dimensional real space. We use $p_{\rvec y}(\vec y)$ to denote the probability density (resp. mass) function of $\rvec y$ if $\rvec y$ is continuous (resp. discrete); when it is clear from the contexts, $p(\vec y)$ is used as a shorthand for $p_{\rvec y}(\vec y)$. 
The multivariate  Gaussian distribution with mean $\vec \mu$ and covariance $\mat \Sigma$ is denoted as $\cal N(\vec \mu, \mat \Sigma)$ and its density function as $\cal N(\cdot; \vec \mu, \mat \Sigma)$. Let $[r] \defeq \{1,2,\ldots,r\}$ for an integer $r$. 




\section{Technical and Experimental Details}\label{sec:main-results}
This paper assumes that (probability) density functions exist (with respect to the Lebesgue measure). 
To reduce the presentation length, the main results are only given under probability density functions; for probability mass functions, one may just change integrals to sums (with technically trivial adaptations). 
To reduce notational clutter, we implicitly mean $\vec y \in \cal Y$ and $\vec \theta \in \Theta$ throughout the paper, unless otherwise stated.

\subsection{Uncertainty-Aware Bayes' Rule}
We begin with the concept of likelihood distribution.
\begin{definition}[Likelihood Distribution]\label{def:likelihood-dist}
Let
\begin{equation}\label{eq:likelihood-dist}
l_{\bm y}(\vec \theta) \defeq \frac{p(\bm y| \vec \theta)}{\int_{\Theta} p(\bm y| \vec \theta) \d \vec \theta},~~~\forall \vec y
\end{equation} 
define the \textit{likelihood distribution} of $\vec \theta$ induced by the likelihood function $\vec \theta \mapsto p(\bm y | \vec \theta)$ evaluated at the sample $\bm y$. \stp
\end{definition}

To ensure that \eqref{eq:likelihood-dist} is well-defined, the following assumption is necessary.

\begin{assumption}\label{assump:integral-finiteness}
We assume that ${\int_{\Theta} p(\bm y| \vec \theta) \d \vec \theta} < \infty$, for every $\vec y \in \cal Y$.
\stp
\end{assumption}

In the \bfit{practice} of Bayesian inference, the above assumption is not restrictive. Representative conditions under which Assumption \ref{assump:integral-finiteness} holds include the following:
\begin{itemize}[label=\scriptsize$\bullet$]
    \item \textit{Location Family}: The parameter $\vec \theta$ denotes the location of a distribution on $\Theta = \R^d$, e.g., Gaussian, Laplacian, Student’s \textit{t}, and Cauchy. In this case, $p(\vec y | \vec \theta) \defeq f(\vec y - \vec \theta)$ for some probability density functions $f$ on $\R^d$. We have $ \int_{\R^d} p(\vec y | \vec \theta) \d \vec \theta = \int_{\R^d} f(\vec y - \vec \theta) \d \vec \theta = \int_{\R^d} f(\vec u) \d \vec u = \int_{\R^d} f(\vec y - \vec \theta) \d \vec y = 1$. Location families have wide applications in statistical signal processing and machine learning, e.g., mean estimation.
    
    \item \textit{Finite-Support Family}: The parameter space $\Theta \subseteq \R^d$ is finite (i.e., discrete and bounded), and $\vec y \mapsto p(\bm y| \vec \theta)$ is bounded on $\cal Y$ for every $\vec \theta \in \Theta$. Finite-support families have wide applications in statistical signal processing and machine learning, e.g., classification problems with finitely many classes and Bayesian model averaging with finitely many models. Also, this condition is standard in the practical and numerical implementations of Bayesian inference (e.g., particle filters, Monte--Carlo simulations), where finite approximations of $\Theta$ are required.

    \item \textit{Compact-Support Family}: The parameter space $\Theta \subseteq \R^d$ is compact and $\vec \theta \mapsto p(\bm y| \vec \theta)$ is continuous (or bounded) on $\Theta$. Requiring the compactness of $\Theta$ and the continuity (or boundedness) of $\vec \theta \mapsto p(\bm y| \vec \theta)$ on $\Theta$ is realistic: for example, facing a Gaussian data-generating distribution $\cal N(0, \sigma^2)$, we can assume that the variance $\sigma^2 \le M$ where $M$ is a sufficiently large value. In this case, a bounded-support prior distribution $p(\sigma^2)$ needs to be specified, which, although potentially breaking conjugacy,\footnote{Note that using an inverse-gamma prior (which has an unbounded support) for $\sigma^2$ is a matter of conjugacy and theoretical convenience, rather than an inherent necessity of Bayesian inference.} does not violate the principle of Bayesian inference in practice. This treatment is also standard in numerical implementations (or numerical approximations) of Bayesian inference, where compact approximations of $\Theta$ are required.

    \item \textit{Special Distributions}: The parameter $\vec \theta$ denotes the rate of a Poisson distribution; the parameter $\vec \theta$ denotes the rate of an exponential  distribution; the parameter $\vec \theta$ denotes the success probability of a Bernoulli distribution.
\end{itemize}

\begin{remark}[Improper Likelihood]
    If $\vec \theta \mapsto p(\bm y| \vec \theta)$ is not integrable on $\Theta$, it plays a similar role as an improper prior $\vec \theta \mapsto p(\vec \theta)$. In this case, although theoretical analysis becomes difficult, the practical use of the proposed $(\alpha, \beta)$-posterior is still possible if $\int_{\Theta} p^\beta(\vec \theta) p^\alpha(\vec y | \vec \theta) \d \vec \theta$ is finite. Otherwise, redesign $\Theta$ to achieve propriety.
    \stp
\end{remark}

The likelihood distribution $l_{\bm y}(\vec \theta)$ is a $\vec y$-parametric distribution of $\vec \theta$. A direct result of Definition \ref{def:likelihood-dist} is that the posterior distribution $p(\vec \theta | \vec y)$ given by the conventional Bayes' rule \eqref{eq:bayes-rule} can be expressed as
\begin{equation}\label{eq:bayes-rule-2}
    p(\vec \theta|\vec y) \propto p(\vec \theta) \cdot l_{\bm y}(\vec \theta).
\end{equation}

Hereafter, when it is clear from the contexts, we suppress the notational dependence on $\bm y$ and use $l(\vec \theta)$ as a shorthand for $l_{\bm y}(\vec \theta)$. The lemma below gives an interpretation of the origin of Bayes' rule \eqref{eq:bayes-rule-2}.
\begin{lemma}[Conventional Bayes' Rule]\label{lem:posterior-dist}
The posterior distribution $p(\vec \theta | \vec y)$ given by Bayes' rule \eqref{eq:bayes-rule} [or \eqref{eq:bayes-rule-2}] solves
\begin{equation}\label{eq:opt-posterior-dist}
\min_{q(\vec \theta)} \KL{q(\vec \theta)}{p(\vec \theta)} + \KL{q(\vec \theta)}{l(\vec \theta)} + \Ent q(\vec \theta).
\end{equation}
\end{lemma}
\begin{proof}
See Appendix \ref{append:posterior-dist} 
in supplementary materials.
\end{proof}

Lemma \ref{lem:posterior-dist} indicates that the posterior distribution $p(\vec \theta|\vec y)$ is an entropy-regularized equal-weight KL-barycenter of the prior distribution $p(\vec \theta)$ and the likelihood distribution $l(\vec \theta)$; intuitively, $p(\vec \theta|\vec y)$ is a minimum-entropy distribution that is simultaneously close to both the prior distribution (i.e., prior belief) and the likelihood distribution (i.e., $\vec y$-data evidence). 

Aligning with the philosophy of Bayesian robustness discussed before (i.e., to balance the relative importance between $p(\bm \theta)$ and $l(\bm \theta)$ and to control the spread of the resultant posterior distribution), we generalize the optimization problem \eqref{eq:opt-posterior-dist} to
\begin{equation}\label{eq:opt-gen-posterior-dist}
\min_{q(\vec \theta)} \alpha_1 \KL{q(\vec \theta)}{p(\vec \theta)} + \alpha_2 \KL{q(\vec \theta)}{l(\vec \theta)} + \alpha_3 \Ent q(\vec \theta),
\end{equation}
where $\alpha_1$, $\alpha_2$, and $\alpha_3$ are independent weights. To let $q(\vec \theta)$ be close to $p(\vec \theta)$ and $l(\vec \theta)$, without loss of practicality, the following assumption is imposed.
\begin{assumption}\label{assum:alpha}
    We assume that $0 \le \alpha_1 < \infty$, $0 \le \alpha_2 < \infty$, and $-\infty < \alpha_3 < \infty$. \stp
\end{assumption}

Depending on whether to pursue the maximum or minimum entropy of $q(\vec \theta)$, $\alpha_3$ can take any finite value on the whole real line $\R$. The following lemma shows that when $\alpha_3 > \alpha_1 + \alpha_2$, the optimization problem \eqref{eq:opt-gen-posterior-dist} is ill-posed.

\begin{lemma}\label{lemma:problem-finiteness}
If $\alpha_3 > \alpha_1 + \alpha_2$, the minimum objective of \eqref{eq:opt-gen-posterior-dist} is negative infinity, achieved at any delta distributions.
\end{lemma}
\begin{proof}
See Appendix \ref{append:problem-finiteness}
in supplementary materials.
\end{proof}

The above result is intuitively immediate because if $\alpha_3 > \alpha_1 + \alpha_2$, the minimization focuses primarily on finding the minimum-entropy distributions, driving the minimizer(s) to delta distributions. Hereafter, for practicality, we only focus on the case where $\alpha_3 \le \alpha_1 + \alpha_2$, so that \eqref{eq:opt-gen-posterior-dist} can be finite. We term the distribution solving the generalized problem \eqref{eq:opt-gen-posterior-dist} the \textit{uncertainty-aware posterior distribution}.
\begin{theorem}[Uncertainty-Aware Bayes’ Rule]\label{thm:gen-posterior-dist}
Under Assumptions \ref{assump:integral-finiteness} and \ref{assum:alpha}, the UA posterior distribution $p_g(\vec \theta | \vec y)$ solving \eqref{eq:opt-gen-posterior-dist} is given by
\begin{equation}\label{eq:gen-bayes-rule}
    p_g(\vec \theta | \vec y) \propto p^{\frac{\alpha_1}{\alpha_1 + \alpha_2 -\alpha_3}}(\vec \theta) \cdot l^{\frac{\alpha_2}{\alpha_1 + \alpha_2 - \alpha_3}}(\vec \theta),
\end{equation}
or equivalently,
\begin{equation}\label{eq:gen-bayes-rule-2}
    p_g(\vec \theta | \vec y) \propto p^{\frac{\alpha_1}{\alpha_1 + \alpha_2 - \alpha_3}}(\vec \theta) \cdot p^{\frac{\alpha_2}{\alpha_1 + \alpha_2 - \alpha_3}}(\vec y | \vec \theta),
\end{equation}
when $\alpha_3 < \alpha_1 + \alpha_2$, provided that the right-hand-side terms are integrable on $\Theta$. When $\alpha_3 = \alpha_1 + \alpha_2$, $p_g(\vec \theta | \vec y)$ is an arbitrary distribution supported on the set $\Theta^*$ where
$
\Theta^* \defeq \argmax_{\vec \theta} \alpha_1 \ln p(\vec \theta) + \alpha_2 \ln l(\vec \theta)
$ 
contains all weighted maximum \textit {a-posteriori} estimates.
\end{theorem}
\begin{proof}
See Appendix \ref{append:gen-posterior-dist} 
in supplementary materials.
\end{proof}


By rewriting the uncertainty-ware posterior distribution in \eqref{eq:gen-bayes-rule} for notational simplicity, we introduce the $(\alpha,\beta)$-posterior.
\begin{definition}[$(\alpha,\beta)$-Posterior]\label{def:alpha-beta-posterior}
The $(\alpha,\beta)$-posterior induced by the prior distribution $p(\vec \theta)$ and the likelihood distribution $l(\vec \theta)$ is defined as
\begin{equation}\label{eq:alpha-beta-posterior}
    p_g(\vec \theta | \vec y) \propto p^{\beta}(\vec \theta) \cdot l^{\alpha}(\vec \theta),~~~~~0 \le \alpha, \beta < \infty
\end{equation}
if the right-hand-side term is integrable on $\Theta$. When $\alpha$ and $\beta$ are infinity, the $(\alpha,\beta)$-posterior $p_g(\vec \theta | \vec y)$ is defined as an arbitrary distribution supported on $\Theta^*$ where $\Theta^* \defeq \argmax_{\vec \theta} \alpha_1 \ln p(\vec \theta) + \alpha_2 \ln l(\vec \theta)$.
\stp
\end{definition}

Note that, compared with \eqref{eq:gen-bayes-rule}, $\beta \defeq \frac{\alpha_1}{\alpha_1 + \alpha_2 - \alpha_3} \in [0, \infty)$ and $\alpha \defeq \frac{\alpha_2}{\alpha_1 + \alpha_2 - \alpha_3} \in [0, \infty)$, if $\alpha_3 < \alpha_1 + \alpha_2$. When $\alpha_3 \uparrow (\alpha_1 + \alpha_2)$, $\alpha$ and $\beta$ simultaneously tend to infinity. 

\begin{remark}
The value of $\alpha_3$ controls the domain of $(\alpha, \beta)$, and therefore, the independence between $\alpha$ and $\beta$; for example, if $\alpha_3 \ge 0$ or $\alpha_3 = 0$ is additionally required, then $\alpha + \beta \ge 1$ or $\alpha + \beta = 1$ must be appended in Definition \ref{def:alpha-beta-posterior}. On the other hand, if we allow $\alpha_1$ and $\alpha_2$ to be negative, then $\beta$ and $\alpha$ can be negative as well. Negative values for $\alpha_1$ (resp. $\alpha_2$) imply that $q(\vec \theta)$ is expected to be far away from $p(\vec \theta)$ [resp. $l(\vec \theta)$]. Without loss of practicality, this paper focuses on the specifications in Assumption \ref{assum:alpha}.
\stp
\end{remark}


Table \ref{tab:examples} gives illustrative examples of the $(\alpha,\beta)$-posterior; for detailed technical conditions, derivations, and insights, see Appendix \ref{append:examples-bayes-rule} in supplementary materials. Note that different parameter pairs $(\alpha_1, \alpha_2, \alpha_3)$ determine whether $\alpha$ is independent of $\beta$, as well as the domains of $\alpha$ and $\beta$.

\begin{table*}[!htbp]
\caption{Particular Examples of $(\alpha,\beta)$-Posterior.}
\centering
\begin{tabularx}{\textwidth}{X|X|X}
\hline
\textbf{Name} & \textbf{Definition} & \textbf{Conditions} \\
\hline

       $\alpha$-posterior       &     $p(\vec \theta) \cdot l^{\alpha}(\vec \theta)$     &   $\alpha_2 = \alpha_3$ and $\alpha_1 > 0$ \\
       $\beta$-posterior        &     $p^{\beta}(\vec \theta) \cdot l(\vec \theta)$   &  $\alpha_1 = \alpha_3$ and $\alpha_2 > 0$ \\
       $\gamma$-posterior       &     $p^{\gamma}(\vec \theta) \cdot l^{\gamma}(\vec \theta)$     &  $\alpha_1 = \alpha_2$ and $2\alpha_1 > \alpha_3$ \\ 
       $\alpha$-likelihood      &     $l^{\alpha}(\vec \theta)$    &  $\alpha_1 = 0$ and $\alpha_2 > \alpha_3$ \\
       $\alpha$-prior           &     $p^{\alpha}(\vec \theta) $    &  $\alpha_2 = 0$ and $\alpha_1 > \alpha_3$ \\
       $\alpha$-pooled posterior  &   $p^{\alpha}(\vec \theta) \cdot l^{1-\alpha}(\vec \theta)$     &  $\alpha_3 = 0$ \\ 
\hline
       \multicolumn{3}{l}{\tabincell{l}{Note: In $\alpha$-pooled posterior, $\alpha \in [0, 1]$; other $\alpha$, $\beta$, and $\gamma$ take values on $[0, \infty)$.}}
\end{tabularx}
\label{tab:examples}
\end{table*}

Below, leveraging the concept of the likelihood distribution and the optimization formulation \eqref{eq:opt-gen-posterior-dist}, we explain one of the deficiencies of the $\alpha$-posterior. Note that in \eqref{eq:opt-gen-posterior-dist}, to recover the $\alpha$-posterior, we have $\alpha_1 \defeq 1$ and $\alpha_2 = \alpha_3 \defeq \alpha$; cf. \eqref{eq:alpha-posterior-intro-opt}.

\begin{remark}[Deficiency of $\alpha$-Posterior: Over-Concentration]\label{rem:deficiency-alpha-overconcentration}
When the prior distribution $p(\vec \theta)$ is misleading, it should be downweighted. In the $\alpha$-posterior, this is realized by using a large $\alpha$ [see \eqref{eq:alpha-posterior-intro-opt}], causing severe posterior concentration, because in \eqref{eq:opt-gen-posterior-dist}, the entropy term is penalized by a large coefficient $\alpha_3 = \alpha$. In addition, when $p(\vec \theta)$ needs to be entirely disregarded using $\alpha \to \infty$, the scaled likelihood $l^\alpha(\vec \theta)$, after normalization, collapses to a concentrated distribution supported on maximum-likelihood estimates (e.g., a delta distribution whose entropy value is negative infinity, if $l(\vec \theta)$ is unimodal). Consequently, the posterior concentrates at maximum-likelihood estimates, e.g., a delta distribution if the estimate is unique. However, in the $(\alpha,\beta)$-posterior, to ignore the misleading prior $p(\vec \theta)$, we just need $\beta = 0$ without overly concentrating the posterior.
\stp
\end{remark}

\subsection{Properties of Uncertainty-Aware Bayes' Rule}
The definition of the $(\alpha,\beta)$-posterior motivates us to study the properties of $\alpha$-scaled distributions.
\begin{definition}[$\alpha$-Scaled Distribution]\label{def:alpha-scale}
The $\alpha$-scaled distribution $h^{(\alpha)}(\vec \theta)$ induced by the distribution $h(\vec \theta)$ is defined as 
\begin{equation}\label{eq:alpha-scale}
    h^{(\alpha)}(\vec \theta) \defeq \frac{h^{\alpha}(\vec \theta)}{\int h^{\alpha}(\vec \theta) \d \vec \theta},
\end{equation}
for $0 \le \alpha < \infty$, if $\int h^{\alpha}(\vec \theta) \d \vec \theta$ exists.
\stp
\end{definition}

As the definition implies, $h^{(\alpha)}(\vec \theta) \equiv h(\vec \theta)$ for every $\alpha$ if $h(\vec \theta)$ is a uniform distribution on a bounded support set $\Theta$. 
Inspired by \eqref{eq:opt-gen-posterior-dist}, we investigate the relation between the entropy of $h^{(\alpha)}(\vec \theta)$ and that of $h(\vec \theta)$, that is, how $\alpha$ controls the spread of $h^{(\alpha)}$.

\begin{theorem}\label{thm:alpha-scale-ent-monotonicity}
Let $h(\vec \theta)$ not be a uniform distribution. The function $\alpha \mapsto \Ent h^{(\alpha)}(\vec \theta)$ is monotonically decreasing in $\alpha$ on $[0, \infty)$.
\end{theorem}
\begin{proof}
    See Appendix \ref{append:alpha-scale-ent-monotonicity} in supplementary materials.
\end{proof}

Let $E(\alpha)$ denote the entropy difference evaluated at $\alpha \in [0,\infty)$:
$
E(\alpha) \defeq \Ent h^{(\alpha)}(\vec \theta) - \Ent h(\vec \theta).
$ 
Theorem \ref{thm:alpha-scale-ent-monotonicity} implies that $E(\alpha)$ is monotonically decreasing in $\alpha$ on $[0, \infty)$. In addition, it is obvious to see that $E(0) > 0$ and $E(1) = 0$. 
As a result, the following corollary is immediate.
\begin{corollary}\label{cor:alpha-scale}
If $0 \le \alpha < 1$, then $h^{(\alpha)}(\vec \theta)$ has larger entropy than $h(\vec \theta)$; if $1 < \alpha < \infty$, then $h^{(\alpha)}(\vec \theta)$ has smaller entropy than $h(\vec \theta)$.
\stp
\end{corollary}

A concrete example for Corollary \ref{cor:alpha-scale} is as follows; for visual illustrations, see Appendices  \ref{append:illlu-alpha-scaling} and \ref{append:illu-ent-diff} in supplementary materials.
\begin{example}\label{ex:alpha-scaled-Gaussian}
Consider an one-dimensional zero-mean Gaussian density function $h(\theta) \propto \exp{(-\frac{1}{2} \frac{\theta^2}{\sigma^2})}$ where $\sigma^2$ is the variance. The $\alpha$-scaled distribution is $h^{(\alpha)}(\theta) \propto \exp{(-\frac{1}{2} \frac{\theta^2}{\sigma^2/\alpha})}$. Therefore, $h^{(\alpha)}(\theta)$ is a Gaussian density function with variance $\sigma^2/\alpha$. When $0 < \alpha < 1$, we have $\sigma^2/\alpha > \sigma^2$ and therefore $\Ent h^{(\alpha)}(\theta) > \Ent h(\theta)$; when $\alpha > 1$, we have $\sigma^2/\alpha < \sigma^2$ and therefore $\Ent h^{(\alpha)}(\theta) < \Ent h(\theta)$. Note that $\Ent h^{(\alpha)}(\theta) = \frac{1}{2} \ln(2 \pi \sigma^2/\alpha) + \frac{1}{2}$, while $\Ent h(\theta) = \frac{1}{2} \ln(2 \pi \sigma^2) + \frac{1}{2}$.
\stp
\end{example}

Theorem \ref{thm:alpha-scale-ent-monotonicity} and Corollary \ref{cor:alpha-scale} collectively imply the following useful insight.
\begin{insight}[Uncertainty Awareness in Posterior]\label{insight:1}
The $\alpha$-scaling operation controls the entropy (i.e., the spread) of $h^{(\alpha)}(\vec \theta)$. Therefore, the $(\alpha, \beta)$-posterior $p_g(\vec \theta | \vec y) \propto p^{\beta}(\vec \theta) \cdot l^{\alpha}(\vec \theta)
$ balances the relative importance between the prior distribution $p(\vec \theta)$ and the likelihood distribution $l(\vec \theta)$. To be specific, if the likelihood distribution is overly aggressive (i.e., exceedingly concentrated), we use $0 \le \alpha < 1$ so that $\Ent l^{(\alpha)}(\vec \theta)$ can be \bfit{increased} compared to $\Ent l(\vec \theta)$; in contrast, if the likelihood distribution is overly conservative (i.e., exceedingly spread), we use $\alpha > 1$ so that $\Ent l^{(\alpha)}(\vec \theta)$ can be \bfit{reduced} compared to $\Ent l(\vec \theta)$. The same operation also applies to the prior distribution, i.e., $p^{(\beta)}(\bm \theta)$ and $p(\bm \theta)$. Note again, that the larger the entropy value of the prior or likelihood distribution is, the more spread it is on $\Theta$, and therefore, the less impact it has on the posterior distribution.
\stp
\end{insight}

Next, we study the closeness to $h(\vec \theta)$ from $h^{(\alpha)}(\vec \theta)$.
\begin{theorem}\label{thm:closeness-from-alpha-scaled}
The closeness to $h(\vec \theta)$ from $h^{(\alpha)}(\vec \theta)$, i.e., 
$
    \KL{h(\vec \theta)}{h^{(\alpha)}(\vec \theta)},
$ 
is a monotonically increasing function if $1 < \alpha < \infty$ and a monotonically decreasing function if $0 \le \alpha < 1$. In addition, it is a convex function in $\alpha$ on $[0, \infty)$.
\end{theorem}
\begin{proof}
    See Appendix \ref{append:closeness-from-alpha-scaled} in supplementary materials.
\end{proof}

An illustration for Theorem \ref{thm:closeness-from-alpha-scaled} is below; for a visualization, see Appendix \ref{append:illu-closeness} in supplementary materials.
\begin{example}[Continued from Example \ref{ex:alpha-scaled-Gaussian}]
We have
$
\KL{h(\vec \theta)}{h^{(\alpha)}(\vec \theta)} = -\frac{1}{2}\ln\left(\alpha\right) + \frac{\alpha}{2} - \frac{1}{2},
$ for $\alpha > 0$. 
The derivative of $\KL{h(\vec \theta)}{h^{(\alpha)}(\vec \theta)}$ with respect to $\alpha$ is $\frac{1}{2} (1 - \frac{1}{\alpha})$, which is positive if $\alpha > 1$ and negative if $0 < \alpha < 1$; the second-order derivative is $\frac{1}{2\alpha^2} > 0$. Therefore, $\KL{h(\vec \theta)}{h^{(\alpha)}(\vec \theta)}$ is monotonically decreasing when $0< \alpha < 1$ and monotonically increasing when $\alpha > 1$. In addition, $\KL{h(\vec \theta)}{h^{(\alpha)}(\vec \theta)}$ is convex.
\stp
\end{example}

Theorem \ref{thm:closeness-from-alpha-scaled} implies the following useful insight.
\begin{insight}[Level of Uncertainty]\label{insight:2}
The more $\alpha$ deviates from $1$, the farther to $h(\vec \theta)$ from $h^{(\alpha)}(\vec \theta)$. Hence, in the $(\alpha, \beta)$-posterior, the more we trust the prior $p(\vec \theta)$ [resp. the likelihood $l(\vec \theta)$], the closer the value of $\beta$ (resp. $\alpha$) should be to $1$.
\stp
\end{insight}

Insights \ref{insight:1} and \ref{insight:2} collectively suggest the \bfit{empirical} usage of the $(\alpha, \beta)$-posterior in balancing the relative importance between the prior knowledge and the data evidence; for visual illustrations, see Appendix \ref{ex:illu-alpha-beta-posterior} in supplementary materials. Next, we show another property of the $(\alpha, \beta)$-posterior that enables its practical usefulness. Let $h_0(\vec \theta)$ be the true distribution and $h(\vec \theta)$ the nominal distribution serving as an estimate of $h_0(\vec \theta)$.  The theorem below states that there exists some $\alpha \ge 0$ such that the $\alpha$-scaled nominal distribution $h^{(\alpha)}(\vec \theta)$ can be closer to the true distribution $h_0(\vec \theta)$ than the original nominal distribution $h(\vec \theta)$.
\begin{theorem}\label{thm:benefits-alpha-distributions}
Given distributions $h_0$ and $h$ on $\Theta$, if $h$ is not a uniform distribution and 
\begin{equation}\label{eq:benefits-alpha-distributions-condition}
\int_{\Theta} [h_0(\vec \theta) - h(\vec \theta)] \ln h(\vec \theta) \d \vec \theta \neq 0,
\end{equation}
there exists $\alpha \ge 0$ such that
\begin{equation}\label{eq:benefits-alpha-distributions}
    \KL{h_0(\vec \theta)}{h^{(\alpha)}(\vec \theta)} < \KL{h_0(\vec \theta)}{h(\vec \theta)}.
\end{equation}
\end{theorem}
\begin{proof}
    See Appendix \ref{append:benefits-alpha-distributions} in supplementary materials.
\end{proof}

Theorem \ref{thm:benefits-alpha-distributions} indicates that, if a proper $\alpha$ is given, $h^{(\alpha)}$ can be a better (i.e., a more accurate) surrogate of $h_0$ than $h$. However, the parameter $\alpha$ cannot be theoretically specified because it depends on the underlying true but unknown distribution $h_0(\vec \theta)$; to be specific, the best value $\alpha^*$ of $\alpha$ solves 
\begin{equation}
    \alpha^* \in \argmin_{\alpha \ge 0} \KL{h_0(\vec \theta)}{h^{(\alpha)}(\vec \theta)},
\end{equation}
which however cannot be conducted in practice. The example below provides an intuitive demonstration of Theorem \ref{thm:benefits-alpha-distributions}.

\begin{example}\label{ex:benefits-alpha-distributions}
    Consider $h_0 \defeq [0.2, 0.8]$, $h \defeq [0.4, 0.6]$, and $\alpha \defeq 2$. We have $h^{(2)} = [0.3, 0.7]$, which is closer to $h_0$ than $h$. Then, consider $h_0 \defeq [0.4, 0.6]$, $h \defeq [0.2, 0.8]$, and $\alpha \defeq 0.6$. We have $h^{(0.6)} = [0.3, 0.7]$, which is closer to $h_0$ than $h$. Third, consider a continuous case where $h_0(\theta) \defeq \cal N(\theta; 0, \sigma^2_0)$ and $h(\theta) \defeq \cal N(\theta; 0, \sigma^2)$. For any $\alpha > 0$, we have $h^{(\alpha)}(\theta) = \cal N(\theta; 0, \sigma^2/\alpha)$. Letting $\alpha \defeq \sigma^2/\sigma^2_0$, we have $h^{(\alpha)}(\theta) = h_0(\theta)$; NB: $\alpha$ can be either larger or smaller than one. \stp
\end{example}

Theorem \ref{thm:benefits-alpha-distributions} motivates the potential that the proposed $(\alpha, \beta)$-posterior can outperform the conventional Bayes' rule in practice: Scaled nominal prior and likelihood distributions can be closer to the true prior and likelihood distributions, respectively, than the original unscaled nominal ones; see Philosophy \ref{phi:robustness}. Specific investigations on this point, under stochastic and deterministic interpretations of Bayes' rule, are as follows. 

\textbf{Stochastic Interpretation.} Under the stochastic interpretation of Bayes' rule (i.e., there physically exists a joint data-generating distribution), given the true joint distribution $p_0(\vec \theta, \vec y)$ and its surrogate $p(\vec \theta, \vec y)$, for a measurable bounded cost function $u: \Theta \times \cal Y \to [-M, M]$ where $0 \le M < \infty$, we have 
\begin{equation}\label{eq:bayesian-cost}
\begin{array}{l}
    \displaystyle \left| \int_{\Theta} \int_{\cal Y} u(\vec \theta, \vec y) p_0(\vec \theta, \vec y) \d \vec \theta \d \vec y - \int_{\Theta} \int_{\cal Y} u(\vec \theta, \vec y) p(\vec \theta, \vec y) \d \vec \theta  \d \vec y \right| \\
    \quad \quad \le \displaystyle M \int_{\Theta} \int_{\cal Y} \Big| p_0(\vec \theta, \vec y) -  p(\vec \theta, \vec y) \Big| \d \vec \theta \d \vec y
    \\
    \quad \quad \le \displaystyle 2M \sqrt{1- \exp({-\KL{p_0(\vec \theta, \vec y)}{p(\vec \theta, \vec y)}})}.
\end{array}
\end{equation}
The first inequality is due to the definition of the total variation distance, while the second is due to the Bretagnolle–Huber inequality. As a result, at the given observation $\vec y$, we have
\begin{equation}\label{eq:bayesian-cost-posterior}
\begin{array}{l}
    \displaystyle \left| \int_{\Theta} u(\vec \theta, \vec y) p_0(\vec \theta | \vec y) \d \vec \theta - \int_{\Theta} u(\vec \theta, \vec y) p(\vec \theta | \vec y) \d \vec \theta \right| \le \displaystyle 2M \sqrt{1- \exp({-\KL{p_0(\vec \theta | \vec y)}{p(\vec \theta | \vec y)}})}.
\end{array}
\end{equation}
Therefore, finding a better surrogate of $p_0(\vec \theta | \vec y)$ is vital in controlling the deviation from the optimal cost, which can be realized by employing $(\alpha, \beta)$-scaled uncertainty-aware Bayes' rule $p_g(\vec \theta | \vec y)$ in \eqref{eq:alpha-beta-posterior}. However, the optimal parameters $(\alpha, \beta)$ cannot be theoretically specified because they depend on the true posterior distribution, and therefore on the true prior and likelihood distributions, which are unknown in practice; similar parameter-tuning problems also exist in using the $\alpha$-posterior \citep[p.~26,~Lemma~1]{medina2022robustness}; \citep[Section~3]{wu2023comparison}. 

When the cost function $u(\bm \theta, \bm y)$ is specified by the outer product of the estimation error, i.e., $u(\bm \theta, \bm y) \defeq [\hat{\bm \theta} (\bm y) - \bm \theta][\hat{\bm \theta} (\bm y) - \bm \theta]^\top$, where $\hat{\bm \theta} (\bm y)$ is an estimate of $\bm \theta$, the expected error covariance matrix is given by
\begin{equation}\label{eq:bayesian-cost-MSE}
\begin{array}{l}
    \displaystyle \int_{\Theta} \int_{\cal Y} [\hat{\bm \theta} (\bm y) - \bm \theta][\hat{\bm \theta} (\bm y) - \bm \theta]^\top p_0(\vec \theta, \vec y) \d \vec \theta \d \vec y.
\end{array}
\end{equation}
In this case, we are particularly interested in the performance upper bound of the estimator $\hat{\bm \theta} (\bm y)$, that is, the Bayesian Cramér--Rao lower bound of \eqref{eq:bayesian-cost-MSE}. When the underlying true distribution $p_0(\vec \theta, \vec y)$ is known, the optimal estimator minimizing the expected cost in \eqref{eq:bayesian-cost-MSE} is the posterior mean: i.e., $\hat{\bm \theta} (\bm y) \defeq \int \bm \theta p_0(\bm \theta | \bm y) \d \bm \theta$. In real-world operation, however, $\hat{\bm \theta} (\bm y)$ is defined through the misspecified surrogate distribution $p(\vec \theta, \vec y)$. Consequently, as per \citet[Eq.~(5),~Thm.~1]{tang2023parametric}, the corresponding Bayesian Cramér--Rao lower bound of \eqref{eq:bayesian-cost-MSE} depends on $p(\vec \theta, \vec y)$, and this bound can be straightforwardly improved by Theorem \ref{thm:benefits-alpha-distributions} because the definition of the pseudo-true parameter can be refined; for technical details, just compare \citet[Eq.~(5)]{tang2023parametric} with \eqref{eq:benefits-alpha-distributions}.

\textbf{Deterministic Interpretation.} Under the deterministic interpretation of Bayes' rule, the true prior $p_0(\vec \theta)$ is understood as the delta distribution $\delta_{\vec \theta_0}(\vec \theta)$ concentrated at the true value $\vec \theta_0$, and $p_0(\vec y | \vec \theta)$ models the true data-generating law; hence, $p_0(\vec \theta, \vec y) \defeq p_0(\vec y | \vec \theta) p_0(\vec \theta)$ and $p_0(\vec \theta | \vec y) = \delta_{\vec \theta_0}(\vec \theta)$. In this case, the KL divergence cannot be well-defined as in \eqref{eq:bayesian-cost} and \eqref{eq:bayesian-cost-posterior}. However, provided that the cost function $u$ is Lipschitz continuous, \eqref{eq:bayesian-cost} and \eqref{eq:bayesian-cost-posterior} can be re-stated under the order-1 Wasserstein distance. To be specific,  if $|u(\vec \theta, \vec y) - u(\vec \theta', \vec y')| \le M \cdot d\big((\vec \theta, \vec y), (\vec \theta', \vec y')\big)$ for some metrics $d$ on $\Theta \times \cal Y$, we have
\begin{equation}\label{eq:bayesian-cost-W}
\begin{array}{l}
    \displaystyle \left| \int_{\Theta} \int_{\cal Y} u(\vec \theta, \vec y) p_0(\vec \theta, \vec y) \d \vec \theta \d \vec y - \int_{\Theta} \int_{\cal Y} u(\vec \theta, \vec y) p(\vec \theta, \vec y) \d \vec \theta  \d \vec y \right| \le \displaystyle M {d_{\text{W}_1}({p_0(\vec \theta, \vec y)},{p(\vec \theta, \vec y)})},
\end{array}
\end{equation}
where $d_{\text{W}_1}$ denotes the order-1 Wasserstein distance. 
Similarly, given the sample $\vec y$, we have
\begin{equation}\label{eq:bayesian-cost-W-conditional}
\begin{array}{l}
    \displaystyle \left| \int_{\Theta} u(\vec \theta, \vec y) p_0(\vec \theta | \vec y) \d \vec \theta - \int_{\Theta}  u(\vec \theta, \vec y) p(\vec \theta | \vec y) \d \vec \theta  \right| \le \displaystyle M {d_{\text{W}_1}({p_0(\vec \theta | \vec y)},{p(\vec \theta | \vec y)})}.
\end{array}
\end{equation}
Hence, the closeness between ${p_0(\vec \theta | \vec y)} = \delta_{\vec \theta_0} (\vec \theta)$ and ${p(\vec \theta | \vec y)}$ is still important. For an illustrative example, see Figure \ref{fig:goodness-example}, where we let the true value $\theta_0 = 0$ and suppose that the nominal posterior ${p(\theta | y)}$ is a Gaussian distribution $\cal N(\theta; m, \sigma^2)$ with mean $m$ and variance $\sigma^2$. If $\sigma^2 \to 0$, $\cal N(\theta; m, \sigma^2)$ degenerates to the delta distribution $\delta_{m}(\theta)$ concentrated at $m$.

\begin{figure}[!htbp]
    \centering
    \includegraphics[width=4.5cm]{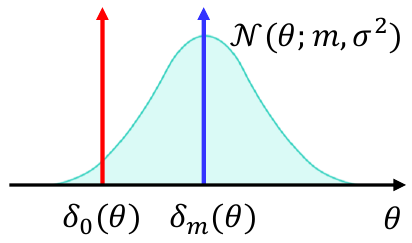}
    \caption{An $\alpha$-scaled version $\cal N(\theta; m, \sigma^2/\alpha)$ can be closer to the ground truth $\delta_0(\theta)$ than the unscaled distribution $\cal N(\theta; m, \sigma^2)$. Note that $\cal N(\theta; m, \sigma^2/\alpha)$ degenerates to $\delta_m(\theta)$ if $\alpha \to \infty$. Note also that for mean $m > 0$ and variance $\sigma^2 > 0$, under the order-$1$ Wasserstein distance, $\delta_m(\theta)$ is closer to $\delta_0(\theta)$ than the Gaussian distribution $\cal N(\theta; m, \sigma^2)$. However, in terms of probabilistic coverage, $\cal N(\theta; m, \sigma^2)$ is better than $\delta_m(\theta)$ because the true value $\theta_0 = 0$ can be included in the support of $\cal N(\theta; m, \sigma^2)$ but cannot be included in that of $\delta_m(\theta)$. Hence, whether $\cal N(\theta; m, \sigma^2/\alpha)$ is favored over $\cal N(\theta; m, \sigma^2)$, for some $\alpha \ge 0$, depends on a specific performance measure; recall Philosophy \ref{phi:robustness}.}
    \label{fig:goodness-example}
\end{figure}

\subsection{Uncertainty-Aware Bayes' Rule for Multiple Samples}
So far, it is implicitly assumed that only one sample $\vec y$ is directly observed for a hidden quantity $\vec \theta$ via $p(\vec y | \vec \theta)$. A typical real-world application is as follows. 

\textbf{Sequential Data.} 
Let $i = 1, 2, \ldots$ denote the discrete time index at which the observation $\vec y_i$ is obtained; let the measurement set, up to and including $i$, be $\cal Y_i \defeq \{\vec y_1, \vec y_2, \ldots, \vec y_i\}$, for all $i \ge 1$; let $\vec \theta_i$ denote the hidden quantity at time $i$, which can differ from $\vec \theta_k$ for $k \le i - 1$. At each $i \ge 2$, the posterior distribution can be written as
\begin{equation}\label{eq:sequential-data}
    p(\vec \theta_i | \cal Y_i) \propto p(\vec y_i | \vec \theta_i) \cdot p(\vec \theta_i | \cal Y_{i-1}),
\end{equation}
where $p(\vec \theta_i | \cal Y_{i-1})$ is determined by the Markov transition kernel $p(\vec \theta_{i} | \vec \theta_{i-1})$ through $p(\vec \theta_i | \cal Y_{i-1}) = \int_{\vec \theta_{i-1}} p(\vec \theta_{i} | \vec \theta_{i-1}) p(\vec \theta_{i-1} | \cal Y_{i-1}) \d \vec \theta_{i-1}$. 
That is, the $\cal Y_{i-1}$-conditional distribution $p(\vec \theta_i | \cal Y_{i-1})$ of $\vec \theta_i$ serves as the prior distribution at time $i$, and the one-measurement likelihood function $\vec \theta_i \mapsto p(\vec y_i | \vec \theta_i)$ induces the likelihood distribution $l(\vec \theta_i)$. Note that for each $\vec \theta_i$, there is only one \bfit{direct} observation $\vec y_i$ via $p(\vec y_i | \vec \theta_i)$. For a visual illustration and more details, see Appendix \ref{append:hidden-Markov-process} in supplementary materials. This time-sequential Bayesian inference scheme is particularly useful in real-time signal processing (e.g., Kalman and particle filters) and machine learning (e.g., sequential-data models); see, e.g., \citet{simon2006optimal} and \citet[Chapter~13]{bishop2006pattern}.

Below, we study the multi-sample cases.

\textbf{Block Data.} In non-sequential cases, when we have more than one sample, e.g., $n$ i.i.d. samples $\{\vec y_1, \vec y_2, \ldots, \vec y_n\}$, for an unknown quantity $\vec \theta$, where $\vec y_i \sim p(\vec y | \vec \theta)$, we can straightforwardly generalize \eqref{eq:opt-gen-posterior-dist} to 
\begin{equation}\label{eq:opt-gen-posterior-dist-multi-sample}
\begin{array}{cl}
\displaystyle \min_{q(\vec \theta)} \alpha_3 \Ent q(\vec \theta) 
 + \alpha_1 \KL{q(\vec \theta)}{p(\vec \theta)} +
 \displaystyle \frac{\alpha_2}{n} \sum^{n}_{i = 1} \KL{q(\vec \theta)}{l_{\vec y_i}(\vec \theta)}.
\end{array}
\end{equation}
The solution of \eqref{eq:opt-gen-posterior-dist-multi-sample} is given in the corollary below.
\begin{corollary}[Multi-Sample Uncertainty-Aware Bayes' Rule]
The uncertainty-aware Bayesian posterior solving \eqref{eq:opt-gen-posterior-dist-multi-sample} for $n$ i.i.d. samples is given by
\begin{equation}\label{eq:gen-bayes-rule-multi-sample}
    p_g(\vec \theta | \vec y_1,\ldots,\vec y_n) \propto p^{\frac{\alpha_1}{\alpha_1 + \alpha_2 - \alpha_3}}(\vec \theta) \cdot \left[\prod^{n}_{i = 1} l^{\frac{1}{n}}_{\vec y_i}(\vec \theta)\right]^{\frac{\alpha_2}{\alpha_1 + \alpha_2 - \alpha_3}},
\end{equation}
if $\alpha_1 + \alpha_2 > \alpha_3$, provided that the right-hand-side term is integrable on $\Theta$. When $\alpha_1 + \alpha_2 = \alpha_3$, $p_g(\vec \theta | \vec y_1,\ldots,\vec y_n)$ is an arbitrary distribution supported on the set $\Theta^*$ where
$
\Theta^* \defeq \argmax_{\vec \theta} \alpha_1 \ln p(\vec \theta) + \alpha_2 \cdot \frac{1}{n} \sum^n_{i = 1} \ln l_{\vec y_i}(\vec \theta)
$ 
contains all weighted maximum posterior estimates; if further $\alpha_1 = 0$, it reduces to maximum likelihood estimation.
\stp
\end{corollary}

The $(\alpha, \beta)$-posterior in Definition \ref{def:alpha-beta-posterior} under $n$ i.i.d. samples, induced by \eqref{eq:gen-bayes-rule-multi-sample}, can be rewritten as
\begin{equation}\label{eq:alpha-beta-posterior-multi-sample}
p_g(\vec \theta | \vec y_1,\ldots,\vec y_n) \propto p^{\beta}(\vec \theta) \cdot \prod^{n}_{i = 1} l^{\frac{\alpha}{n}}_{\vec y_i}(\vec \theta),~~~0 \le \alpha, \beta \le \infty.
\end{equation}

Table \ref{tab:examples} can be restated accordingly; we do not repeat here. Note that $\alpha$ may absorb the sample size $n$. Specifically, if we replace $\alpha$ with $n \alpha$, \eqref{eq:alpha-beta-posterior-multi-sample} becomes 
\begin{equation}\label{eq:alpha-beta-posterior-multi-sample-2}
    p_g(\vec \theta | \vec y_1,\ldots,\vec y_n) \propto p^{\beta}(\vec \theta) \cdot \left[ \prod^{n}_{i = 1} l_{\vec y_i}(\vec \theta) \right]^{{\alpha}},~~~0 \le \alpha, \beta \le \infty,
\end{equation}
which reduces to the conventional multi-sample Bayes' rule if $\alpha = \beta = 1$. 

Below, we explain the deficiency of the $\alpha$-posterior under multiple observations, supposing that the true value $\vec \theta_0$ of the unknown quantity is deterministic and fixed. As the sample size $n$ tends to infinity, the likelihood distribution collapses to a delta distribution $\delta_{\vec \theta_{\text{MLE}}}(\vec \theta)$ concentrated at the maximum likelihood estimate (MLE) $\vec \theta_{\text{MLE}}$, leading \eqref{eq:alpha-beta-posterior-multi-sample-2} to
\begin{equation}\label{eq:alpha-beta-posterior-multi-sample-limit-n}
    p_g(\vec \theta | \vec y_1,\ldots,\vec y_\infty) \propto p^{\beta}(\vec \theta) \cdot \left[ \delta_{\vec \theta_{\text{MLE}}}(\vec \theta) \right]^{{\alpha}},~~~0 \le \alpha, \beta \le \infty.
\end{equation}
If $\alpha \ne 0$ and the support set of $p^\beta(\vec \theta)$ contains $\vec \theta_{\text{MLE}}$, the above posterior distribution is again $\delta_{\vec \theta_{\text{MLE}}}(\vec \theta)$. However, when the true likelihood model is not exactly known, the MLE $\vec \theta_{\text{MLE}}$ under the nominal likelihood model does not necessarily coincide with the true value $\vec \theta_0$. Consequently, the posterior $\delta_{\vec \theta_{\text{MLE}}}(\vec \theta)$ loses the uncertainty quantification ability (i.e., posterior's frequentist coverage). Therefore, in the $\alpha$-posterior, a \bfit{small} $\alpha \ge 0$ serves to prevent the likelihood distribution from collapsing into a delta distribution, if treating $\delta_{\vec \theta_{\text{MLE}}}(\vec \theta)$ as $\cal N(\vec \theta_{\text{MLE}}, \epsilon \mat I_d)$ for $\epsilon \to 0$. Nevertheless, by reducing $\alpha$, the prior distribution $p(\vec \theta)$ becomes dominating in calculating the posterior. Hence, another independent parameter to control the contribution of the prior distribution $p(\vec \theta)$ in calculating the posterior is mandated, giving rise to the $\beta$-parameterization of the prior: using a small $\beta$ if $p(\vec \theta)$ is unreliable, e.g., overly concentrated on a wrong point; use a large $\beta$ if $p(\vec \theta)$ is trustworthy, e.g., centered close to the correct point $\vec \theta_0$. 

In summary, considering Examples \ref{ex:alpha-post-deficiency} and \ref{ex:alpha-post-deficiency-2}, and Remark \ref{rem:deficiency-alpha-overconcentration}, the deficiencies of the $\alpha$-posterior are highlighted below.
\begin{remark}[Deficiency of $\alpha$-Posterior]\label{rem:deficiency-alpha-posterior}
The deficiency of the $\alpha$-posterior lies in its limited ability to control the influence of the prior distribution, relative to the likelihood distribution, while simultaneously avoiding posterior collapse. Specifically, when a large $\alpha$ is used, the likelihood is upweighted and the prior is relatively downweighted, but the resulting posterior tends to degenerate to a wrongly concentrated distribution (e.g., a misplaced delta distribution, due to the concentration of the likelihood distribution under large $\alpha$'s; see Remark \ref{rem:deficiency-alpha-overconcentration}), thereby losing its uncertainty quantification capability; recall Example \ref{ex:alpha-post-deficiency}. In contrast, when a small $\alpha$ is used, as in the infinite-sample situation \eqref{eq:alpha-beta-posterior-multi-sample-limit-n}, the likelihood is downweighted and the prior is relatively upweighted, which can be problematic when the prior is unreliable, e.g., in Example \ref{ex:alpha-post-deficiency-2} where the prior is a misplaced delta distribution. Therefore, for both cases, the independent $\beta$-parameterization for the prior is required.
\stp
\end{remark}

\subsection{Uncertainty-Aware Bayes' Rule for Multiple Priors and Samples}
We can further extend \eqref{eq:opt-gen-posterior-dist-multi-sample} when multiple priors are present:
\begin{equation}\label{eq:opt-gen-posterior-dist-multi-sample-multi-prior}
\begin{array}{l}
\displaystyle \min_{q(\vec \theta)} \alpha_3 \Ent q(\vec \theta) 
 + \alpha_1 \left[\sum^m_{i = 1} \beta_i \cdot \KL{q(\vec \theta)}{p_i(\vec \theta)} \right] +
 \displaystyle \alpha_2 \left[ \frac{1}{n} \sum^{n}_{i = 1} \KL{q(\vec \theta)}{l_{\vec y_i}(\vec \theta)} \right],
\end{array}
\end{equation}
where $m$ priors $p_1({\vec \theta}), p_2({\vec \theta}), \ldots, p_m({\vec \theta})$ are available with weights $\beta_1, \beta_2, \ldots, \beta_m$, respectively; $\beta_i \in [0, 1]$ for every $i \in [m]$ and $\sum^m_{i=1} \beta_i = 1$. For a concrete engineering example where multiple priors apply, see Appendix \ref{append:multi-priors} in supplementary materials.

The solution of \eqref{eq:opt-gen-posterior-dist-multi-sample-multi-prior} is given in the corollary below.
\begin{corollary}[Multi-Prior-Multi-Sample Uncertainty-Aware Bayes' Rule]
The uncertainty-aware Bayesian posterior solving \eqref{eq:opt-gen-posterior-dist-multi-sample-multi-prior} for $m$ priors and $n$ i.i.d. samples is given by
\begin{equation}\label{eq:gen-bayes-rule-multi-prior}
\begin{array}{l}
    p_g(\vec \theta | \vec y_1,\ldots,\vec y_n) \propto 
    \displaystyle \left[\prod^{m}_{i = 1} p^{\beta_i}_i(\vec \theta) \right]^{\frac{\alpha_1}{\alpha_1 + \alpha_2 - \alpha_3}} \cdot \left[\prod^{n}_{i = 1} l^{\frac{1}{n}}_{\vec y_i}(\vec \theta)\right]^{\frac{\alpha_2}{\alpha_1 + \alpha_2 - \alpha_3}},
\end{array}
\end{equation}
if $\alpha_1 + \alpha_2 > \alpha_3$, provided that the right-hand-side term is integrable on $\Theta$. When $\alpha_1 + \alpha_2 = \alpha_3$, $p_g(\vec \theta | \vec y_1,\ldots,\vec y_n)$ is an arbitrary distribution supported on the set $\Theta^*$ where
$
\Theta^* \defeq \argmax_{\vec \theta} \alpha_1 \sum^m_{i = 1} \beta_i \cdot \ln p_i(\vec \theta) + \alpha_2 \sum^n_{i = 1} \frac{1}{n} \cdot \ln l_{\vec y_i}(\vec \theta)
$ 
contains all weighted maximum posterior estimates.
\stp
\end{corollary}

The $(\alpha, \beta)$-posterior in Definition \ref{def:alpha-beta-posterior} under $m$ priors and $n$ i.i.d. samples, induced by \eqref{eq:gen-bayes-rule-multi-prior}, can be rewritten as
\begin{equation}\label{eq:alpha-beta-posterior-multi-prior}
\begin{array}{ll}
p_g(\vec \theta | \vec y_1,\ldots,\vec y_n) \propto \displaystyle \prod^{m}_{i = 1} p^{\beta \cdot \beta_i}_i(\vec \theta)\cdot \prod^{n}_{i = 1} l^{\frac{\alpha}{n}}_{\vec y_i}(\vec \theta),& 0 \le \alpha, \beta \le \infty,
\end{array}
\end{equation}
where $\alpha$ can depend on $n$. 
Table \ref{tab:examples} can be restated accordingly; we do not repeat here.

\subsection{Examples of Application}\label{subsec:applications}
The uncertainty-aware Bayes' rule \eqref{eq:alpha-beta-posterior}, i.e., the $(\alpha,\beta)$-posterior, has several potential applications in statistical machine learning and statistical signal processing. We specifically discuss classification and estimation problems. Suppose that the true joint distribution is $p_0(\vec y, \vec \theta)$, the true prior distribution is $p_0(\vec \theta)$, and the true likelihood distribution is $l_{0, \vec y}(\vec \theta)$ under the measurement $\vec y$. For both classification and estimation tasks, from the perspective of the Bayesian decision theory, we aim to minimize the following population risk \citep[Sections~5.1~and~5.4]{murphy2022probabilistic}
\begin{equation}\label{eq:cost-minimization}
    \min_{\vech \theta(\cdot)} \int_{\Theta} \int_{\cal Y}  c[\vech \theta(\vec y), \vec \theta] p_0(\vec y, \vec \theta) \d \vec y \d \vec \theta,
\end{equation}
where $\vech \theta: \cal Y \to \Theta$ is a classifier, if $\Theta$ is discrete, or an estimator, if $\Theta$ is continuous, of the ground truth $\vec \theta$, and $c: \Theta \times \Theta \to \R_+$ is the cost function; see also \citet[Chapter~1]{vapnik2000nature} and \citet[Section~1.5]{bishop2006pattern} for details on statistical learning. For classification problems, $c$ can typically be the $0/1$ loss function (i.e., indicator loss function), under which the optimal classifier is the MAP classifier
\begin{equation}\label{eq:MAP-classifier}
    \vech \theta_0(\vec y) \defeq \max_{\vec \theta \in \Theta} p_0(\vec \theta | \vec y).
\end{equation} 
For estimation problems, $c$ can typically be the quadratic loss function (i.e., mean-squared error function), under which the optimal estimator is the posterior mean 
\begin{equation}\label{eq:MMSE-estimator}
    \vech \theta_0(\vec y) \defeq \int_{\Theta} \vec \theta p_0(\vec \theta | \vec y) \d \vec \theta,
\end{equation} 
which is also known as the minimum mean-squared error (MMSE) estimator.  As we can see from \eqref{eq:cost-minimization}, \eqref{eq:MAP-classifier}, and \eqref{eq:MMSE-estimator}, the true joint distribution $p_0(\vec y, \vec \theta)$ plays a central role, as the corresponding true posterior distribution $p_0(\vec \theta | \vec y)$ does. 

In practice, when $p_0(\vec y, \vec \theta)$ is not accessible, we need a high-quality surrogate $p(\vec y, \vec \theta)$; recall \eqref{eq:bayesian-cost} and \eqref{eq:bayesian-cost-W}. Let the nominal prior distribution and the nominal likelihood distribution be $p(\vec \theta)$ and $l_{\vec y}(\vec \theta)$, respectively. 

\textbf{Specification of Robustness}. 
As per the high-level concept in Philosophy \ref{phi:robustness}, robustness can be specifically defined for Bayesian statistical learning \eqref{eq:cost-minimization}. Let $\vech \theta_a(\cdot)$ solve \eqref{eq:cost-minimization} under the assumed nominal distribution $p(\vec y, \vec \theta)$, and $\vech \theta_g(\cdot)$ solve \eqref{eq:cost-minimization} under the generalized distribution $p_g(\vec y, \vec \theta)$. Then, $\vech \theta_g(\cdot)$ is a robust solution if
\begin{equation}\label{eq:bayes-decision-robustness}
    \int_{\Theta} \int_{\cal Y}  c[\vech \theta_g(\vec y), \vec \theta] p_0(\vec y, \vec \theta) \d \vec y \d \vec \theta < \int_{\Theta} \int_{\cal Y}  c[\vech \theta_a(\vec y), \vec \theta] p_0(\vec y, \vec \theta) \d \vec y \d \vec \theta.
\end{equation}
That is, if testing under the true distribution $p_0(\vec y, \vec \theta)$, the uncertainty-aware solution $\vech \theta_g(\cdot)$ would lead to a smaller testing error than the nominal solution $\vech \theta_a(\cdot)$. We adopt this definition of robustness because the left-hand-side term in \eqref{eq:bayes-decision-robustness} is closer to the optimal risk \eqref{eq:cost-minimization} than the right-hand-side term. Hence, to find a robust solution, using a better surrogate for $p_0(\vec y, \vec \theta)$ is crucial, as shown in \eqref{eq:bayesian-cost}, \eqref{eq:bayesian-cost-posterior}, \eqref{eq:bayesian-cost-W}, and \eqref{eq:bayesian-cost-W-conditional}.

Below, we specifically discuss uncertainty-aware Bayesian MAP classification and Bayesian MMSE estimation.

\subsubsection{Bayesian MAP Classification}\label{subsec:bayes-classification}
Classification is a fundamental component of statistical machine learning. Let $\Theta \defeq \{1,2,3,\ldots,r\}$ be a categorical set where the integer $r$ is finite. The Bayes MAP classifier based on the conventional Bayes' rule is
$
    \hat \theta_0 (\vec y) \defeq \argmax_{\theta \in \Theta} \log p_0(\theta) + \log l_{0, \vec y} (\theta);
$ 
see \eqref{eq:MAP-classifier}. 
In practice, however, the true distributions $p_0(\theta)$ and $l_{0, \vec y}(\theta)$ are unknown. By employing the nominal distributions $p(\theta)$ and $l_{\vec y}(\theta)$, the uncertainty-aware Bayes MAP classifier based on the $(\alpha,\beta)$-posterior \eqref{eq:alpha-beta-posterior} can be obtained as
\begin{equation}\label{eq:generalized-bayes-classifier}
    \hat \theta_g (\vec y) \defeq \argmax_{\theta \in \Theta}  \beta \log p(\theta) + \alpha \log l_{\vec y} (\theta),
\end{equation} 
which is equivalent, in the sense of the same MAP classifier, to
\begin{equation}\label{eq:generalized-bayes-classifier-elite}
    \hat \theta_g (\vec y) \defeq \argmax_{\theta \in \Theta} (1-\lambda) \log p(\theta) + \lambda \log l_{\vec y} (\theta),
\end{equation} 
for $\lambda \defeq \alpha/(\alpha + \beta)$. As we can see, for Bayesian MAP classification problems, the $(\alpha,\beta)$-posterior with $\alpha, \beta \in [0, \infty)$ has the same effect as the $(\lambda, 1-\lambda)$-posterior with $\lambda \in [0, 1]$; that is, only the ratio $\alpha/\beta$ matters. The same effect can also be achieved by the $\frac{\alpha}{\beta}$-posterior in \eqref{eq:alpha-posterior-intro}. However, when the classification method is not a Bayes MAP classifier (e.g., when the cost function is non-$0/1$), such equivalence between the $(\alpha,\beta)$-posterior and the $\alpha$-posterior cannot be guaranteed. Hence, studying the $(\alpha,\beta)$-posterior is still necessary for generic classification problems under generic cost functions $c$.

\subsubsection{Bayesian MMSE Estimation}\label{subsec:bayes-estimation}
Estimation aims to infer an unknown quantity $\vec \theta$ based on observed data $\vec y$. It is of particular importance in statistical signal processing. Upon collecting $\vec y$, the Bayes estimator in the MMSE sense based on the conventional Bayes' rule is
$
    \vech \theta_0 (\vec y) \defeq \int \vec \theta p_0(\vec \theta | \vec y)  \d \vec \theta,
$ 
where $p_0(\vec \theta | \vec y) \propto p_0(\vec \theta, \vec y)$ is the true Bayesian posterior distribution; see \eqref{eq:MMSE-estimator}. In practice, however, the true distributions $p_0(\vec \theta)$ and $l_{0, \vec y}(\vec \theta)$ are unknown. By employing the nominal distributions $p(\vec \theta)$ and $l_{\vec y}(\vec \theta)$, the uncertainty-aware Bayes MMSE estimator based on the $(\alpha, \beta)$-posterior \eqref{eq:alpha-beta-posterior} can be obtained as
\begin{equation}\label{eq:generalized-bayes-estimator}
    \vech \theta_g (\vec y) \defeq \int_{\Theta} \vec \theta p_{g}(\vec \theta | \vec y)  \d \vec \theta.
\end{equation}

\subsubsection{Other Examples}
In this subsection, we discuss other applications of the $(\alpha, \beta)$-posterior. First, we give an example of application in statistical signal processing.
\begin{example}[Uncertainty-Aware Particle Filter]\label{ex:UA-PF}
Suppose that $\Theta$ is a finite set containing $r$ points (each point is known as a particle): i.e., $\Theta \defeq \{\vec \theta_1, \vec \theta_2, \ldots, \vec \theta_r\}$. For each particle $\vec \theta_i$, $i \in [r]$, the likelihood evaluated at the measurement $\vec y$ is assumed to be $p(\bm y|\vec \theta_i)$. Further, we suppose that the prior distribution is $\vec p \defeq [p_1,p_2,\ldots,p_r]$ and the $\vec y$-likelihood distribution (recall Definition \ref{def:likelihood-dist}) is $\vec l \defeq [l_1,l_2,\ldots,l_r]$. Then the uncertainty-aware Bayesian posterior, i.e., the $(\alpha,\beta)$-posterior, of particles equals 
$
p_g(\vec \theta_i | \vec y) \propto p^\beta_i \cdot l^\alpha_i
$, for every $i \in [r]$. 
The above formula leads to the uncertainty-aware particle filter for dynamic stochastic nonlinear systems.
\stp
\end{example}

Second, we give an example of application in statistical machine learning. It is shown that the uncertainty-aware MAP estimation can lead to the popular ridge regression.
\begin{example}[Uncertainty-Aware MAP Estimation]\label{ex:Ridge}
We consider a nonlinear regression model $\rscl y = f(\rvec x; \vec \theta) + \rscl v$ where $\rscl y \in \R$ is the response, $\rvec x \in \R^l$ is the feature vector, $\rscl v \sim \cal N(0,1)$ is the regression error, and $\vec \theta$ is the parameter vector. Supposing the prior distribution of $\vec \theta$ is $\cal N(\vec 0, \mat I_d)$, upon the collection of the data $(y, \vec x)$, the MAP estimator of $\vec \theta$ is 
$
\min_{\vec \theta \in \Theta} [y - f(\vec x; \vec \theta)]^2 + \vec \theta^\top \vec \theta.
$ 
By using the $(\alpha,\beta)$-posterior rule, the uncertainty-aware MAP estimator of $\vec \theta$ can be written as 
$
\min_{\vec \theta \in \Theta} {\alpha} [y - f(\vec x; \vec \theta)]^2 + {\beta} \vec \theta^\top \vec \theta.
$ 
By letting $\lambda \defeq \beta/\alpha$, we have the $\lambda$-ridge regression
$
\min_{\vec \theta \in \Theta} [y - f(\vec x; \vec \theta)]^2 + \lambda \vec \theta^\top \vec \theta.
$ 
Therefore, the popular ridge regression method in statistical machine learning can be interpreted as an uncertainty-aware MAP estimation.
\stp
\end{example}

Note that the MAP estimator (which is optimal under the $0/1$ loss) is not necessarily optimal under the quadratic loss, in the sense of \eqref{eq:cost-minimization}, but the MMSE estimator \eqref{eq:MMSE-estimator} is. When both the true prior $p_0(\vec \theta)$ and the true data-generating law $p_0[(\vec x, y) | \vec \theta]$ are Gaussian, the MAP and MMSE estimators coincide; for non-Gaussian distributions, the two usually differ. Third, we discuss an example of application that is popular in both statistical signal processing and statistical machine learning.
\begin{example}[Uncertainty-Aware Bayesian Model Averaging]\label{ex:Bayesian-model-averaging}
Suppose that we have $r$ models indexed by $\{\theta_1, \theta_2, \ldots, \theta_r\}$ to account for a signal processing or a machine learning problem. Let the prior distribution of models be $\vec p \defeq [p_1,p_2,\ldots,p_r]$ and the likelihood distribution (recall Definition \ref{def:likelihood-dist}) be $\vec l \defeq [l_1,l_2,\ldots,l_r]$. Then the uncertainty-aware Bayesian posterior, i.e., the $(\alpha,\beta)$-posterior, of models equals 
$
p_g(\theta_i | \vec y) \propto p^\beta_i \cdot l^\alpha_i,
$ 
for every $i \in [r]$. 
The above formula induces the uncertainty-aware Bayesian model averaging method.
\stp
\end{example}

\subsection{Parameter Tuning}\label{subsec:parameter-tuning}
This subsection discusses the tuning methods for the parameters $(\alpha, \beta)$ in practice. The main purpose is to find some $(\alpha, \beta)$ such that the $(\alpha, \beta)$-posterior can outperform the conventional Bayesian posterior and the $\alpha$-posterior; recall Philosophy \ref{phi:robustness}. 

Let $\vec \omega \defeq [\alpha, \beta]^\top$ denote the parameter vector of an $(\alpha,\beta)$-posterior and $f(\vec \omega)$ the loss function of using the $(\alpha,\beta)$-posterior for a specific application, where $f: \R^2_{+} \to \R_+$.  We assume that $\vec \omega$ takes values on $[0, \tau]^2$ where $\tau > 1$ is a bounded positive real number (for example, $\tau = 5$). The optimal parameter tuning for $\vec \omega$ can be formulated as a minimization problem
\begin{equation}\label{eq:parameter-tuning}
    \min_{\vec \omega \in [0, \tau]^2} f(\vec \omega).
\end{equation}
Examples of $f(\vec \omega)$ are given below, under the Bayesian decision framework \eqref{eq:cost-minimization}.

\begin{example}[MAP Classification]\label{ex:classification}
For classification problems, the loss function is the misclassification probability, i.e.,
\begin{equation}\label{eq:bayes-classifier-cost}
    f(\vec \omega) \defeq \int \int \bb I_{\{\theta \neq \hat \theta_{g, \vec \omega} (\vec y)\}} \cdot p_0(\vec y, \vec \theta) \d \vec y \d \vec \theta,
\end{equation}
where $\bb I_{\{\cdot\}}$ denotes the indicator function (i.e., $0/1$ loss) and the uncertainty-aware MAP classifier $\hat \theta_{g, \vec \omega} (\vec y)$ is defined in \eqref{eq:generalized-bayes-classifier}. Note that $\hat \theta_{g} (\vec y)$ in \eqref{eq:generalized-bayes-classifier} depends on $\vec \omega$. If the classification method in \eqref{eq:generalized-bayes-classifier-elite} is employed, we have the cost function 
\[
f(\lambda) \defeq \int \int \bb I_{\{\theta \neq \hat \theta_{g, \lambda} (\vec y)\}} \cdot p_0(\vec y, \vec \theta) \d \vec y \d \vec \theta,~\lambda \in [0, 1],
\] 
since the MAP classifier defined in \eqref{eq:generalized-bayes-classifier-elite} is parameterized by $\lambda$. \stp
\end{example}

\begin{example}[MMSE Estimation]\label{ex:estimation}
For estimation problems, the loss function in the mean-squared error sense can be written as
\begin{equation}\label{eq:bayes-estimation-cost}
    f(\vec \omega) \defeq \int \int [\vec \theta - \vech \theta_{g, \vec \omega} (\vec y)]^\top [\vec \theta - \vech \theta_{g, \vec \omega} (\vec y)] \cdot p_0(\vec y, \vec \theta) \d \vec y \d \vec \theta,
\end{equation}
where the uncertainty-aware MMSE estimator $\vech \theta_{g, \vec \omega} (\vec y)$ is defined in \eqref{eq:generalized-bayes-estimator}. Note that $\vech \theta_{g} (\vec y)$ in \eqref{eq:generalized-bayes-estimator} depends on $\vec \omega$.
\stp
\end{example}

\textbf{Data-Driven Case}. In practice, however, the loss function $f(\vec \omega)$ is unknown due to the unavailability of $p_0(\vec y, \vec \theta)$. Using historical data $\cal D_t \defeq \{(\vec \theta_1, \vec y_1), (\vec \theta_2, \vec y_2), \ldots, (\vec \theta_t, \vec y_t)\}$, we can use data-driven estimate $\hat f(\vec \omega)$ as a surrogate of $f(\vec \omega)$. For instance, 
\begin{equation}\label{eq:bayes-classifier-cost-empirical}
\hat f(\vec \omega) = \frac{1}{t} \sum^t_{i = 1} \bb I_{\{\theta_i \neq \hat \theta_{g, \vec \omega} (\vec y_i)\}}
\end{equation} 
in Example \ref{ex:classification} and 
\begin{equation}\label{eq:bayes-estimation-cost-empirical}
\hat f(\vec \omega) = \frac{1}{t} \sum^t_{i = 1} [\vec \theta_i - \vech \theta_{g, \vec \omega} (\vec y_i)]^\top [\vec \theta_i - \vech \theta_{g, \vec \omega} (\vec y_i)]
\end{equation} 
in Example \ref{ex:estimation}. When the training data set $\cal D_t$ is sufficiently large, $\hat f(\vec \omega)$ can be a point-wisely good estimate of $f(\vec \omega)$. For real-world statistical machine learning and signal processing tasks, using data-driven estimates \eqref{eq:bayes-classifier-cost-empirical} and \eqref{eq:bayes-estimation-cost-empirical} for \eqref{eq:parameter-tuning} is standard.

In the following, we suggest two empirical methods for tuning $\vec \omega$. Since \eqref{eq:parameter-tuning} is a low-dimensional optimization problem with at most two variables on a hyper-cube domain [cf. \eqref{eq:generalized-bayes-classifier-elite} and \eqref{eq:generalized-bayes-estimator}], both methods can be statistically and computationally efficient; for empirical validation, see experiments in Subsection \ref{sec:applications}. 

\subsubsection{Grid Search}
For simplicity in operation, we can generate a two-dimensional uniform grid $\Omega_{\text{grid}}$ on $[0, \tau]^2$ to obtain discrete empirical evaluations $\hat f(\vec \omega)$ for $\vec \omega \in \Omega_{\text{grid}} \subset [0, \tau]^2$; note that $(1, 1)$ should be included in $\Omega_{\text{grid}}$. Then, we solve 
$
    (\alpha^*, \beta^*) = \argmin_{\vec \omega \in \Omega_{\text{grid}}} \hat f(\vec \omega)
$ 
to obtain the best parameter pair $(\alpha^*, \beta^*)$. Since $(1, 1) \in \Omega_{\text{grid}}$, we have $\hat f([\alpha^*, \beta^*]^\top) \le \hat f([1, 1]^\top)$. Namely, the loss under the proposed $(\alpha, \beta)$-posterior is no larger than that under the usual Bayesian posterior. In addition, if $(\alpha, 1) \in \Omega_{\text{grid}}$, we have $\hat f([\alpha^*, \beta^*]^\top) \le \hat f([\alpha, 1]^\top)$. Therefore, the loss under the proposed $(\alpha, \beta)$-posterior is no larger than that under the existing $\alpha$-posterior. These claims can empirically but sufficiently validate the advantage of the $(\alpha, \beta)$-posterior over the conventional Bayesian posterior and the $\alpha$-posterior. Note that the popular SafeBayes algorithm to tune $\alpha$ in the $\alpha$-posterior is also operationally and essentially a grid search method, but on $[0, 1]$ instead of $[0, \tau]$ \citep[Algorithm~1]{grunwald2012safe}. 


\subsubsection{Surrogate Optimization}
The grid search method cannot pursue the global optimality on $[0, \tau]^2$ due to the discretization. Considering the optimization problem \eqref{eq:parameter-tuning} and its characteristics (i.e., $f$ is unknown but some noisy evaluations $\hat f$ are available), we can leverage the surrogate optimization frameworks, for example, radial-basis-function surrogate optimization \citep{shen2020global}, Gaussian-process surrogate optimization (i.e., Bayesian optimization) \citep{wang2023recent,garnett2023bayesian}.
To put it simply, surrogate optimization uses as few as possible noisy evaluations $\hat f(\vec \omega)$ to learn the true function $f(\vec \omega)$ in an online manner, and simultaneously search for the globally optimal minimizer(s) $\vec \omega^*$ such that $f(\vec \omega^*) \le f(\vec \omega)$ for all $\vec \omega$ on $[0, \tau]^2$.

\subsection{Concrete Applications and Experiments}\label{sec:applications}
This subsection presents experimental results to show the necessity and usefulness of the $(\alpha, \beta)$-posterior, compared to the existing $\alpha$-posterior and Bayesian posterior. All the source data and codes are available online at GitHub: \url{https://github.com/Spratm-Asleaf/Bayes-Rule}. 
The experimental results are obtained by a Lenovo laptop with 16G RAM and 11th Gen Intel(R) Core(TM) i5-11300H CPU @ 3.10GHz.

\subsubsection{Illustrating Examples: Hidden-Quantity Estimation}\label{subsec:experiment-hidden-quantity}
We study a hidden-quantity estimation problem subject to a linear data-generating process, under the deterministic and stochastic interpretations of Bayes' rule, respectively corresponding to the estimation of a constant and a random variable. This is a simplified but motivating example of a large body of real-world statistical signal processing problems, e.g., sensor localization, target tracking, and state estimation (also known as state filtering). Well-established algorithms in this domain include the Kalman and particle filters, both of which fall under the framework of sequential Bayesian inference \eqref{eq:sequential-data}; see \citet{simon2006optimal}; \citet[Chapter~13]{bishop2006pattern}. Experiments on real-world state estimation and target tracking problems are presented later in this subsection. Note that in such signal-processing applications, only one observation $y$ is available for statistical tasks; recall \eqref{eq:sequential-data}.

\textbf{Deterministic Interpretation.} We consider a data-generating process 
\[
\rscl y = \theta_0 + \rscl v
\] 
where $\theta_0$ denotes the fixed hidden quantity to be estimated using the observable quantity $\rscl y$, and $\rscl v$ the noise that obeys $\cal N(v; v_0, r_0)$. For the experimental purpose, we let the true values be $\theta_0 = 1$, $v_0 = 0.1$, and $r_0 = 1.1$. Hence, the true noise distribution is $\cal N(v; 0.1, 1.1)$ and the true likelihood function is $\theta \mapsto \cal N(y; \theta + 0.1, 1.1)$. Since the true prior distribution is $\delta_{1}(\theta)$, treated as $\cal N(\theta; 1, \epsilon)$ for $\epsilon \to 0$, the true posterior distribution is $\delta_{1}(\theta)$ as well.

Under model misspecifications, the nominal noise distribution is assumed to be $\cal N(v; 0, 1)$, so the nominal likelihood function is $\theta \mapsto \cal N(y; \theta, 1)$. In addition, we suppose the nominal prior distribution (i.e., normalized subjective belief of $\theta_0$) is $\cal N(\theta; 0.8, 0.1)$.

At the measurement $y = 1.2$, the nominal likelihood distribution is $\cal N(\theta; 1.2, 1.0)$. By incorporating the nominal prior distribution $\cal N(\theta; 0.8, 0.1)$, we have
\begin{itemize}
    \item \textit{Bayesian Posterior}: $\cal N \left(\theta; 0.84, 0.09 \right)$, which significantly deviates from the true posterior $\delta_{1}(\theta)$.
    
    \item \textit{The $\alpha$-Posterior}: $\cal N \left(\theta; \frac{1.2 \alpha + 8}{\alpha + 10}, \frac{1}{\alpha + 10} \right)$. To match the true posterior mean, $\alpha$ must equal $10$, which cannot be obtained by SafeBayes because SafeBayes requires that $\alpha \in [0, 1]$ \citep{grunwald2012safe,grunwald2017inconsistency}. Consequently, the posterior variance is $\frac{1}{20}$. Note that in existing literature, $\alpha \in [0, 1]$ is largely suggested for some favored theoretical properties under specific settings \citep{miller2019robust,medina2022robustness}. However, this paper echoes \citet{wu2023comparison} that $\alpha \ge 1$ should also be actively considered in practice; recall Example \ref{ex:benefits-alpha-distributions}.
    
    \item \textit{The $(\alpha, \beta)$-Posterior}: $\cal N \left(\theta; \frac{1.2 \alpha + 8 \beta}{\alpha + 10 \beta}, \frac{1}{\alpha + 10 \beta} \right)$. To math the true posterior mean, we must have $\alpha = 10 \beta$. Consequently, the posterior variance is $\frac{1}{20\beta}$. For a sufficiently large $\beta$, the $(\alpha, \beta)$-posterior can match the true posterior $\delta_{1}(\theta)$ well.
\end{itemize}
As indicated, in this studied example, both the $(\alpha, \beta)$-posterior and the $\alpha$-posterior can well match the true posterior mean if the involved parameters can be satisfactorily tuned. However, only the proposed $(\alpha, \beta)$-posterior can simultaneously well match the posterior variance, while the existing $\alpha$-posterior (even under best tuning) and Bayesian posterior cannot.

\textbf{Stochastic Interpretation.} We consider a data-generating process 
\[
\rscl y = \rscl \theta + \rscl v
\] 
where $\rscl \theta$ denotes the hidden unobservable quantity to be estimated using the observable quantity $\rscl y$, and $\rscl v$ the noise. The true prior distribution of $\rscl \theta$, from which $\rscl \theta$ is physically sampled, is $\cal N(\theta; \theta_0, q_0)$. The true noise distribution of $\rscl v$ is $\cal N(v; v_0, r_0)$. For the experimental purpose, we let the true values be $\theta_0 = 1.0$, $q_0 = 0.5$, $v_0 = 0.1$, and $r_0 = 1.1$. Hence, the true noise distribution is $\cal N(v; 0.1, 1.1)$, the true likelihood function is $\theta \mapsto \cal N(y; \theta + 0.1, 1.1)$, and the true prior distribution is $\cal N(\theta; 1.0, 0.5)$.

Under model misspecifications, the nominal noise distribution is assumed to be $\cal N(v; 0, 1)$, so the nominal likelihood function is $\theta \mapsto \cal N(y; \theta, 1)$. In addition, we suppose the nominal prior distribution is $\cal N(\theta; 1.0, 1.0)$.

At the measurement $y = 1.2$, the true likelihood distribution is $\cal N(\theta; 1.1, 1.1)$. By incorporating the true prior distribution $\cal N(\theta; 1.0, 0.5)$, the true posterior distribution is $\cal N(\theta; 1.03, 0.34)$. In addition, the nominal likelihood distribution is $\cal N(\theta; 1.2, 1.0)$. By incorporating the nominal prior distribution $\cal N(\theta; 1.0, 1.0)$, we have
\begin{itemize}
    \item \textit{Bayesian Posterior}: $\cal N \left(\theta; 1.1, 0.5 \right)$, which deviates from the true posterior distribution $\cal N(\theta; 1.03, 0.34)$ in both mean and variance.
    
    \item \textit{The $\alpha$-Posterior}: $\cal N \left(\theta; \frac{1.2 \alpha + 1}{\alpha + 1}, \frac{1}{\alpha + 1} \right)$. To exactly match the true posterior mean, $\alpha$ must equal $0.18$. As a result, the posterior variance is $0.85$, which is significantly different from the ground truth $0.34$. Hence, the $\alpha$-posterior cannot exactly match the true posterior distribution $\cal N(\theta; 1.03, 0.34)$ for any $\alpha \ge 0$.
    
    \item \textit{The $(\alpha, \beta)$-Posterior}: $\cal N \left(\theta; \frac{1.2 \alpha + \beta}{\alpha + \beta}, \frac{1}{\alpha + \beta} \right)$. To exactly match the true posterior mean, we have $\alpha = 0.18 \beta$. To further exactly match the posterior variance, we can let $\beta = 2.49$. Hence, the $(\alpha, \beta)$-posterior can exactly match the true posterior distribution $\cal N(\theta; 1.03, 0.34)$ for $\alpha = 0.45$ and $\beta = 2.49$.
\end{itemize}
As indicated, in this studied example, only the proposed $(\alpha, \beta)$-posterior can perfectly match the true posterior distribution if the involved parameters can be well tuned, while the existing $\alpha$-posterior (even under best tuning) and Bayesian posterior cannot.

\subsubsection{Illustrating Example: Bayesian Linear Regression With Fixed Design}\label{subsebsec:linear-regression}
In this subsection, we study a Bayesian linear regression problem, where $n$ i.i.d. samples $\{(\vec x_i, y_i)\}_{i \in [n]}$ are collected; $\{\vec x_i\}$ denote the feature vectors and $\{y_i\}$ the responses. Suppose the true data-generating process is
\begin{equation}\label{eq:true-data-generating}
    \rscl y_i = \rvec x^\top_i \vec \theta + \rscl w_i + \rscl v_i,~~~i = 1, 2, \ldots, n,
\end{equation}
and the nominal (i.e., assumed) data-generating model is
\begin{equation}\label{eq:nominal-data-generating}
    \rscl y_i = \rvec x^\top_i \vec \theta + \rscl v_i,~~~i = 1, 2, \ldots, n,
\end{equation}
where the true value of $\vec \theta \in \R^d$ is $\vec \theta_0$; $\rvec x_i \sim \cal N(\vec 0, \mat I_d)$, for every $i \in [n]$; $\rscl w_i \sim \cal N(\mu_w, \sigma^2_w)$ denotes an \bfit{unmodeled} term and $\rscl v_i \sim \cal N(0, \sigma^2)$ denotes the noise term. We assume that the true values of $\mu_w$, $\sigma^2_w$, and $\sigma^2$ are exactly known, and the quantity to estimate is only $\vec \theta \in \R^d$; note that we treat $\{\vec x_i\}_{i \in [n]}$ as fixed in the statistical inference of $\vec \theta$.

\textit{True Prior and Posterior}: Both the true prior and true posterior are $\delta_{\vec \theta_0}(\vec \theta)$, treated as $\cal N(\vec \theta_0, \epsilon \mat I_d)$ for $\epsilon \to 0$.

\textit{Nominal Prior Distribution}: We employ a Gaussian prior $\vec \theta \sim \cal N(\vec \mu_0, \mat \Sigma_0)$.

\textit{Nominal Likelihood Function}: Let $\mat X \defeq [\vec x^\top_1; \vec x^\top_2; \ldots; \vec x^\top_n]$ and $\vec y \defeq [y_1; y_2; \ldots; y_n]$;\footnote{MATLAB's stacking notation for matrix and vector is used in this subsection.} i.e., $\mat X \in \R^{n \times d}$ and $\vec y \in \R^n$. The likelihood function is $p(\vec y \mid \vec \theta) = \cal N(\vec y - \mat X \vec \theta, \sigma^2 \mat I_n)$.


\textit{Bayesian Posterior}: The conventional Bayesian posterior is $\cal N(\vec \mu_n, \mat \Sigma_n)$, where 
\begin{equation}\label{eq:BLR-Bayesian}
\mat \Sigma_n = \left(\mat \Sigma_0^{-1}+\frac{1}{\sigma^2} \mat X^{\top} \mat X\right)^{-1},~~~ \vec \mu_n  = \mat \Sigma_n \cdot \left(\mat \Sigma_0^{-1} \vec \mu_0 + \frac{1}{\sigma^2} \mat X^{\top} \vec y\right).
\end{equation}

\textit{The $\alpha$-Posterior}: The $\alpha$-posterior is $\cal N(\vec \mu_{n, \alpha}, \mat \Sigma_{n, \alpha})$, where 
\begin{equation}\label{eq:BLR-alpha}
\mat \Sigma_{n, \alpha} = \left(\mat \Sigma_0^{-1}+\frac{\alpha}{\sigma^2} \mat X^{\top} \mat X\right)^{-1},~~~ \vec \mu_{n, \alpha}  = \mat \Sigma_{n, \alpha} \cdot \left(\mat \Sigma_0^{-1} \vec \mu_0 + \frac{\alpha}{\sigma^2} \mat X^{\top} \vec y\right).
\end{equation}

\textit{The $(\alpha, \beta)$-Posterior}: The $(\alpha, \beta)$-posterior is $\cal N(\vec \mu_{n, \alpha, \beta}, \mat \Sigma_{n, \alpha, \beta})$, where 
\begin{equation}\label{eq:BLR-alpha-beta}
\mat \Sigma_{n, \alpha, \beta} = \left(\beta \mat \Sigma_0^{-1}+\frac{\alpha}{\sigma^2} \mat X^{\top} \mat X\right)^{-1},~~~ \vec \mu_{n, \alpha, \beta}  = \mat \Sigma_{n, \alpha, \beta} \cdot \left(\beta \mat \Sigma_0^{-1} \vec \mu_0 + \frac{\alpha}{\sigma^2} \mat X^{\top} \vec y\right).
\end{equation}

\textit{Theoretical Insights}: A posterior mean is a weighted sum of the prior mean $\vec \mu_0$ and the MLE $(\mat X^\top \mat X)^{-1} \mat X^{\top} \vec y$, where the weight matrices are determined by their corresponding covariances $\mat \Sigma_0$ and $\sigma^2 \cdot (\mat X^\top \mat X)^{-1}$. In the following, we discuss two specific cases to show the advantages of the $(\alpha, \beta)$-posterior: the key message is that only the $(\alpha, \beta)$-posterior allows independent tuning of the posterior mean and covariance, while the Bayesian posterior and the $\alpha$-posterior do not. 
\begin{itemize}[label=\scriptsize$\bullet$]
    \item Suppose that the prior mean $\vec \mu_0$ exactly equals the true regression parameter $\vec \theta_0$. Meanwhile, assume that the nominal data-generating model \eqref{eq:nominal-data-generating} is highly unreliable due to the unmodeled term $\rscl w_i$, relative to the true process \eqref{eq:true-data-generating}, so that the MLE $(\mat X^\top \mat X)^{-1}\mat X^\top \vec y$ is extremely uninformative. In this case, to center the posterior at $\vec \theta_0$ requires suppressing the influence of the likelihood. For the $\alpha$-posterior, the only way to achieve this is to set $\alpha \defeq 0$, which removes the MLE contribution entirely. However, the resulting posterior covariance coincides with the prior covariance $\mat \Sigma_0$, which is fixed and cannot be further adjusted to reflect increased confidence in the correctness of the posterior mean. In contrast, under the $(\alpha,\beta)$-posterior, setting $\alpha \defeq 0$ again lets the posterior mean exactly match $\vec \theta_0$, but the posterior covariance becomes $\mat \Sigma_0/\beta$. By tuning $\beta$, the uncertainty around the posterior mean can be continuously controlled, and in the limit $\beta \to \infty$, the posterior covariance shrinks to $\mat 0$. This demonstrates a key advantage of the $(\alpha,\beta)$-posterior: it allows one to encode absolute confidence in a correct prior mean without being constrained by the fixed scale of the prior covariance. Note that, in this experiment, the true prior and true posterior are $\delta_{\vec \theta_0}(\vec \theta)$, a zero-covariance distribution.

    \item Suppose that the prior knowledge is strongly misleading, in the sense that the prior mean $\vec \mu_0$ deviates significantly from the true parameter $\vec \theta_0$, while the prior covariance $\mat \Sigma_0$ is small and thus overconfident. In this case, it is desirable for the posterior mean to be dominated by the data and centered at the MLE. For the $\alpha$-posterior, to completely remove the influence of the prior, we must take $\alpha \defeq \infty$. However, this operation also forces the posterior covariance to collapse to zero, implying absolute certainty in the MLE and eliminating any meaningful uncertainty quantification. In contrast, by setting $\beta \defeq 0$, the contribution of the prior is entirely removed regardless of how small $\mat \Sigma_0$ is, and the posterior mean coincides with the MLE. At the same time, the posterior covariance remains nonzero and equals $\sigma^2(\mat X^\top \mat X)^{-1}/\alpha$, allowing uncertainty to be calibrated independently through $\alpha$. This highlights a fundamental advantage of the $(\alpha,\beta)$-posterior: it permits independent tuning of posterior mean and covariance, which is not possible under either the Bayesian posterior or the $\alpha$-posterior.
\end{itemize}

As indicated, fixing the posterior mean, the employment of the $\beta$ parameter in the $(\alpha,\beta)$-posterior allows a new freedom to arbitrarily tune the posterior covariance. However, in the $\alpha$-posterior, the tuning of the posterior mean and covariance is coupled by $\alpha$: given the posterior mean, the posterior covariance cannot be arbitrarily tuned. Indeed, for a given $(\alpha_1, \beta_1)$-posterior, by letting $\alpha \defeq \frac{\alpha_1}{\beta_1}$ in the $\alpha$-posterior, we can always have $\vec \mu_{n, \frac{\alpha_1}{\beta_1}} = \vec \mu_{n, \alpha_1, \beta_1}$ because the MMSE and MAP estimates coincides for Gaussian prior and likelihood distributions; cf. \eqref{eq:BLR-alpha} and \eqref{eq:BLR-alpha-beta}. That is, the posterior mean of the $\frac{\alpha_1}{\beta_1}$-posterior matches that of the $(\alpha_1, \beta_1)$-posterior; only the ratio $\frac{\alpha_1}{\beta_1}$ matters. However, unless $\beta_1 \equiv 1$, it holds that
\[
    \mat \Sigma_{n, \frac{\alpha_1}{\beta_1}} \overset{\eqref{eq:BLR-alpha}}{=} \left(\mat \Sigma_0^{-1}+\frac{\alpha_1}{\sigma^2 \beta_1} \mat X^{\top} \mat X\right)^{-1} \ne \left(\beta_1 \mat \Sigma_0^{-1}+\frac{\alpha_1}{\sigma^2} \mat X^{\top} \mat X\right)^{-1} \overset{\eqref{eq:BLR-alpha-beta}}{=} \mat \Sigma_{n, \alpha_1, \beta_1}.
\]
That is, the posterior covariance of the $\frac{\alpha_1}{\beta_1}$-posterior cannot always coincide with that of the $(\alpha_1, \beta_1)$-posterior. 

The above properties and advantages of the $(\alpha, \beta)$-posterior have also been noted in Examples \ref{ex:alpha-post-deficiency} and \ref{ex:alpha-post-deficiency-2}, Remarks \ref{rem:deficiency-alpha-overconcentration} and \ref{rem:deficiency-alpha-posterior}, and the experiment \quotemark{\textit{Illustrating Examples: Hidden-Quantity Estimation}} before.

\textit{Performance Measure}: Let $\cal N(\vech \mu, \hat{\mat \Sigma})$ be a generic posterior, among the above three types of posteriors, which serves as an estimate of the true posterior $\cal N(\vec \theta_0, \epsilon \mat I_d)$. In this experiment, we use the Wasserstein distance $d_{\text{W}}$ between $\cal N(\vech \mu, \hat{\mat \Sigma})$ and $\cal N(\vec \theta_0, \epsilon \mat I_d)$ to define the the estimation error (EstiError) of $\cal N(\vech \mu, \hat{\mat \Sigma})$, which, due to $\epsilon \to 0$, is 
\begin{equation}\label{eq:error-linear-regression}
    \text{EstiError}[\cal N(\vech \mu, \hat{\mat \Sigma})] \defeq d_{\text{W}}[\cal N(\vech \mu, \hat{\mat \Sigma}), \cal N(\vec \theta_0, \epsilon \mat I_d)] = \sqrt{\|\vech \mu - \vec \theta_0\|^2_2 + \Tr(\hat{\mat \Sigma})}.
\end{equation}

\textit{Parameter Tuning}: The $\alpha$ parameter in the $\alpha$-posterior is tuned using the popular SafeBayes algorithm, where the search grid for $\alpha$ is set to $\Lambda \defeq \{\frac{1}{2^\kappa}, \frac{1}{2^{\kappa-1}}, \frac{1}{2^{\kappa-2}}, \ldots, \frac{1}{2}, 1\}$ and $\kappa \defeq 10$; see \citet[Algorithm~1]{grunwald2012safe}; \citet[Algorithm~1]{grunwald2017inconsistency}. Note that the SafeBayes algorithm in \citet[Algorithm~1]{grunwald2012safe}; \citet[Algorithm~1]{grunwald2017inconsistency} is essentially a grid search method on $[0, 1]$. Since the predictive distribution of the response $y_t$ given a new test feature vector $\vec x_t$ is $\cal N(\vec x^\top_t \vec \mu_{n, \alpha}, \sigma^2 + \vec x^\top_t \mat \Sigma_{n, \alpha} \vec x_t)$, SafeBayes finds the empirically best value $\hat \alpha$ of $\alpha$ via the following optimization \citep[Eq.~(3.1)]{wu2023comparison}
\begin{equation}\label{eq:SafeBayes-cost}
\hat \alpha \defeq \argmin_{\alpha \in \Lambda} \sum^n_{i = 2} \frac{\left(y_i-\mu_i\right)^2}{2 v_i} + \frac{1}{2} \log \left(2 \pi v_i\right),
\end{equation}
where $\mu_i \defeq \vec x^\top_i \vec \mu_{i-1, \alpha}$ and $v_i \defeq \sigma^2 + \vec x^\top_{i} \mat \Sigma_{i-1, \alpha} \vec x_{i}$. In addition, we also employ the surrogate optimization method on $[0, 10]$ to tune the $\alpha$ parameter. For the proposed $(\alpha, \beta)$-posterior, we tune the parameters $(\alpha, \beta)$ using the surrogate optimization method on $[0, 10] \times [0, 10]$. For technical details of the surrogate optimization method, see Subsection \ref{subsec:parameter-tuning}; for implementation, see MATLAB's $\mathsf{surrogateopt}$ function:  \url{www.mathworks.com/help/gads/surrogateopt.html}.

\textit{Operational Setups}: For the experimental purpose, we set $\mu_w \defeq 2.5$, $\sigma^2_w \defeq 2.5$, $\sigma^2 \defeq 0.25$, and $\vec \theta_0 \defeq [2.5; -2.5; 1; -1; 0.05; -0.05]$ (i.e., $d = 6$). Namely, a large modeling error $\rscl w_i$ exists in the nominal model \eqref{eq:nominal-data-generating}, relative to the true process \eqref{eq:true-data-generating}, so that the likelihood function is highly uncertain. We conduct 1000 independent Monte Carlo (MC) episodes. In each episode, the data set $\{(\vec x_i, y_i)\}_{i \in [n]}$ is generated according to the true data-generating law \eqref{eq:true-data-generating} with the true parameter $\vec \theta_0$, and two independent realizations of this data set are generated, yielding $2n$ observations in total. The first realization is used for training, and the second is used for testing. 
To tune the $\alpha$-posterior and the $(\alpha, \beta)$-posterior using the surrogate optimization method, the cost function \eqref{eq:error-linear-regression} is evaluated on the training data set, where the ground truth $\vec \theta_0$ is labeled for the experimental comparison. 
For each posterior, the overall performance is measured by the average estimation error over 1000 independent MC trials; for each MC, the estimation error \eqref{eq:error-linear-regression} is evaluated on the testing data set.

\textit{Experimental Results}: We report experimental results for the following two cases, across which the prior distribution $\cal N(\vec \mu_0, \mat \Sigma_0)$ varies. Case 1: $\vec \mu_0 \defeq \vec 0$ and $\mat \Sigma_0 \defeq 2.5 \mat I_d$. Case 2: $\vec \mu_0 \defeq \vec \theta_0$ and $\mat \Sigma_0 \defeq 2.5 \mat I_d$. Namely, the prior in Case 1 is miscalibrated (i.e., the center is significantly away from $\vec \theta_0$), while that in Case 2 is well-calibrated. The training and testing errors are displayed in Tables \ref{tab:Bayes-linear-regression-training} and \ref{tab:Bayes-linear-regression-testing}, respectively.

\begin{table}[!htbp]
\centering
\caption{Average Estimation Errors (Training) of Different Posteriors}
\label{tab:Bayes-linear-regression-training}
\begin{tabular}{l|l|ccccc}
\bottomrule[0.8pt]
\multicolumn{1}{l}{}  &          & \tabincell{c}{Bayesian \\ Posterior}  & \tabincell{c}{$\alpha$-Posterior \\ (SafeBayes)}   &   \tabincell{c}{$\alpha$-Posterior \\ (Surrogate)}     &   \tabincell{c}{$(\alpha, \beta)$-Posterior \\ (Surrogate)}  \\ 
\hline 
\multirow{2}{*}{$n=10$} & Case 1   & 4.2049       & 4.3380      & 3.3360    & \textbf{2.5319} \\ 
                      & Case 2   & 4.2202       & 3.5406      & 2.9677    & \textbf{1.2167} \\ 
\hline 
\multirow{2}{*}{$n=20$} & Case 1   & 2.2929       & 4.1306      & 2.2275    & \textbf{1.8705} \\ 
                      & Case 2   & 2.2869       & 3.4586      & 2.2039    & \textbf{1.1854} \\ 
\hline 
\multirow{2}{*}{$n=50$} & Case 1   & 1.2932       & 3.7973      & 1.2822    & \textbf{1.2143} \\ 
                      & Case 2   & 1.2841       & 3.2125      & 1.2734    & \textbf{1.0479} \\ 
\toprule[0.8pt]
\end{tabular}
\end{table}

\begin{table}[!htbp]
\centering
\caption{Average Estimation Errors (Testing) of Different Posteriors}
\label{tab:Bayes-linear-regression-testing}
\begin{tabular}{l|l|ccccc}
\bottomrule[0.8pt]
\multicolumn{1}{l}{}  &          & \tabincell{c}{Bayesian \\ Posterior}  & \tabincell{c}{$\alpha$-Posterior \\ (SafeBayes)}   &   \tabincell{c}{$\alpha$-Posterior \\ (Surrogate)}     &   \tabincell{c}{$(\alpha, \beta)$-Posterior \\ (Surrogate)}  \\ 
\hline 
\multirow{2}{*}{$n=10$} & Case 1   & 4.2058       & 4.4192      & 3.7848    & \textbf{2.8890} \\ 
                      & Case 2   & 4.1314       & 3.5268      & 3.4523    & \textbf{1.2521} \\ 
\hline 
\multirow{2}{*}{$n=20$} & Case 1   & 2.3042       & 4.1778      & 2.3164    & \textbf{2.0968} \\ 
                      & Case 2   & 2.2674       & 3.4505      & 2.2899    & \textbf{1.2826} \\ 
\hline 
\multirow{2}{*}{$n=50$} & Case 1   & 1.2968       & 3.8224      & 1.2907    & \textbf{1.2812} \\ 
                      & Case 2   & 1.3116       & 3.2112      & 1.3022    & \textbf{1.2094} \\ 
\toprule[0.8pt]
\end{tabular}
\end{table}

As we can see from Tables \ref{tab:Bayes-linear-regression-training} and \ref{tab:Bayes-linear-regression-testing}, the $\alpha$-posterior tuned by SafeBayes may contrarily worsen the performance, compared to the conventional Bayesian posterior, because in our experiments, the performance measure is the distributional estimation error defined in \eqref{eq:error-linear-regression}, while in SafeBayes, the cost function to minimize is the cumulative prediction error of $\{y_i\}_{i \in [n]}$ defined in \eqref{eq:SafeBayes-cost}. This observation echoes our previous claim about parameter tuning: that is, the optimal parameter in one sense does not necessarily guarantee satisfactory performance in another. However, the $\alpha$-posterior tuned by the surrogate method consistently outperforms the conventional Bayesian posterior. In addition, the $(\alpha, \beta)$-posterior tuned by the surrogate optimization method can outperform both the conventional Bayesian posterior and the $\alpha$-posterior. This is because both the conventional Bayesian posterior and the $\alpha$-posterior are special cases of the $(\alpha, \beta)$-posterior, under parameter pairs $(1, 1)$ and $(\alpha, 1)$, respectively. Hence, in tuning $(\alpha, \beta)$ on $[0, 10] \times [0, 10]$ using the surrogate optimization method, the special points $(1, 1)$ and $(\alpha, 1)$ have already been implicitly evaluated and compared.

    
    

\subsubsection{Real-World Applications: Text/Image Classification and State Estimation}\label{subsebsec:real-world-applications}
This subsection presents simulated and real-data experiments on classification and estimation tasks, in Bayesian machine learning and signal processing, to show the superiority of the $(\alpha, \beta)$-posterior \eqref{eq:generalized-bayes-rule}, referred to as the UA posterior, over the conventional Bayesian posterior \eqref{eq:bayes-rule} and the $\alpha$-posterior \eqref{eq:alpha-posterior-intro}. 
For technical basics on problem formulation and parameter tuning, see Subsections \ref{subsec:applications} and \ref{subsec:parameter-tuning}, respectively. 
In particular, the UA naive Bayes MAP classifier, the UA Kalman filter, the UA particle filter, and the UA interactive-multiple-model filter are suggested and experimentally validated. Since the $\alpha$-posterior is a special case of the UA posterior under the parameter pair $(\alpha, 1)$, when empirically tuning $(\alpha, \beta)$ on $[0, \tau] \times [0, \tau]$ for $\tau > 1$, the performance of the UA posterior is guaranteed to be no worse than that of the $\alpha$-posterior. Therefore, in these experiments, we do not explicitly take into consideration the $\alpha$-posterior for performance comparison, unless it is necessary. For detailed experimental setups, results, discussions, and insights, see Appendix \ref{sec:append-applications} in supplementary materials, due to page limit.

\section{Conclusions}\label{sec:conclusion}
To combat model misspecifications in prior distributions and/or data distributions, we propose to generalize the conventional Bayes' rule to the uncertainty-aware Bayes' rule. The uncertainty-aware Bayes' rule balances the relative importance of the prior distribution and the likelihood distribution by simply taking the exponentiation of the prior distribution and the likelihood distribution. We show that this exponentiation operator essentially adjusts the entropy (i.e., the concentration, the spread) of the prior distribution and the likelihood distribution. Therefore, with different exponents, the prior distribution and/or the likelihood distribution can be upweighted (i.e., by reducing the entropy) or downweighted (i.e., by inflating the entropy) in computing the posterior. Compared to the existing maximum entropy scheme and the Rule-of-Three framework, the uncertainty-aware Bayes' rule does not introduce much additional computational burden because the exponentiation operation is computationally lightweight. In addition, compared to the existing maximum entropy scheme, the uncertainty-aware Bayes' rule is able to combat the conservativeness (i.e., upweight the useful distributional information) of the employed prior distributions and the likelihood distributions. Moreover, compared to the existing $\alpha$-posterior scheme, the overconcentration issue can be avoided (see Remarks \ref{rem:deficiency-alpha-overconcentration} and \ref{rem:deficiency-alpha-posterior}). Simulated and real-world applications further demonstrate that both aggressiveness and conservativeness can be beneficial in practice, indicating that, when tuning $\alpha$ and $\beta$, the focus should \bfit{not} be restricted to $[0,1]$. However, the optimal parameters $(\alpha, \beta)$ cannot be theoretically specified because they depend on the true prior distribution and data-generating distribution, which are unknown in practice. Therefore, given an application-specific performance measure (see Philosophy \ref{phi:robustness}), the grid search (as the SafeBayes approach does for the $\alpha$-posterior) and surrogate-optimization-based methods are largely suggested to find the best values empirically.

\begin{supplement}
\stitle{Supplementary Materials.pdf}
\sdescription{All appendices are placed in the supplementary materials, including 1) the technical proofs of theorems and lemmas, 2) additional clarifications, illustrative examples, and visual arts for better readability, and 3) extensive experiments.}
\end{supplement}

\bibliographystyle{ba}
\bibliography{References}

\newpage

\begin{figure*}
    \textbf{\Large Supplementary Materials} \\~\\
    \textbf{Uncertainty-Aware Bayes’ Rule and Its Applications} \\
    \textit{By Shixiong Wang}
\end{figure*}

\setcounter{section}{0}
\setcounter{page}{1}

\appendix

\section{Additional Clarifications}
\subsection{Interpretations of Bayes' Rule \captext{\eqref{eq:bayes-rule}} and the \captext{$(\alpha, \beta)$}-Posterior \captext{\eqref{eq:generalized-bayes-rule}}}\label{append:interpretation-bayes-rule}
\textbf{Deterministic and Stochastic Interpretations of Bayes' Rule \eqref{eq:bayes-rule}}. 
In the main body of the paper, we introduced two different interpretations of Bayes' rule \eqref{eq:bayes-rule}. The main difference lies in whether there \bfit{physically} exists a joint data-generating distribution $p(\vec y, \vec \theta)$. Below, we provide specific engineering examples for the two interpretations.

\begin{example}[Deterministic Interpretation in Wireless Sensing]\label{eaxm:deterministic-sensing}
    Suppose we aim to estimate the position of a static target using a radar (Radar A). In this case, the actual position $\vec \theta_0$ of the target is deterministic, fixed, but unknown. The radar's measurement model is $p_{\vec \theta_0}(\vec y)$, where $\vec y$ is the received radio signal that is reflected off the target. In addition to relying on Radar A, we may also consult with an experienced expert, whose prior probabilistic belief of the target's position is $p(\vec \theta)$; this expert's prior information $p(\vec \theta)$ may be from another independent radar (Radar B). Hence, $\vec \theta_0$ is not physically sampled from $p(\vec \theta)$, and no joint data-generating distribution $p(\vec y, \vec \theta)$ can be justified. When the expert is unreliable (e.g., Radar B may be dysfunctional), the prior belief $p(\vec \theta)$ might be uncertain (e.g., the center is far away from $\vec \theta_0$ with small variance). Similarly, when Radar A is dysfunctional, the data distribution $p(\vec y | \vec \theta)$ might be uncertain.
    \stp
\end{example}

\begin{example}[Stochastic Interpretation in Wireless Communications]
    Suppose we aim to estimate the transmitted signal $\vec \theta_0$ of a wireless communications system (e.g., a mobile phone), whose probabilistic channel model (i.e., the signal transmission law) is $p_{\vec \theta_0}(\vec y)$, where $\vec y$ is the received signal. From information transmission theory, we know that the transmitted signal $\vec \theta_0$ is a realization of a random signal source $p(\vec \theta)$ \citep[Chapters~7 and 9]{cover2006elements}. Hence, in this case, there indeed physically exists a joint data-generating distribution $p(\vec y, \vec \theta)$, and the data-generating process can be mathematically stated as follows: $\vec \theta_0 \sim p(\vec \theta)$ and $\vec y \sim p_{\vec \theta_0}(\vec y)$. Because the assumed source law $p(\vec \theta)$ can differ from the actual source law $p_0(\vec \theta)$, the uncertainty in $p(\vec \theta)$ should be considered. Likewise, uncertainties may exist in the channel model $p(\vec y | \vec \theta)$ due to humans' inexact knowledge of the propagation law of radio signals.
    \stp
\end{example}

\textbf{Uncertainty Awareness in the $(\alpha, \beta)$-Posterior \eqref{eq:generalized-bayes-rule}}. As indicated, both the prior distribution $p(\vec \theta)$ and the data-generating distribution $p(\vec y | \vec \theta)$ can be \bfit{independently} uncertain, compared to their ground truths $p_0(\vec \theta)$ and $p_0(\vec y | \vec \theta)$, respectively. Note that, in the deterministic interpretation, the true prior distribution $p_0(\vec \theta)$ should be understood as the delta distribution $\delta_{\vec \theta_0}(\vec \theta)$ concentrated at the true value $\vec \theta_0$. Therefore, the uncertainties in the two distributions should be independently treated in calculating the posterior distribution, which raises the conceptual insufficiency of the $\alpha$-posterior. This reasoning gives a direct motivation for why we should independently modify the two distributions in the $(\alpha, \beta)$-posterior \eqref{eq:generalized-bayes-rule}; see \eqref{eq:opt-gen-posterior-dist}, Theorem \ref{thm:gen-posterior-dist}, and Definition \ref{def:alpha-beta-posterior}.

\subsection{Fusion of Multiple Priors}\label{append:multi-priors}
In this appendix, we present an engineering example where multiple priors may exist.
\begin{example}[Continued from Example \ref{eaxm:deterministic-sensing}]
Suppose we have multiple experts with whom we consult. In addition to the measurement model $p_{\vec \theta_0}(\vec y)$, we can have $m$ priors $p_1(\vec \theta)$, $p_2(\vec \theta)$, $\ldots$, $p_m(\vec \theta)$. In this case, the fusion rule is given in \eqref{eq:opt-gen-posterior-dist-multi-sample-multi-prior} and \eqref{eq:gen-bayes-rule-multi-prior}.
\stp
\end{example}

Using federated-learning language, the above example can be stated as follows. Let $m$ statisticians collaboratively estimate the same unknown quantity, and the $i^\th$ statistician independently observes the data set $\cal Y_i$, where $i = 1, 2, \ldots, m$. However, due to data privacy considerations, statisticians do not share their local raw data with others. Hence, taking the $m^\th$ statistician as an example, other $m-1$ statisticians can send their local posterior distributions $\{p(\vec \theta | \cal Y_i)\}_{i = 1, 2, \ldots, m-1}$ for information fusion: Specifically, for the $m^\th$ statistician, $\{p_i(\vec \theta) \defeq p(\vec \theta | \cal Y_i)\}_{i = 1, 2, \ldots, m-1}$ act as $m-1$ priors, while $p(\cal Y_m | \vec \theta)$ serves as the likelihood model. In such a situation, if the $m^\th$ statistician does not entirely trust the $m-1$ priors, the $\beta$-parameterization is needed, as in \eqref{eq:alpha-beta-posterior-multi-prior}.

\section{Proof of Lemma \captext{\ref{lem:posterior-dist}}}\label{append:posterior-dist}
\begin{proof}
According to the optimization-centric interpretation of the Bayesian posterior distribution $p(\vec \theta | \vec y)$ in \eqref{eq:bayes-rule}, $p(\vec \theta | \vec y)$ solves
\[
\min_{q(\bm \theta)} \KL{q(\bm \theta)}{p(\bm \theta)} + \E_{\vec \theta \sim q(\vec \theta)} [-\ln{p(\bm y |\bm \theta)}].
\]
Let $C_{\vec y} \defeq \int p(\bm y| \vec \theta) \d \vec \theta$. The above optimization problem is equivalent, in the sense of the same minimizer, to
\[
\begin{array}{l}
\displaystyle \min_{q(\bm \theta)} \KL{q(\bm \theta)}{p(\bm \theta)} + \E_{\vec \theta \sim q(\vec \theta)}\left[-\ln\frac{p(\bm y |\bm \theta)}{C_{\vec y}}\right], \\
\quad = \displaystyle \min_{q(\bm \theta)} \KL{q(\bm \theta)}{p(\bm \theta)} + \E_{\vec \theta \sim q(\vec \theta)}[-\ln l(\bm \theta)], \\
\quad = \displaystyle \min_{q(\bm \theta)} \KL{q(\bm \theta)}{p(\bm \theta)} + \KL{q(\vec \theta)}{l(\vec \theta)} + \Ent q(\vec \theta).
\end{array}
\]
This completes the proof.
\end{proof}

\section{Proof of Lemma \captext{\ref{lemma:problem-finiteness}}}\label{append:problem-finiteness}
\begin{proof}
Problem \eqref{eq:opt-gen-posterior-dist} can be rewritten as
\[
\min_{q(\vec \theta)} \quad (\alpha_3 - \alpha_1 - \alpha_2) \Ent q(\vec \theta)
- \alpha_1 \E_{\vec \theta \sim q(\vec \theta)}[\ln p(\bm \theta)]
- \alpha_2 \E_{\vec \theta \sim q(\vec \theta)}[\ln l(\bm \theta)].
\]
If $\alpha_3 > \alpha_1 + \alpha_2$, any delta distributions on $\Theta$ can minimize the above objective to negative infinity because the differential entropy of a delta distribution is negative infinity. This completes the proof.
\end{proof}

\section{Proof of Theorem \captext{\ref{thm:gen-posterior-dist}}}\label{append:gen-posterior-dist}
\begin{proof}
We have
\[
\begin{array}{l}
\displaystyle \min_{q(\vec \theta)} \alpha_1 \KL{q(\vec \theta)}{p(\vec \theta)} + \alpha_2 \KL{q(\vec \theta)}{l(\vec \theta)} + \alpha_3 \Ent q(\vec \theta) \\ 
\quad  = \displaystyle \min_{q(\vec \theta)} \int q(\vec \theta) \ln\frac{q^{\alpha_1 + \alpha_2}(\vec \theta)}{p^{\alpha_1}(\vec \theta)l^{\alpha_2}(\vec \theta)q^{\alpha_3}(\vec \theta)} \d \vec \theta \\
\quad = \displaystyle \min_{q(\vec \theta)} \int q(\vec \theta) \ln\frac{q^{\alpha_1 + \alpha_2-\alpha_3}(\vec \theta)}{p^{\alpha_1}(\vec \theta)l^{\alpha_2}(\vec \theta)} \d \vec \theta.
\end{array}
\]

When $\alpha_3 = \alpha_1 + \alpha_2$, any distribution supported on $\Theta^*$ solves the above optimization problem, where $\Theta^*$ contains all maximizers of $\ln[p^{\alpha_1}(\vec \theta)l^{\alpha_2}(\vec \theta)]$.

When $\alpha_3 < \alpha_1 + \alpha_2$, the above optimization problem is equivalent to
\[
(\alpha_1 + \alpha_2 - \alpha_3) \cdot \displaystyle \min_{q(\vec \theta)} \int q(\vec \theta) \ln\frac{q(\vec \theta)}{p^{\frac{\alpha_1}{\alpha_1 + \alpha_2 - \alpha_3}}(\vec \theta) \cdot l^{\frac{\alpha_2}{\alpha_1 + \alpha_2 - \alpha_3}}(\vec \theta)} \d \vec \theta,
\]
which is further equivalent, in the sense of the same minimizer, to
\[
\displaystyle \min_{q(\vec \theta)} \KL{q(\vec \theta)}{\frac{p^{\frac{\alpha_1}{\alpha_1 + \alpha_2 - \alpha_3}}(\vec \theta) \cdot l^{\frac{\alpha_2}{\alpha_1 + \alpha_2 - \alpha_3}}(\vec \theta)}{C}},
\]
where $C$ is the normalizer. This completes the proof.
\end{proof}

\section{Examples of Uncertainty-Aware Bayes' Rule}\label{append:examples-bayes-rule}
Below we give several motivational examples of the uncertainty-aware Bayes' rule \eqref{eq:gen-bayes-rule} or \eqref{eq:alpha-beta-posterior}; the first two are well-established in existing literature of Bayesian statistics.

\begin{example}[Conventional Bayesian Posterior]\label{ex:conventional-Bayes-posterior}
The conventional Bayes' rule \eqref{eq:bayes-rule} [resp. \eqref{eq:bayes-rule-2}] is a special case of the uncertainty-aware Bayes' rule \eqref{eq:gen-bayes-rule-2} [resp. \eqref{eq:gen-bayes-rule} or \eqref{eq:alpha-beta-posterior}] when $\alpha_1 = \alpha_2 = \alpha_3$ or $\alpha = \beta = 1$.
\stp
\end{example}

\begin{example}[$\alpha$-Posterior]\label{ex:alpha-posterior}
When $\alpha_2 = \alpha_3$ and $\alpha_1 > 0$, \eqref{eq:gen-bayes-rule} reduces to
\begin{equation}\label{eq:alpha-posterior-origin}
    p_g(\vec \theta | \vec y) \propto p(\vec \theta) \cdot l^{\frac{\alpha_2}{\alpha_1}}(\vec \theta).
\end{equation}
By letting $\alpha \defeq \frac{\alpha_2}{\alpha_1}$, we obtain the $\alpha$-posterior
\begin{equation}\label{eq:alpha-posterior}
    p_g(\vec \theta | \vec y) \propto p(\vec \theta) \cdot l^{\alpha}(\vec \theta),
\end{equation}
where $0 \le \alpha < \infty$. When $\alpha_1 = 0$, $p_g(\vec \theta | \vec y)$ is an arbitrary distribution supported on $\Theta^*$ where $\Theta^* \defeq \argmax_{\vec \theta} \ln l(\vec \theta)$ contains all maximum likelihood estimates.
\stp
\end{example}

The $\alpha$-posterior in \eqref{eq:alpha-posterior} is a well-established proposal in Bayesian statistics; see, e.g., \citet[Sec.~6.8.5 and Sec.~8.6]{ghosal2017fundamentals}; \citet{medina2022robustness,alquier2020concentration}. Given $0 < \alpha < 1$ and several other technical regularity conditions (e.g., $\Theta^*$ is a singleton), the $\alpha$-posteriors can be shown to have posterior consistency [see \citet[Thm.~6.54,~Ex.~8.44]{ghosal2017fundamentals} and \citet{alquier2020concentration}], asymptotic normality [see \citet[Thm.~1]{medina2022robustness}], and robustness against likelihood-model misspecifications [see \citet[Sec.~4]{medina2022robustness}]. However, these conditions might be practically restrictive; see Example \ref{ex:alpha-post-deficiency} where $\alpha > 1$ is preferred; see also \citet[Tables~1-4]{wu2023comparison} where $\alpha > 1$ is suggested by some parameter-tuning methods for a given performance measure.

\begin{example}[$\beta$-Posterior]\label{ex:beta-posterior}
When $\alpha_1 = \alpha_3$ and $\alpha_2 > 0$, \eqref{eq:gen-bayes-rule} reduces to
\begin{equation}\label{eq:beta-posterior-origin}
    p_g(\vec \theta | \vec y) \propto p^{\frac{\alpha_1}{\alpha_2}}(\vec \theta) \cdot l(\vec \theta).
\end{equation}
By letting $\beta \defeq \frac{\alpha_1}{\alpha_2}$, we obtain the $\beta$-posterior\footnote{To differentiate with the $\alpha$-posterior, we name it using $\beta$.}
\begin{equation}\label{eq:beta-posterior}
    p_g(\vec \theta | \vec y) \propto p^{\beta}(\vec \theta) \cdot l(\vec \theta),
\end{equation}
where $0 \le \beta < \infty$. When $\alpha_2 = 0$, $p_g(\vec \theta | \vec y)$ is an arbitrary distribution supported on $\Theta^*$ where $\Theta^* \defeq \argmax_{\vec \theta} \ln p(\vec \theta)$ contains all maximum prior estimates.
\stp
\end{example}

Compared with the $\alpha$-posterior in Example \ref{ex:alpha-posterior} that modifies the likelihood distribution $l(\vec \theta)$, the $\beta$-posterior modifies the prior distribution $p(\vec \theta)$.

\begin{example}[$\gamma$-Posterior]\label{ex:gamma-posterior}
When $\alpha_1 = \alpha_2$ and $2 \alpha_1 > \alpha_3$, \eqref{eq:gen-bayes-rule} reduces to
\begin{equation}\label{eq:gamma-posterior-origin}
    p_g(\vec \theta | \vec y) \propto p^{\frac{\alpha_1}{2\alpha_1 - \alpha_3}}(\vec \theta) \cdot l^{\frac{\alpha_1}{2\alpha_1 - \alpha_3}}(\vec \theta),
\end{equation}
By letting $\gamma \defeq \frac{\alpha_1}{2\alpha_1 - \alpha_3}$, we obtain the $\gamma$-posterior\footnote{To differentiate with the $\alpha$- and $\beta$-posterior, we name it using $\gamma$.}
\begin{equation}\label{eq:gamma-posterior}
    p_g(\vec \theta | \vec y) \propto p^{\gamma}(\vec \theta) \cdot l^{\gamma}(\vec \theta),
\end{equation}
where $0 \le \gamma < \infty$. When $2 \alpha_1 = \alpha_3$, $p_g(\vec \theta | \vec y)$ is an arbitrary distribution supported on $\Theta^*$ where $\Theta^* \defeq \argmax_{\vec \theta} \ln p(\vec \theta) + \ln l(\vec \theta)$ contains all usual maximum \textit{a-posteriori} estimates.
\stp
\end{example}

Compared with the $\alpha$-posterior and the $\beta$-posterior, the $\gamma$-posterior modifies both the prior distribution and the likelihood distribution. The $\beta$- and $\gamma$-posteriors are new proposals of this paper and have not been studied in the literature.

\begin{example}[$\alpha$-Likelihood]\label{ex:alpha-likelihood}
When $\alpha_1 = 0$ and $\alpha_2 > \alpha_3$, \eqref{eq:gen-bayes-rule} reduces to
\begin{equation}\label{eq:alpha-likelihood-origin}
    p_g(\vec \theta | \vec y) \propto l^{\frac{\alpha_2}{\alpha_2 - \alpha_3}}(\vec \theta).
\end{equation}
By letting $\alpha \defeq \frac{\alpha_2}{\alpha_2 - \alpha_3}$, we obtain the $\alpha$-likelihood
\begin{equation}\label{eq:alpha-likelihood}
    p_g(\vec \theta | \vec y) \propto l^{\alpha}(\vec \theta),
\end{equation}
where $0 \le \alpha < \infty$. When $\alpha_2 = \alpha_3$, $p_g(\vec \theta | \vec y)$ is an arbitrary distribution supported on $\Theta^*$ where $\Theta^* \defeq \argmax_{\vec \theta} \ln l(\vec \theta)$ contains all maximum likelihood estimates.
\stp
\end{example}

In the uncertainty-aware Bayes' posterior \eqref{eq:alpha-likelihood} induced by the $\alpha$-likelihood $l^{\alpha}(\vec \theta)$, the prior distribution $p(\vec \theta)$ is completely ignored. To clarify further, we only trust the likelihood distribution with the confidence level $\alpha$ and absolutely distrust the prior distribution.

\begin{example}[$\alpha$-Prior]\label{ex:alpha-prior}
When $\alpha_2 = 0$ and $\alpha_1 > \alpha_3$, \eqref{eq:gen-bayes-rule} reduces to
\begin{equation}\label{eq:alpha-prior-origin}
    p_g(\vec \theta | \vec y) \propto p^{\frac{\alpha_1}{\alpha_1 - \alpha_3}}(\vec \theta).
\end{equation}
By letting $\alpha \defeq \frac{\alpha_1}{\alpha_1 - \alpha_3}$, we obtain the $\alpha$-prior
\begin{equation}\label{eq:alpha-prior}
    p_g(\vec \theta | \vec y) \propto p^{\alpha}(\vec \theta),
\end{equation}
where $0 \le \alpha < \infty$. When $\alpha_1 = \alpha_3$, $p_g(\vec \theta | \vec y)$ is an arbitrary distribution supported on $\Theta^*$ where $\Theta^* \in \argmax_{\vec \theta} \ln p(\vec \theta)$ contains all maximum prior estimates.
\stp
\end{example}

In the uncertainty-aware Bayes' posterior \eqref{eq:alpha-prior} induced by the $\alpha$-prior $p^{\alpha}(\vec \theta)$, the likelihood distribution $l(\vec \theta)$ is completely ignored. To clarify further, we only trust the prior distribution with the confidence level $\alpha$ and absolutely distrust the likelihood distribution.

\begin{example}[$\alpha$-Pooled Posterior]\label{ex:alpha-pooled-posterior}
When $\alpha_3 = 0$, \eqref{eq:gen-bayes-rule} reduces to
\begin{equation}\label{eq:alpha-pooled-posterior-origin}
    p_g(\vec \theta | \vec y) \propto p^{\frac{\alpha_1}{\alpha_1 + \alpha_2}}(\vec \theta) \cdot l^{\frac{\alpha_2}{\alpha_1 + \alpha_2}}(\vec \theta).
\end{equation}
By letting $\alpha \defeq \alpha_1/(\alpha_1 + \alpha_2)$, we obtain the $\alpha$-pooled-posterior
\begin{equation}\label{eq:alpha-pooled-posterior}
    p_g(\vec \theta | \vec y) \propto p^{\alpha}(\vec \theta) \cdot l^{1 - \alpha}(\vec \theta),
\end{equation}
where $0 \le \alpha \le 1$.
\stp
\end{example}

An engineering application of \eqref{eq:alpha-pooled-posterior} is reported in social learning \citep[Eq.~(7)]{bordignon2021adaptive}; \citep[Eq.~(10)]{hu2023optimal}. The $\alpha$-pooled-posterior in \eqref{eq:alpha-pooled-posterior} is a $\alpha$-log-linear pooling \citep{rufo2012loglinear,koliander2022fusion} of the prior distribution $p(\vec \theta)$ and the likelihood distribution $l(\vec \theta)$, which is an alternative of the linear pooling rule ${\alpha} p(\vec \theta) + {(1 - \alpha)} l(\vec \theta)$. The log-linear pooling rule is standard in Bayesian statistics for fusing multiple priors (from multiple experts) to obtain an integrated prior \citep{rufo2012loglinear}. However, the $\alpha$-pooled-posterior in \eqref{eq:alpha-pooled-posterior} claims that the likelihood distribution can also be used as a \quotemark{prior} where the collected data $\vec y$ serve as an expert. (But this data-driven expert does not insist on his/her opinion due to randomness of $\vec y$; he/she is a flexible expert; instead, conventional non-data-driven experts are stubborn.) In addition, from \eqref{eq:opt-gen-posterior-dist}, we can see that when $\alpha_3 = 0$, the entropy regularizer is removed. Therefore, it is natural to imagine that the $\alpha$-pooled posterior in \eqref{eq:alpha-pooled-posterior} has larger entropy than the conventional Bayes' posterior $p(\vec \theta | \vec y)$.

From Examples \ref{ex:conventional-Bayes-posterior}-\ref{ex:alpha-pooled-posterior}, we can see that the $(\alpha,\beta)$-posterior, compared to the conventional Bayes' posterior (when $\alpha=\beta=1$), defines general fusing rules of the prior distribution $p(\vec \theta)$ and the likelihood distribution $l(\vec \theta)$; we refer to the law of defining the $(\alpha,\beta)$-posterior as the \textit{uncertainty-aware Bayes' rule}.

\section{Proof of Theorem \captext{\ref{thm:alpha-scale-ent-monotonicity}}}\label{append:alpha-scale-ent-monotonicity}
\begin{proof}
Let $C_\alpha \defeq \int h^{\alpha}(\vec \theta) \d \vec \theta$.
We have
\begin{equation*}
    \begin{array}{cl}
        \Ent h^{(\alpha)}(\vec \theta) &= \displaystyle \int \frac{h^{\alpha}(\vec \theta)}{\int h^{\alpha}(\vec \theta) \d \vec \theta} \cdot -\ln \frac{h^{\alpha}(\vec \theta)}{\int h^{\alpha}(\vec \theta) \d \vec \theta} \d \vec \theta \\

        &= \displaystyle \frac{1}{C_\alpha} \int h^{\alpha}(\vec \theta) \cdot \Big[-\ln h^{\alpha}(\vec \theta) + \ln C_\alpha \Big] \d \vec \theta \\

        &= \ln C_\alpha + \displaystyle \frac{\alpha}{C_\alpha} \int h^{\alpha}(\vec \theta) \cdot -\ln h(\vec \theta) \d \vec \theta.
    \end{array}
\end{equation*}
Therefore,
\begin{equation*}
    \begin{array}{l}
        \displaystyle \frac{\d \Ent h^{(\alpha)}(\vec \theta)}{\d \alpha} = \frac{\alpha}{C^2_{\alpha}} \cdot \Bigg\{ \left[ \displaystyle \int h^\alpha(\vec \theta) \ln h(\vec \theta) \d \vec \theta \right]^2 - \\
        \quad \quad \quad \quad \quad \quad \displaystyle \int h^\alpha (\vec \theta) \d \vec \theta \cdot \int h^\alpha (\vec \theta) \ln h(\vec \theta) \ln h(\vec \theta) \d \vec \theta
        \Bigg\}.
    \end{array}
\end{equation*}
Since $h^\alpha (\vec \theta) \ge 0$ for every $\alpha$ and $\vec \theta$, according to the Cauchy--Schwarz inequality, we have
\begin{equation*}
    \begin{array}{l}
        \displaystyle \int \left[\sqrt{h^\alpha (\vec \theta)}\right]^2 \d \vec \theta \int \left[\sqrt{ h^\alpha (\vec \theta) \cdot \ln h(\vec \theta) \cdot \ln h(\vec \theta)}\right]^2 \d \vec \theta \\
        \quad \quad \ge \left[ \displaystyle \int \sqrt{ h^\alpha (\vec \theta) \cdot h^\alpha (\vec \theta) \cdot \ln h(\vec \theta) \cdot \ln h(\vec \theta)} \d \vec \theta \right]^2 \\
        \quad \quad = \left[ \displaystyle \int h^\alpha (\vec \theta) \cdot \big| \ln h(\vec \theta) \big| \d \vec \theta \right]^2 \\
        \quad \quad \ge \left[ \displaystyle \int h^\alpha  (\vec \theta) \cdot \ln h(\vec \theta) \d \vec \theta \right]^2.
    \end{array}
\end{equation*}
As a result,
\[
\displaystyle \frac{\d \Ent h^{(\alpha)}(\vec \theta)}{\d \alpha} \le 0.
\]
If $\alpha \neq 0$, the equality holds if and only if $h(\vec \theta)$ is a uniform distribution. This completes the proof.
\end{proof}

\section{Illustrative Examples}\label{append:illustration-examples}

\subsection{Illustration of \captext{$\alpha$}-Scaled Distribution}\label{append:illlu-alpha-scaling}
We visualize the $\alpha$-scaled distribution $h^{(\alpha)}(\vec \theta)$ when $h(\vec \theta)$ is discrete; see Figure \ref{fig:alpha_scaling}. As we can see, when $0 \le \alpha < 1$, the probabilities become more balanced, while when $\alpha > 1$, they become more unbalanced.
\begin{figure}[!htbp]
    \centering
    \includegraphics[height=3.5cm]{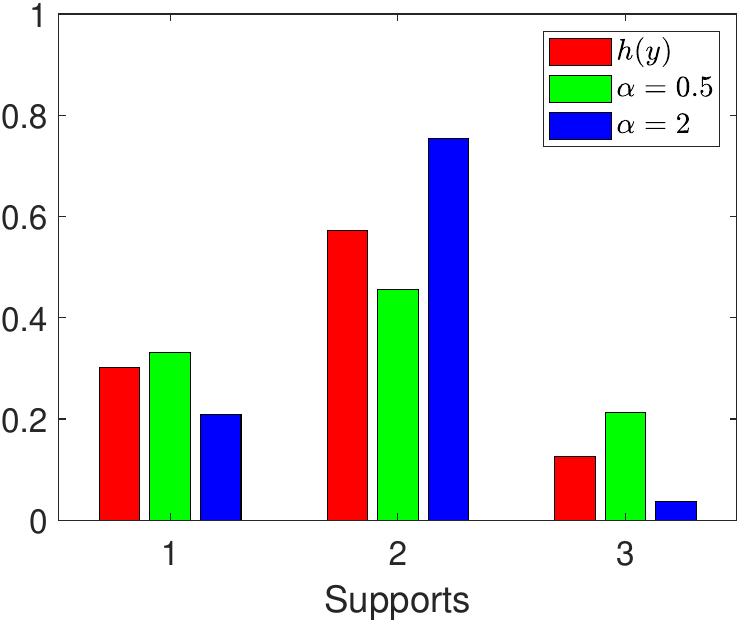}
    \caption{A $3$-atom discrete distribution $h(\theta)$ and induced $h^{(\alpha)}(\theta)$ with different $\alpha \in \{0.5, 2\}$. Under $\alpha = 0.5$, $\Ent h^{(\alpha)}(\theta) > \Ent h(\theta)$ (i.e., the former has more balanced masses, while the latter has more unbalanced masses). Under $\alpha = 2$, $\Ent h^{(\alpha)}(\theta) < \Ent h(\theta)$.}
    \label{fig:alpha_scaling}
\end{figure}

\subsection{Illustration of Entropy Difference \captext{$E(\alpha)$}}\label{append:illu-ent-diff}

A visual illustration of the entropy difference $E(\alpha)$ when $h(\vec \theta)$ is discrete is given in Figure \ref{fig:fig_ent_diff}.
\begin{figure}[!htbp]
    \centering
    \includegraphics[height=3.5cm]{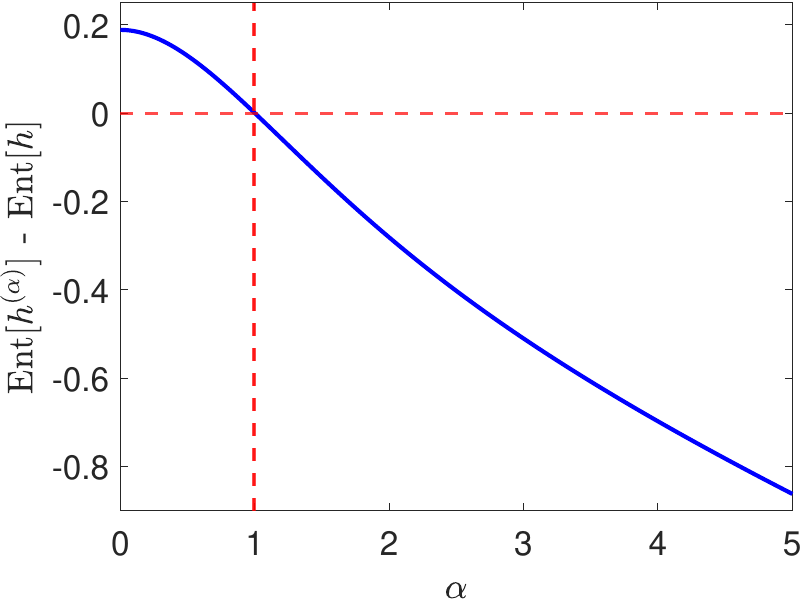}
    \caption{The entropy difference $E(\alpha) \defeq \Ent h^{(\alpha)}(\vec \theta) - \Ent h(\vec \theta)$ against $\alpha$; $h(\vec \theta)$ is a randomly generated $50$-atom discrete distribution.}
    \label{fig:fig_ent_diff}
\end{figure}

\subsection{Illustration of Closeness \captext{$\KL{h(\vec \theta)}{h^{(\alpha)}(\vec \theta)}$}}\label{append:illu-closeness}
A visual illustration of Theorem \ref{thm:closeness-from-alpha-scaled} when $h(\vec \theta)$ is discrete is given in Figure \ref{fig:fig_KL_closeness}.
\begin{figure}[!htbp]
    \centering
    \includegraphics[height=3.5cm]{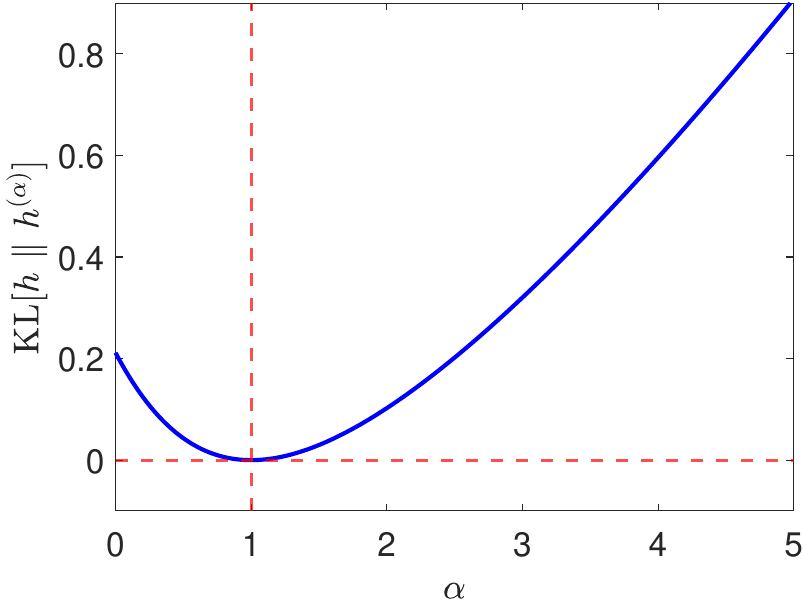}
    \caption{A visual illustration of the closeness $\KL{h(\vec \theta)}{h^{(\alpha)}(\vec \theta)}$ to $h(\vec \theta)$ from $h^{(\alpha)}(\vec \theta)$; $h(\vec \theta)$ is a randomly generated  $50$-atom discrete distribution.}
    \label{fig:fig_KL_closeness}
\end{figure}

\subsection{Illustration of \captext{$(\alpha, \beta)$}-Posterior}\label{ex:illu-alpha-beta-posterior}
We visualize the $(\alpha,\beta)$-posterior, compared with the conventional Bayes' posterior. We focus on a mean estimation problem where $\theta_0$ is the true mean of the random variable $\rscl y$. We suppose that the prior distribution is $p(\theta) \defeq \cal N(\theta; 0, 1)$. The likelihood function is $\theta \mapsto \cal N(y; \theta, 1)$; if we assume that the measurement is $y \defeq 5$, the likelihood distribution is therefore $l(\theta) \defeq \cal N(\theta; 5, 1)$. The visualizations of the $(\alpha,\beta)$-posterior, under different value pairs of $(\alpha,\beta)$, are given in Figure \ref{fig:alpha-beta-posterior-illustrations}. 

\begin{figure*}[!htbp]
    \centering
    \subfigure[$(0,2.5)$-posterior]{
        \begin{minipage}[htbp]{0.31\linewidth}
            \centering
            \includegraphics[height=3.2cm]{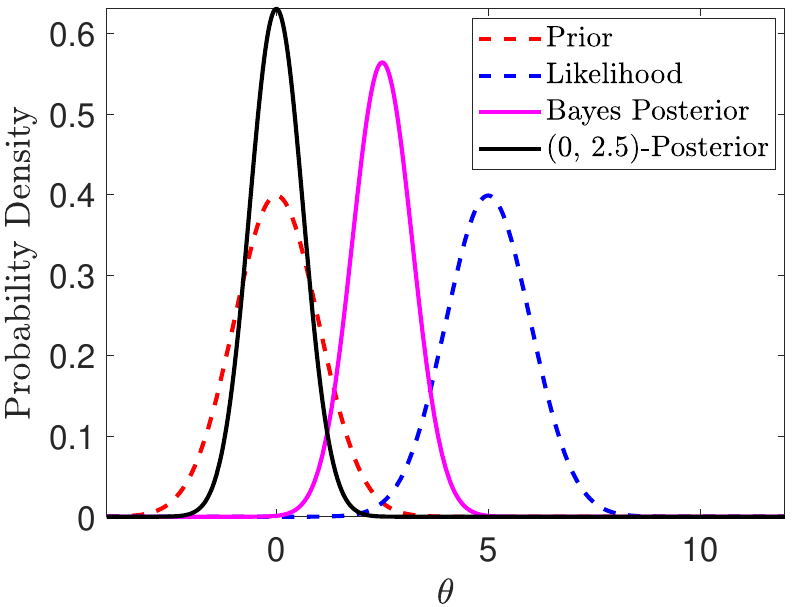}
        \end{minipage}
    }
    \subfigure[$(0.25,0.25)$-posterior]{
        \begin{minipage}[htbp]{0.31\linewidth}
            \centering
            \includegraphics[height=3.2cm]{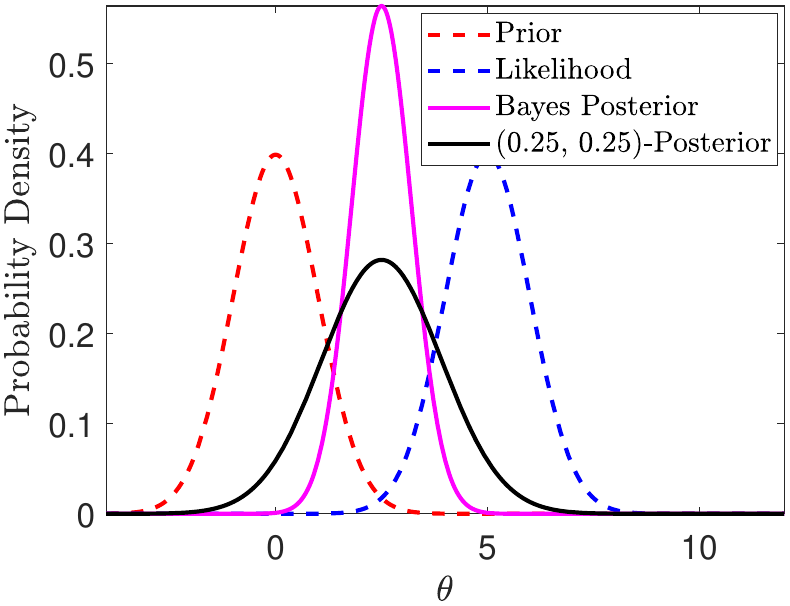}
        \end{minipage}
    }
    \subfigure[$(0.25,0.5)$-posterior]{
        \begin{minipage}[htbp]{0.31\linewidth}
            \centering
            \includegraphics[height=3.2cm]{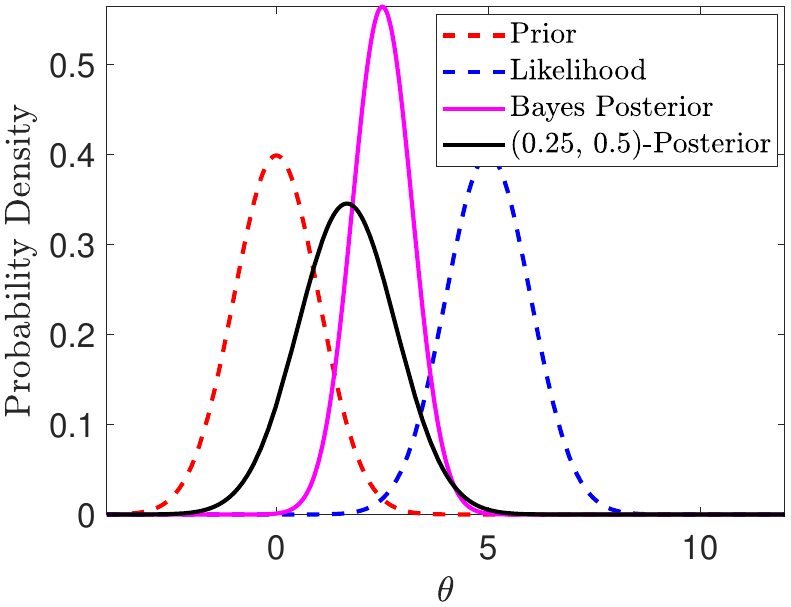}
        \end{minipage}
    }
    
    \subfigure[$(0.5,2)$-posterior]{
        \begin{minipage}[htbp]{0.31\linewidth}
            \centering
            \includegraphics[height=3.2cm]{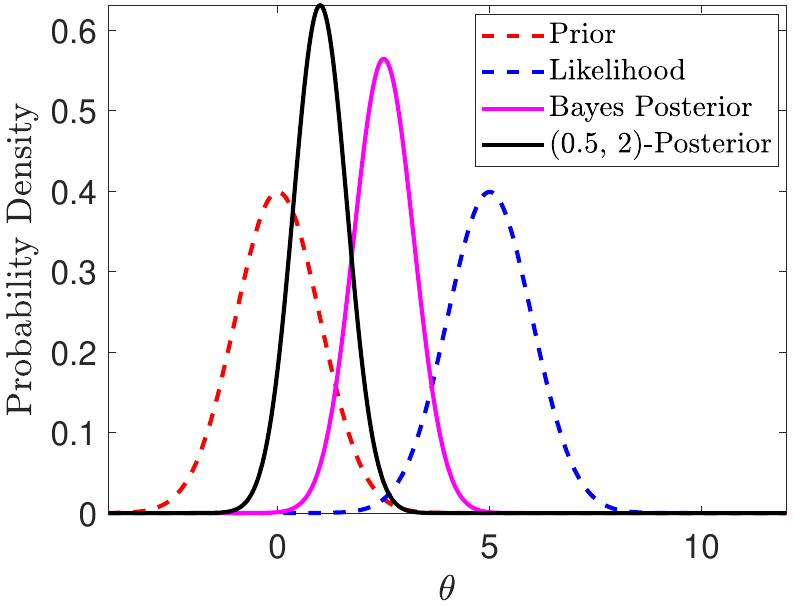}
        \end{minipage}
    }
    \subfigure[$(1.5,2)$-posterior]{
        \begin{minipage}[htbp]{0.31\linewidth}
            \centering
            \includegraphics[height=3.2cm]{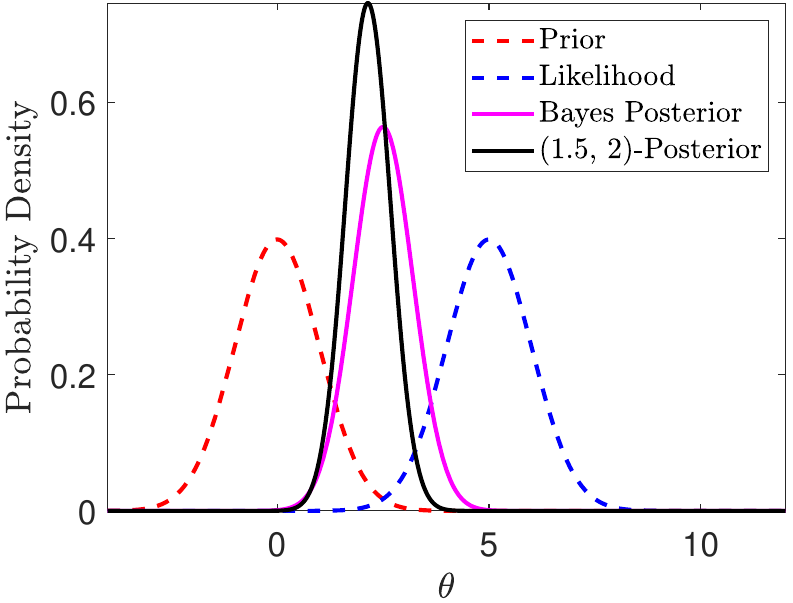}
        \end{minipage}
    }
    \subfigure[$(2,2)$-posterior]{
        \begin{minipage}[htbp]{0.31\linewidth}
            \centering
            \includegraphics[height=3.2cm]{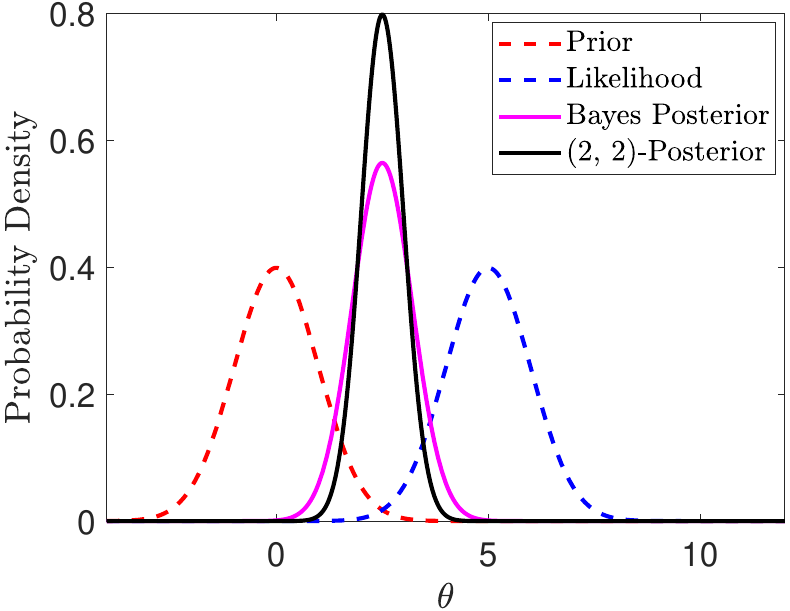}
        \end{minipage}
    }
    \caption{Illustrating examples of the $(\alpha,\beta)$-posterior. The prior distribution is $p(\theta) \defeq \cal N (\theta; 0, 1)$. The likelihood distribution is $l_{\vec y \defeq 5}(\theta) \defeq \cal N (\theta; 5, 1)$.}
    \label{fig:alpha-beta-posterior-illustrations}
\end{figure*}

As motivational examples, we only interpret the first three sub-figures. In Figure \ref{fig:alpha-beta-posterior-illustrations}(a), $\alpha = 0$ means that the likelihood distribution is absolutely unreliable and the data evidence is completely ignored. Therefore, the $(0,2.5)$-posterior distribution tends to overlap with the prior distribution. Moreover, because $\beta = 2.5 > 1$, the specified prior distribution is thought to be overly conservative and we propose to increase its concentration. In Figure \ref{fig:alpha-beta-posterior-illustrations}(b), $\alpha = \beta = 0.25$ means that both the prior distribution and the likelihood distribution are believed to be overly aggressive, and therefore, we propose to reduce their concentration. As a result, the $(0.25,0.25)$-posterior distribution has larger entropy than the conventional Bayes' posterior distribution. In Figure \ref{fig:alpha-beta-posterior-illustrations}(c), since $\alpha = 0.25 < 0.5 = \beta$, both the prior distribution and the likelihood distribution are thought to be aggressive, but the prior distribution is thought to be less aggressive than the likelihood distribution.

\section{Proof of Theorem \captext{\ref{thm:closeness-from-alpha-scaled}}}\label{append:closeness-from-alpha-scaled}
Before proving Theorem \ref{thm:closeness-from-alpha-scaled}, we prepare with the following lemma. 

\begin{lemma}\label{lem:jeffrey-divergence}
If $\alpha > 1$, then
\[
    \displaystyle \int h^{(\alpha)}(\vec \theta) \ln h(\vec \theta) \d \vec \theta - \int h(\vec \theta) \ln h(\vec \theta) \d \vec \theta >0;
\]
if $\alpha < 1$, then
\[
    \displaystyle \int h^{(\alpha)}(\vec \theta) \ln h(\vec \theta) \d \vec \theta - \int h(\vec \theta) \ln h(\vec \theta) \d \vec \theta < 0.
\]
\end{lemma}
\begin{proof}[Proof of Lemma \ref{lem:jeffrey-divergence}]
Let $C_{\alpha} \defeq \int h^{\alpha}(\vec \theta) \d \vec \theta$.
Jeffrey's divergence (NB: it is symmetric) between ${h(\vec \theta)}$ and ${h^{(\alpha)}(\vec \theta)}$ can be given as
\begin{equation*}
    \begin{array}{l}
      \KL{h^{(\alpha)}(\vec \theta)}{h(\vec \theta)} + \KL{h(\vec \theta)}{h^{(\alpha)}(\vec \theta)} \\
      \quad = \displaystyle \int \frac{h^\alpha (\vec \theta)}{C_\alpha} \ln \frac{h^\alpha (\vec \theta)}{C_\alpha \cdot h(\vec \theta)} \d \vec \theta + \displaystyle \int h(\vec \theta) \ln \frac{C_\alpha \cdot h(\vec \theta)}{h^{\alpha}(\vec \theta)} \d \vec \theta \\
      \quad = \displaystyle \frac{1}{C_\alpha} \Big[  \int h^\alpha (\vec \theta) \ln h^\alpha (\vec \theta) \d \vec \theta -

      \displaystyle \int h^\alpha (\vec \theta) \ln C_\alpha \d \vec \theta - \\

      \quad \quad \quad \displaystyle \int h^\alpha (\vec \theta) \ln h(\vec \theta) \d \vec \theta 
      
      \Big] + \Big[
      \displaystyle \int h(\vec \theta) \ln C_\alpha \d \vec \theta  + \\
      
      \quad \quad \quad \displaystyle \int h(\vec \theta) \ln h(\vec \theta)\d \vec \theta - \int h(\vec \theta) \ln {h^{\alpha}(\vec \theta)} \d \vec \theta
      \Big] \\

      \quad = \displaystyle \frac{\alpha - 1}{C_\alpha} \int h^\alpha (\vec \theta) \ln h (\vec \theta) \d \vec \theta - \ln C_{\alpha} + \ln C_{\alpha} +\\
      \quad \quad \quad \displaystyle (1-\alpha) \int h(\vec \theta) \ln h(\vec \theta) \d \vec \theta \\

      \quad = (\alpha - 1) \Big[\displaystyle \displaystyle \int h^{(\alpha)}(\vec \theta) \ln h(\vec \theta) \d \vec \theta - \int h(\vec \theta) \ln h(\vec \theta) \d \vec \theta \Big] \\

      \quad \ge 0.
    \end{array}
\end{equation*}
The equality holds if and only if $\alpha = 1$. This completes the proof.
\end{proof}

Now we proceed to the proof of Theorem \ref{thm:closeness-from-alpha-scaled}.
\begin{proof}[Proof of Theorem \ref{thm:closeness-from-alpha-scaled}]
Let $C_{\alpha} \defeq \int h^{\alpha}(\vec \theta) \d \vec \theta$.
We have 
\begin{equation*}
    \begin{array}{l}
      \KL{h(\vec \theta)}{h^{(\alpha)}(\vec \theta)} \\
      \quad \quad = \displaystyle \int h(\vec \theta) \ln \frac{h(\vec \theta)}{h^{(\alpha)}(\vec \theta)} \d \vec \theta \\
      
      \quad \quad = \displaystyle \int h(\vec \theta) \ln \frac{C_{\alpha} \cdot h(\vec \theta)}{h^{\alpha}(\vec \theta)} \d \vec \theta \\

      \quad \quad = \ln C_{\alpha} + (\alpha - 1) \Ent h(\vec \theta).
    \end{array}
\end{equation*}

Therefore,
\begin{equation*}
    \begin{array}{l}
        \displaystyle \frac{\d \KL{h(\vec \theta)}{h^{(\alpha)}(\vec \theta)}}{\d \alpha} \\
        \quad \quad = \displaystyle \int h^{(\alpha)}(\vec \theta) \ln h(\vec \theta) \d \vec \theta + \Ent h(\vec \theta) \\
        \quad \quad = \displaystyle \int h^{(\alpha)}(\vec \theta) \ln h(\vec \theta) \d \vec \theta - \int h(\vec \theta) \ln h(\vec \theta) \d \vec \theta.
    \end{array}
\end{equation*}
Then, by Lemma \ref{lem:jeffrey-divergence}, the monotonicity is immediate.

In addition,
\begin{equation*}
    \begin{array}{l}
        \displaystyle \frac{\d^2 \KL{h(\vec \theta)}{h^{(\alpha)}(\vec \theta)}}{\d^2 \alpha} \\
        
        \quad = \displaystyle \frac{1}{C^2_\alpha} \Bigg\{\displaystyle \int h^\alpha (\vec \theta) \d \vec \theta \cdot \int h^\alpha (\vec \theta) \ln h(\vec \theta) \ln h(\vec \theta) \d \vec \theta - \\
         
        \quad \quad \left[ \displaystyle \int h^\alpha(\vec \theta) \ln h(\vec \theta) \d \vec \theta \right]^2
        \Bigg\} \\

        \quad \ge 0.
    \end{array}
\end{equation*}
The last inequality is due to the Cauchy--Schwarz inequality; the equality holds if and only if $h(\vec \theta)$ is a uniform distribution; see Appendix \ref{append:alpha-scale-ent-monotonicity}.
This completes the proof.
\end{proof}

\section{Proof of Theorem \ref{thm:benefits-alpha-distributions}}\label{append:benefits-alpha-distributions}
\begin{proof}
Let $C_{\alpha} \defeq \int h^{\alpha}(\vec \theta) \d \vec \theta$. We have
    \begin{equation*}
        \begin{array}{cl}
            f(\alpha) &\defeq \KL{h_0(\vec \theta)}{h(\vec \theta)}  - \KL{h_0(\vec \theta)}{h^{(\alpha)}(\vec \theta)}\\
            &= \displaystyle \int h_0(\vec \theta) \ln \frac{h_0(\vec \theta)}{h(\vec \theta)} \d \vec \theta - \int h_0(\vec \theta) \ln \frac{h_0(\vec \theta) C_{\alpha}}{h^\alpha(\vec \theta)} \d \vec \theta \\
            &= \displaystyle \int h_0(\vec \theta) \ln \frac{h^\alpha(\vec \theta)}{h(\vec \theta) C_{\alpha}} \d \vec \theta \\
            &= (\alpha - 1) \displaystyle \int h_0(\vec \theta) \ln h(\vec \theta) \d \vec \theta - \ln \int h^{\alpha}(\vec \theta) \d \vec \theta.
        \end{array}
    \end{equation*}
    
    Therefore,
    \[
        \begin{array}{cl}
            \displaystyle \frac{\d f(\alpha)}{\d \alpha} &=\displaystyle \int h_0(\vec \theta) \ln h(\vec \theta) \d \vec \theta - \displaystyle \frac{\displaystyle \int h^{\alpha}(\vec \theta) \ln h(\vec \theta) \d \vec \theta}{\displaystyle \int h^{\alpha}(\vec \theta) \d \vec \theta} \\
            &=\displaystyle \int h_0(\vec \theta) \ln h(\vec \theta) \d \vec \theta - \displaystyle \int h^{(\alpha)}(\vec \theta) \ln h(\vec \theta) \d \vec \theta.
        \end{array}
    \]
    As a result,
    \[
        \displaystyle \left.\frac{\d f(\alpha)}{\d \alpha} \right|_{\alpha = 1} = \displaystyle \int [h_0(\vec \theta) - h(\vec \theta)] \ln h(\vec \theta) \d \vec \theta.
    \]
    
    In addition,
    \[
        \begin{array}{l}
        \displaystyle \frac{\d^2 f(\alpha)}{\d \alpha^2} \\
        
        \quad = - \displaystyle \frac{1}{C^2_\alpha} \Bigg\{\displaystyle \int h^\alpha (\vec \theta) \d \vec \theta \cdot \int h^\alpha (\vec \theta) \ln h(\vec \theta) \ln h(\vec \theta) \d \vec \theta - \\
        
        \quad \quad \left[ \displaystyle \int h^\alpha(\vec \theta) \ln h(\vec \theta) \d \vec \theta \right]^2
        \Bigg\} \\

        \quad \le 0,
        \end{array}
    \]
    where the last inequality is due to the Cauchy--Schwarz inequality; the equality holds if and only if $h(\vec \theta)$ is a uniform distribution; see Appendix \ref{append:alpha-scale-ent-monotonicity}.

    Since $h(\vec \theta)$ is not a uniform distribution, we have $\frac{\d^2 f(\alpha)}{\d \alpha^2} < 0$ for all $\alpha \ge 0$. That is, $f(\alpha)$ is strictly concave on $[0, \infty)$. In addition, we have $f(1) = 0$.  Therefore, as long as $\frac{\d f(\alpha)}{\d \alpha}|_{\alpha = 1} \neq 0$, there exists $\alpha \in [0, \infty)$ such that $f(\alpha) > 0$. Note that $f(\alpha)$ is continuous in $\alpha$. This completes the proof.
\end{proof}

\section{Hidden Markov Model}\label{append:hidden-Markov-process}
In this section, we explain the time-sequential Bayesian inference for a hidden Markov model \citep[Chapter~13]{bishop2006pattern}; for an illustration, see Figure \ref{fig:hidden-Markov-process}. 
\begin{figure}[!htbp]
    \centering
    \includegraphics[height=2.8cm]{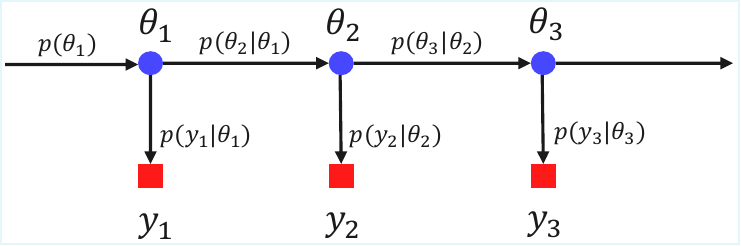}
    \caption{Bayesian inference for a Hidden Markov Model.}
    \label{fig:hidden-Markov-process}
\end{figure}

Let $\{\vec \theta_i\}$ and $\{\rvec y_i\}$ denote two stochastic processes indexed by the discrete time $i$. The process $\{\vec \theta_i\}$ is Markovian, however, not directly observable (i.e., hidden). The process $\{\rvec y_i\}$ is directly observable. At each time $i$, we aim to estimate the hidden quantity $\vec \theta_i$ using the historical observations $\{\vec y_i\}$, up to and including time $i$, and the transition kernel $p(\vec \theta_{i} | \vec \theta_{i-1})$.  

The process $\{\vec \theta_i\}$ is Markovian because at every time $i \ge 2$, the hidden quantity $\vec \theta_i$ is probabilistically determined only by $\vec \theta_{i-1}$ through $p(\vec \theta_{i} | \vec \theta_{i-1})$. The \bfit{direct} measurement model of $\vec \theta_i$ is $p(\vec y_i | \vec \theta_i)$. 

The hidden Markov model has wide applications in real-world statistical signal processing, for instance, state estimation for dynamical systems \citep{simon2006optimal}; representative algorithms encompass the Kalman filter, the particle filter, etc. In state estimation, a hidden Markov model in Figure \ref{fig:hidden-Markov-process} is mathematically described by a state-space model as follows:
\begin{equation}\label{eq:state-space-model}
\left\{\begin{array}{rllll}
\vec \theta_{i} &=  \vec f_{i-1} (\vec \theta_{i-1}) + \rvec w_{i-1},  \\
\rvec y_{i}   &=  \vec h_{i} (\vec \theta_{i}) + \rvec v_{i},
\end{array}\right.
\end{equation}
where $\vec \theta_i$ is the (hidden) state vector at time $i$ and $\vec y_i$ the measurement vector of $\vec \theta_i$; $\rvec w_{i-1}$ and $\rvec v_{i}$ are state transition and state measurement noises, respectively; $\vec f_{i-1}$ and $\vec h_i$ are deterministic functions that determine the state transition and state measurement relations, respectively. In typical settings of state estimation problems, we assume that $\rvec w_{i-1} \sim \cal N(\vec 0, \mat Q)$ and $\rvec v_{i} \sim \cal N(\vec 0, \mat R)$ for some known noise covariances $\mat Q$ and $\mat R$. When the functions $\vec f_{i-1}$ and $\vec h_i$ are linear, we have the linear state-space model
\begin{equation}\label{eq:linear-state-space-model}
\left\{\begin{array}{rllll}
\vec \theta_{i} &=  \mat F_{i-1} \vec \theta_{i-1} + \rvec w_{i-1},  \\
\rvec y_{i}   &=  \mat H_{i} \vec \theta_{i} + \rvec v_{i},
\end{array}\right.
\end{equation}
for some state transition and state measurement matrices $\mat F_{i-1}$ and $\mat H_i$, respectively.

The mathematical foundation of this time-sequential Bayesian inference problem is \eqref{eq:sequential-data}; see the texts around the equation for clarifications. The stochastic interpretation of Bayes' rule applies to the inference of this hidden Markov model since there physically exists a joint data-generating distribution. Specifically, the hidden quantity $\vec \theta_i$ is indeed sampled from the distribution $p(\vec \theta_{i} | \vec \theta_{i-1})$; for $i=1$, $\vec \theta_1$ is sampled from $p(\vec \theta_{1})$.

The optimal state estimation method for the general nonlinear state-space model \eqref{eq:state-space-model} is the particle filter, while for the linear state-space model \eqref{eq:linear-state-space-model} is the Kalman filter; see \cite{simon2006optimal}. Both the particle filter and the Kalman filter are derived from the generic principle \eqref{eq:sequential-data}. Certainly, such optimality depends on the exact knowledge of the state-space models \eqref{eq:state-space-model} and \eqref{eq:linear-state-space-model}.

When the state-space models \eqref{eq:state-space-model} and \eqref{eq:linear-state-space-model} are uncertain, experimental results of the proposed $(\alpha, \beta)$-posterior, against the existing $\alpha$-posterior and Bayesian posterior, are conducted in Appendix \ref{sec:append-applications} (see Appendix \ref{append:experiment-state-estimation}: Bayesian Estimation: State Estimation). Particularly, when the state-transition equation is uncertain, we modify the prior distribution $p(\vec \theta_i | \cal Y_{i-1})$; when the state-measurement equation is uncertain, we modify the likelihood distribution $l_{\vec y_i}(\vec \theta_i)$. As indicated, \bfit{independently} manipulating the prior and likelihood distributions is \bfit{physically natural and straightforward}.

\section{Concrete Applications and Experiments}\label{sec:append-applications}
The UA Bayes' rule \eqref{eq:alpha-beta-posterior} can give birth to several UA Bayesian machine learning methods and UA Bayesian signal processing methods. As demonstrations and without loss of generality, in this section, we focus on naive Bayes MAP classification problems in machine learning, and state estimation (i.e., state filtering) problems of dynamic stochastic state-space systems in signal processing. 
In particular, the UA Bayes MAP classifier, the UA Kalman filter, the UA particle filter, and the UA interactive-multiple-model filter are suggested and experimentally validated.
The main purpose is to show the existence of $(\alpha, \beta)$ such that the $(\alpha, \beta)$-posterior can outperform the conventional Bayesian posterior; cf. Philosophy \ref{phi:robustness}. In addition, we aim to show that the $\alpha$-posterior is not sufficient for some real-world applications, and therefore, the $(\alpha, \beta)$-posterior is necessary. The technical basics on problem formulation and parameter tuning of Bayesian classification and estimation can be seen in Subsections \ref{subsec:applications} and \ref{subsec:parameter-tuning}, respectively. 

This section builds on standard domain-specific setups and concepts, including priors and likelihoods in multinomial and Gaussian naive Bayes classifiers [see \citet{scikitlearnnaivebayes}], as well as those used in hidden-state estimation (see Appendix \ref{append:hidden-Markov-process}). For readers not interested in these application areas, this section may be skipped without missing the key points of the paper; focusing on illustrative examples and experiments in the main body of the paper is sufficient.

\subsection{Bayesian MAP Classification: Text and Image Classification}
We consider two real-world classification problems. One is a natural-language text classification problem, while the other is a medical image classification problem.
\subsubsection{IMDB Dataset}
We investigate the performance of the uncertainty-aware (UA) Bayes MAP classifier in \eqref{eq:generalized-bayes-classifier-elite} on the IMDB dataset containing 
movie reviews \cite[Section~4.3.2]{maas2011learning}. This problem is also known as text-based sentiment analysis where we determine whether a review of a movie is positive or negative. The multinomial naive Bayes MAP classification algorithm with Laplacian smoothing is employed; that is, the features are distributed according to multinomial distributions because movie reviews are typically represented as word vector counts, and the features are assumed to be conditionally independent given the target class;  see \cite{scikitlearnnaivebayes}. 

For every sample size $L$ in \{50, 100, 250, 500, 1000, 2000, 5000, 10000\}, we conduct $500$ independent Monte--Carlo tests. In detail, in each Monte--Carlo trial, we randomly select $L$ samples from the IMDB dataset, of which $80\%$ are training samples and $20\%$ are testing samples. For each Monte--Carlo test, the radial-basis-function surrogate optimization on $[0, 1]$ is employed to find the empirically optimal hyper-parameter $\lambda$ in \eqref{eq:generalized-bayes-classifier-elite};\footnote{See MATLAB's $\mathsf{surrogateopt}$ function:  \url{www.mathworks.com/help/gads/surrogateopt.html}.} the testing performance, which is a function of $\lambda$, serves as the loss function in \eqref{eq:parameter-tuning} to be minimized. The surrogate-optimization-based search process starts with $\lambda = 0.5$ (i.e., the conventional Bayes MAP classifier) so that the performance of the uncertainty-aware Bayes MAP classifier \eqref{eq:generalized-bayes-classifier-elite} is at least as good as that of the conventional Bayes MAP classifier. The experimental results are shown in Figure \ref{fig:IMDB-classification}.

\begin{figure}[!htbp]
    \centering
    \subfigure[Averaged Accuracy]{
        \begin{minipage}[htbp]{0.3\linewidth}
            \centering
            \includegraphics[height=3.4cm]{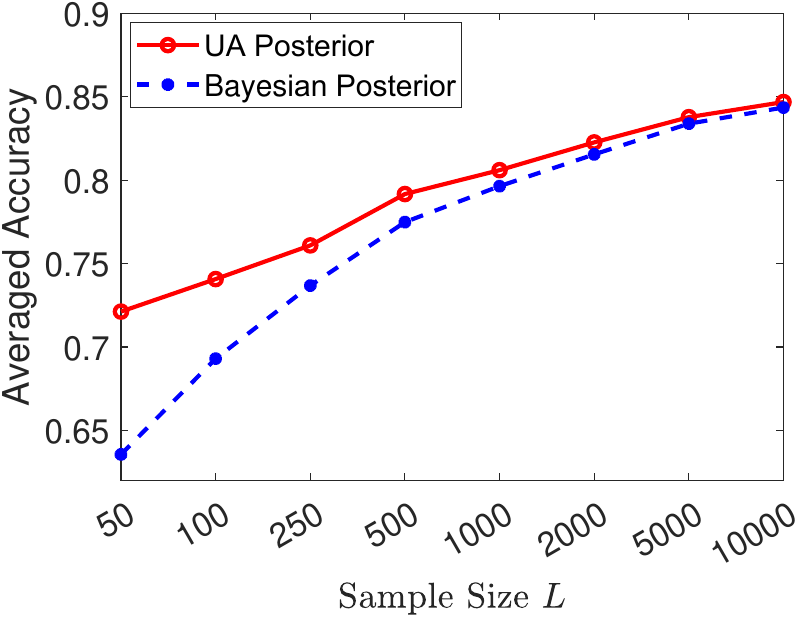}
        \end{minipage}
        \label{fig:IMDB-a}
    }
    \subfigure[Grid Search]{
        \begin{minipage}[htbp]{0.3\linewidth}
            \centering
            \includegraphics[height=3.4cm]{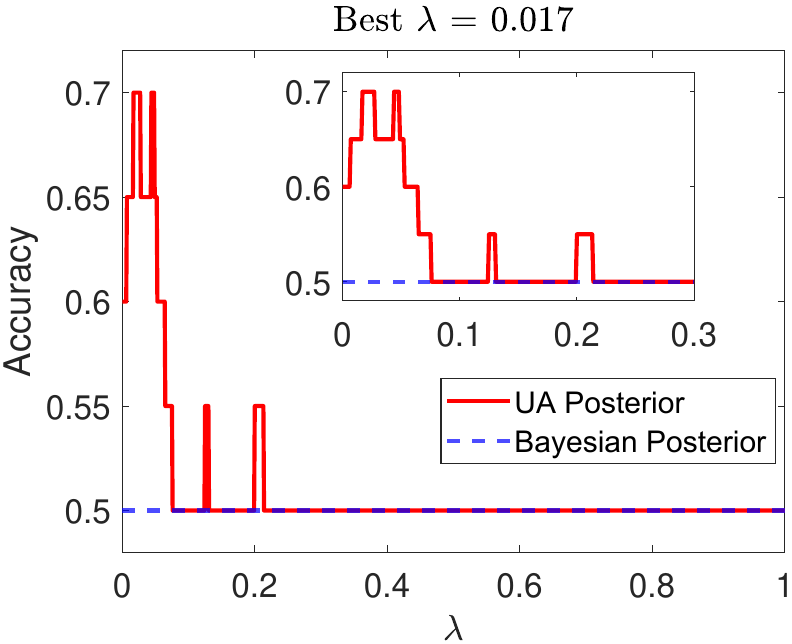}
        \end{minipage}
        \label{fig:IMDB-b}
    }
    \subfigure[Surrogate Optimization]{
        \begin{minipage}[htbp]{0.3\linewidth}
            \centering
            \includegraphics[height=3.4cm]{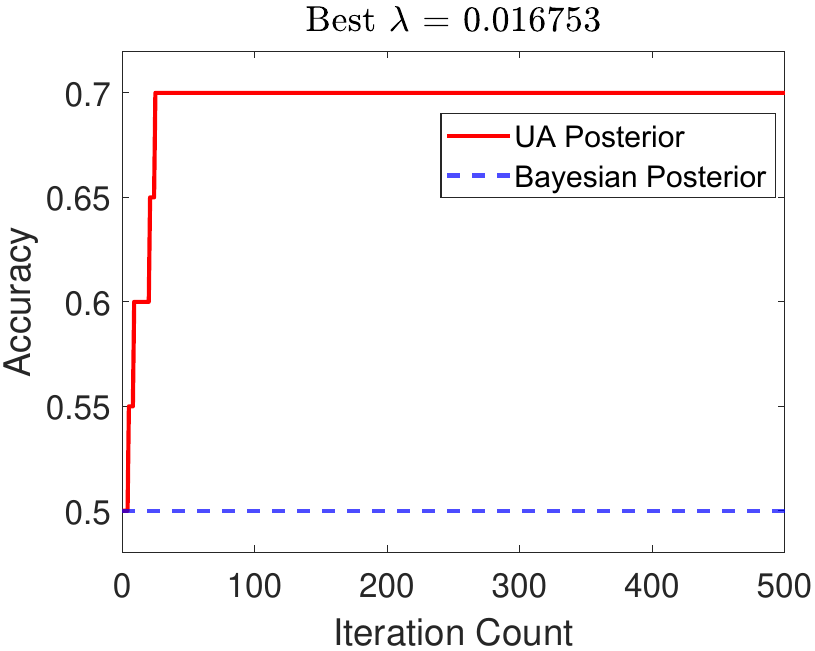}
        \end{minipage}
        \label{fig:IMDB-c}
    }

    \caption{Experimental results on the IMDB data set. (a): The classification accuracy against sample size $L$, averaged over $500$ Monte--Carlo episodes. (b): The grid search for $\lambda$ with step size of $0.001$ in a Monte--Carlo trial when $L = 100$. (c): The visual illustration of the surrogate optimization process for $\lambda$ on $[0, 1]$; the found optimal value is $\lambda = 0.016753$.}
    \label{fig:IMDB-classification}
\end{figure}

As sample size $L$ becomes larger, the nominal distributions $p(\theta)$ and $l_{\vec y}(\theta)$ get closer to the true distributions  $p_0(\theta)$ and $l_{0, \vec y}( \theta)$, respectively. Therefore, the UA Bayes MAP classifier tends to obtain the best classification accuracy as $L$ increases; cf. Figure \ref{fig:IMDB-a}. In Figures \ref{fig:IMDB-b} and \ref{fig:IMDB-c}, the parameter tuning process is visualized for a Monte--Carlo trial when $L = 100$. Since the found optimal $\lambda = 0.016753$, any hyper-parameter pair of $(\alpha, \beta)$ in \eqref{eq:generalized-bayes-classifier} such that $\alpha/(\alpha + \beta) = 0.016753$ is optimal. Both two tuning methods find the best accuracy of $0.7$. 
The average running times of the two tuning methods, against sample size $L$, are shown in Figure \ref{fig:IMDB-running-times}. Note that the larger the testing sample size, the more running time is needed to calculate the empirical performance \eqref{eq:bayes-classifier-cost-empirical}.

Since not all real-world applications allow sufficient sample sizes (due to, e.g., the difficulty in data collection), performance improvement under small sample sizes is still practically meaningful.

\begin{figure}[!htbp]
    \centering
    \includegraphics[height=3.4cm]{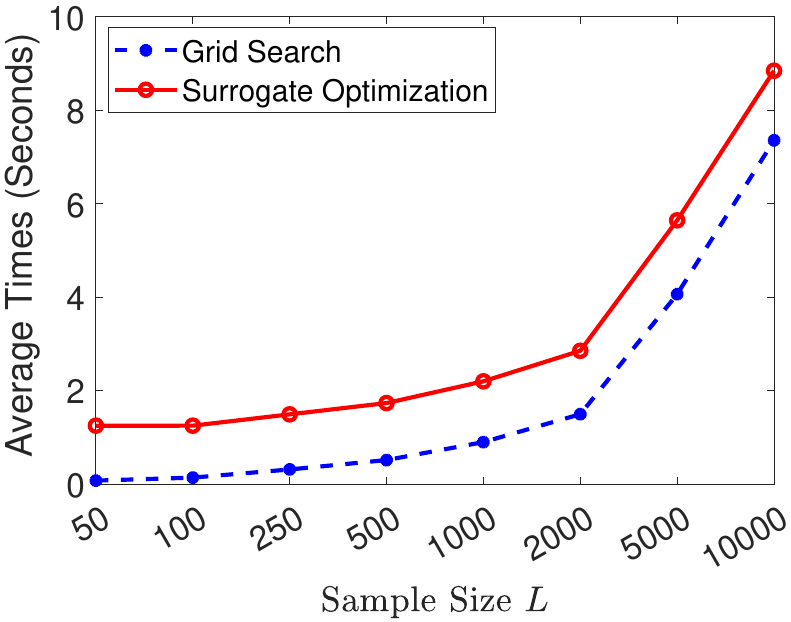}
    \caption{Average running times of the two tuning methods against sample size $L$. The time difference of the two tuning methods remains about 1.5 seconds for every $L$. The running time of grid search depends on the step size (NB: the smaller the step size, the larger the maximum accuracy, but the more the running times), while that of surrogate optimization depends on the maximum allowed iteration count (NB: the larger the iteration count, the larger the maximum accuracy, but the more the running times).}
    \label{fig:IMDB-running-times}
\end{figure}

\subsubsection{UCI Breast Cancer Dataset}
We examine the performance of the UA Bayes MAP classifier in \eqref{eq:generalized-bayes-classifier-elite} on the UCI Breast Cancer dataset \citep{wolberg1995breast}. This is an image-classification-based medical diagnosis problem. The Gaussian naive Bayes MAP classifier is employed; that is, the features are distributed according to Gaussian distributions, and the features are assumed to be conditionally independent given the target class;  see \cite{scikitlearnnaivebayes}. 

Since there are only $569$ instances in this data set, we use all of them in the experiment. We conduct $500$ independent Monte--Carlo tests. In each Monte--Carlo trial, $80\%$ of the $569$ instances are randomly drawn to serve as training samples, while the remaining $20\%$ are testing samples. The radial-basis-function surrogate optimization on $[0, 1]$ is employed to find the empirically optimal hyper-parameter $\lambda$ in \eqref{eq:generalized-bayes-classifier-elite}; the testing performance, which is a function of $\lambda$, serves as the loss function in \eqref{eq:parameter-tuning} to be minimized. The surrogate-optimization-based search process starts with $\lambda = 0.5$ (i.e., the conventional Bayes MAP classifier) so that the performance of the UA Bayes MAP classifier \eqref{eq:generalized-bayes-classifier-elite} is at least as good as that of the conventional Bayes MAP classifier. The experimental results are shown in Figure \ref{fig:UCI-a}. In a Monte--Carlo trial, both parameter tuning methods find the best accuracy of $0.9027$; see Figures \ref{fig:UCI-b} and \ref{fig:UCI-c}.

\textbf{Running Times}. The average running time for the grid search method is 0.2 seconds and for the surrogate optimization method is 1.5 seconds.

\begin{figure}[!htbp]
    \centering
    \subfigure[Boxplots of the Accuracies]{
        \begin{minipage}[htbp]{0.3\linewidth}
            \centering
            \includegraphics[height=3.2cm]{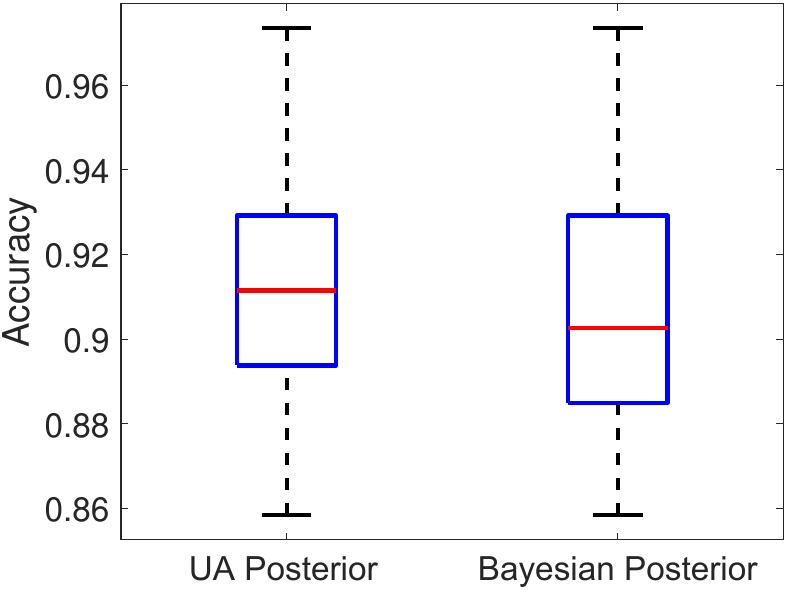}
        \end{minipage}
        \label{fig:UCI-a}
    }
    \subfigure[Grid Search]{
        \begin{minipage}[htbp]{0.3\linewidth}
            \centering
            \includegraphics[height=3.2cm]{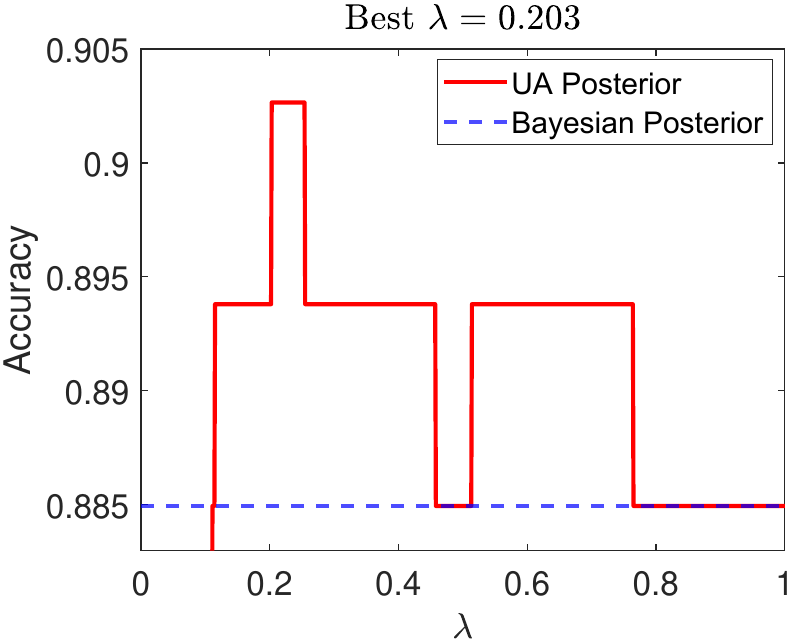}
        \end{minipage}
        \label{fig:UCI-b}
    }
    \subfigure[Surrogate Optimization]{
        \begin{minipage}[htbp]{0.3\linewidth}
            \centering
            \includegraphics[height=3.2cm]{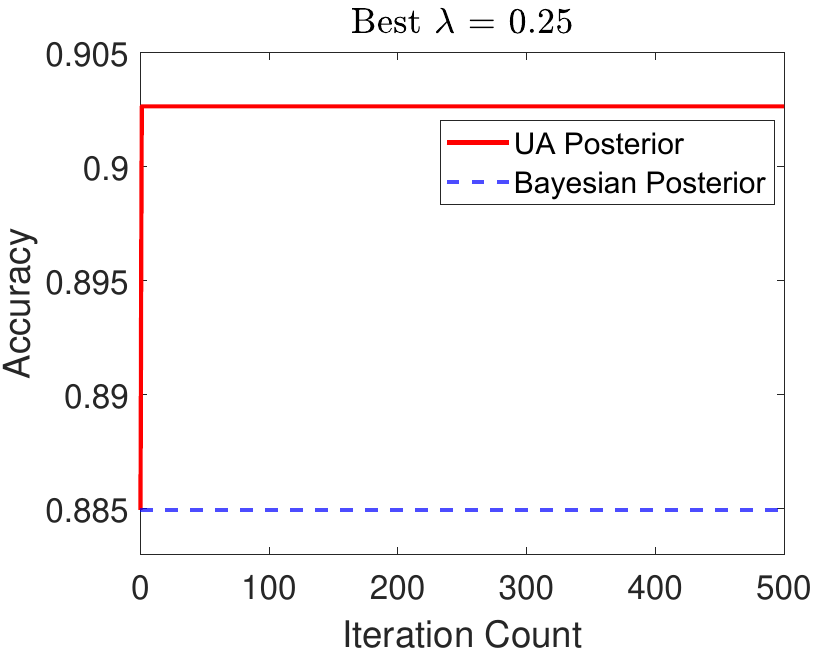}
        \end{minipage}
        \label{fig:UCI-c}
    }

    \caption{Experimental results on the UCI Breast Cancer data set. (a): The boxplots for the classification accuracies of $500$ Monte--Carlo episodes. (b): The grid search for $\lambda$ with step size of $0.001$ in a Monte--Carlo trial. (c): The visual illustration of the surrogate optimization process for $\lambda$ on $[0, 1]$; the found optimal value is $\lambda = 0.25$.}
    \label{fig:UCI-classification}
\end{figure}

\subsection{Bayesian Estimation: State Estimation}\label{append:experiment-state-estimation}
From the perspective of applied statistics, state estimation problems can be interpreted as sequential Bayesian inference. Briefly speaking, state estimation aims to estimate the unknown (i.e., unobservable) state vector $\rvec x_k$ at the discrete time step $k$ using the known (i.e., observable) measurement set $\cal Y_k \defeq (\rvec y_1, \rvec y_2, \ldots, \rvec y_k)$ and the probabilistic model $p_{\rvec x_k, \cal Y_k}(\vec x_k, \mat Y_k)$; the probabilistic model $p_{\rvec x_k, \cal Y_k}(\vec x_k, \mat Y_k)$ is induced by the stochastic state-space models. The aim of state estimation is to obtain the posterior state distribution $p_{\rvec x_k|\cal Y_k}(\vec x_k|\mat Y_k)$ or the posterior mean $\E(\rvec x_k|\cal Y_k)$. 

\subsubsection{Uncertainty-Aware Kalman Filter}
We consider the state estimation problem of linear Gaussian state-space models \cite[Chap.~3]{anderson1979optimal}
\begin{equation}\label{eq:linear}
\left\{\begin{array}{rllll}
\rvec x_{k} &=  \mat F_{k-1} \rvec x_{k-1} + \mat G_{k-1} \rvec w_{k-1},  \\
\rvec y_{k}   &=  \mat H_{k} \rvec x_{k} + \rvec v_{k},
\end{array}\right.
\end{equation}
where, for every time $k$, the system matrices $\mat F_{k}$, $\mat G_{k}$, and $\mat H_{k}$ are assumed to be known. The process noise vector and the measurement noise vector are denoted by $\rvec w_{k} \sim \cal N(\vec 0, \mat Q_k)$ and $\rvec v_{k}\sim \cal N(\vec 0, \mat R_k)$, respectively, and $\mat Q_k$ and $\mat R_k$ are assumed to be known, for every $k$. For this linear Gaussian system, under several uncorrelatedness assumptions among $\rvec x_0$, $\{\rvec w_k\}_{\forall k}$, and $\{\rvec v_k\}_{\forall k}$ (see \citet[p.~38]{anderson1979optimal}), the Kalman filter is the optimal state estimator in the sense of minimum mean-squared error.

Suppose that at time $k-1$, the posterior state distribution is 
\[
\hat{\rvec x}_{k-1|k-1} \sim \cal N(\hat{\vec x}_{k-1|k-1}, \mat P_{k-1|k-1}).
\]
At time $k$, the prior state distribution is 
\[
\hat{\rvec x}_{k|k-1} \sim \cal N(\hat{\vec x}_{k|k-1}, \mat P_{k|k-1})
\]
where $\mat P_{k|k-1} \defeq \mat F_{k-1} \mat P_{k-1|k-1} \mat F^\top_{k-1} + \mat G_{k-1} \mat Q_{k-1} \mat G^\top_{k-1}$ and $\hat{\vec x}_{k|k-1} \defeq \mat F_{k-1} \hat{\vec x}_{k-1|k-1}$; the measurement distribution conditioned on state $\rvec x_k$ is 
\[
\rvec y_k | \rvec x_k \sim \cal N(\mat H_k \rvec x_k, \mat R_k).
\]
Therefore, by applying the $(\alpha, \beta)$-posterior rule \eqref{eq:alpha-beta-posterior}, the prior state distribution should be modified to
\[
\hat{\rvec x}_{k|k-1} \sim \cal N(\hat{\vec x}_{k|k-1}, ~\mat P_{k|k-1}/\beta)
\]
and the conditional measurement distribution should be modified to
\[
\rvec y_k | \rvec x_k \sim \cal N(\mat H_k \rvec x_k, ~\mat R_k/\alpha).
\]
Note that for a Gaussian distribution $\cal N(\vec \mu, \mat \Sigma)$, the $\alpha$-scaled version is equal to $\cal N(\vec \mu, \mat \Sigma/\alpha)$.
As a result, the $(\alpha, \beta)$-uncertainty-aware Kalman filter can be accordingly obtained by just applying the following two assignment operations in each Kalman iteration: $\mat P_{k|k-1} \leftarrow \mat P_{k|k-1}/\beta$ and $\mat R_k \leftarrow \mat R_k/\alpha$, for $\alpha,\beta > 0$. The former operation admits that the state-transition equation is uncertain (thus to be compensated by $\beta$), while the latter admits that the state-measurement equation is uncertain (thus to be compensated by $\alpha$).

This modification is reminiscent of the distributionally robust state estimation (DRSE) method proposed in \citet{wang2021robust}; see also \citet[p.~22,~p.~53]{wang2022thesisdistributionally} for intuitive interpretations. Hence, when $0< \alpha,\beta \le 1$, the $(\alpha, \beta)$-uncertainty-aware Kalman filter is equivalent to distributionally robust state estimation methods in \cite{wang2021robust,wang2021distributionally}, provided that there are no outliers in measurements. Since $0< \alpha,\beta \le 1$, $\mat P_{k|k-1}$ and $\mat R_k$ are assumed to be overly aggressive, and therefore, they need to be inflated. When $\alpha,\beta \ge 1$, the values of $\mat P_{k|k-1}$ and $\mat R_k$ are reduced, which implies that $\mat P_{k|k-1}$ and $\mat R_k$ are assumed to be overly conservative. However, from the technical derivations of the DRSE method, the value reduction of $\mat P_{k|k-1}$ and $\mat R_k$ cannot be achieved in the DRSE method. In this sense, therefore, the $(\alpha, \beta)$-uncertainty-aware Kalman filter presented in this paper generalizes the DRSE method in \cite{wang2021robust,wang2021distributionally} when there are no measurement outliers: the former can handle not only aggressive cases (by inflating covariances) but also conservative cases (by reducing covariances), while the latter can only address aggressive cases. (NB: To combat aggressiveness is to introduce robustness/conservativeness; i.e., robust methods innately cannot fight against conservativeness.)

The experiments of the DRSE method in \citet[Chap.~2]{wang2022thesisdistributionally} can sufficiently support the practical values of the $(\alpha, \beta)$-uncertainty-aware Kalman filter; just notice the mathematical equivalence to the DRSE method (when $0 < \alpha,\beta \le 1$). We do not repeat these experiments here.

\subsubsection{Uncertainty-Aware Particle Filter}
The particle filter is standard to handle the state estimation problem of nonlinear systems \citep{arulampalam2002tutorial}. By employing the $(\alpha, \beta)$-posterior in \eqref{eq:alpha-beta-posterior}, the $(\alpha, \beta)$-uncertainty-aware particle filter can be straightforwardly obtained; see Example \ref{ex:UA-PF}.

As an illustration, we specifically consider the state estimation problem of a 1-dimensional nonlinear system model \citep{carlin1992monte,arulampalam2002tutorial,wang2022distributionally}
\begin{equation}\label{eq:nonlinear-syst}
\left\{
\begin{array}{cl}
    \rscl x_k &= \displaystyle \frac{\rscl x_{k-1}}{2} + \displaystyle \frac{25 \rscl x_{k-1}}{1 + \rscl x_{k-1}^{2}}+8 \cos (1.2k) + \rscl w_{k-1}, \\
    \rscl y_k &= \displaystyle \frac{\rscl x^2_k}{20} + 0.5 \sin(\rscl x_k) + \rscl v_k,
\end{array}
\right.
\end{equation}
where for every $k$, $\rscl w_k \sim \cal N(0, 10)$ and $\rscl v_k \sim \cal N(0, 1)$; for every $k_1$ and $k_2$, $\rscl w_{k_1}$ and $\rscl w_{k_2}$ are uncorrelated, so are $\rscl v_{k_1}$ and $\rscl v_{k_2}$; for every $k$, $\rscl w_{k}$ and $\rscl v_{k}$ are uncorrelated. As in \cite{wang2022distributionally}, we assume that the \textit{nominal} measurement equation  is
\[
    \rscl y_k = \displaystyle \frac{\rscl x^2_k}{20} + \rscl v_k,
\]
which is slightly different from the true one in \eqref{eq:nonlinear-syst}. Therefore, in this case, there is a modeling uncertainty in the measurement equation, i.e., in likelihood distributions at every $k$. As a result, to combat this model uncertainty and improve the filtering accuracy, we should use $\beta = 1$ and $\alpha < 1$.

For the purpose of experimental demonstration, we use $50$, $100$, and $200$ particles, respectively, to report the performance of the $(\alpha,1)$-uncertainty-aware particle filter. The systematic resampling method is adopted to address particle degeneracy, and the effective sample size is set to half the number of particles. Given the number of particles, the experiment is independently conducted with $500$ episodes and each episode contains $100$ time steps. The performance measure, in every episode, is the rooted time-averaged mean-squared error (RTAMSE) along $100$ time steps, i.e.,
$
\sqrt{
\frac{1}{100} \sum_{k=1}^{100} (x_{k} - \hat x_{k})^2
}
$, 
where $\hat x_{k}$  denotes the estimate of the true value $x_{k}$; the overall performance measure for a given number of particles is the averaged RTAMSEs of the $500$ episodes. The filtering result of the $(\alpha,1)$-uncertainty-aware particle filter is shown in Figure \ref{fig:UAPF}, plotted against the value of $\alpha$.
\begin{figure}[!htbp]
    \centering
    \subfigure[50 Particles]{
        \begin{minipage}[htbp]{0.3\linewidth}
            \centering
            \includegraphics[height=3.4cm]{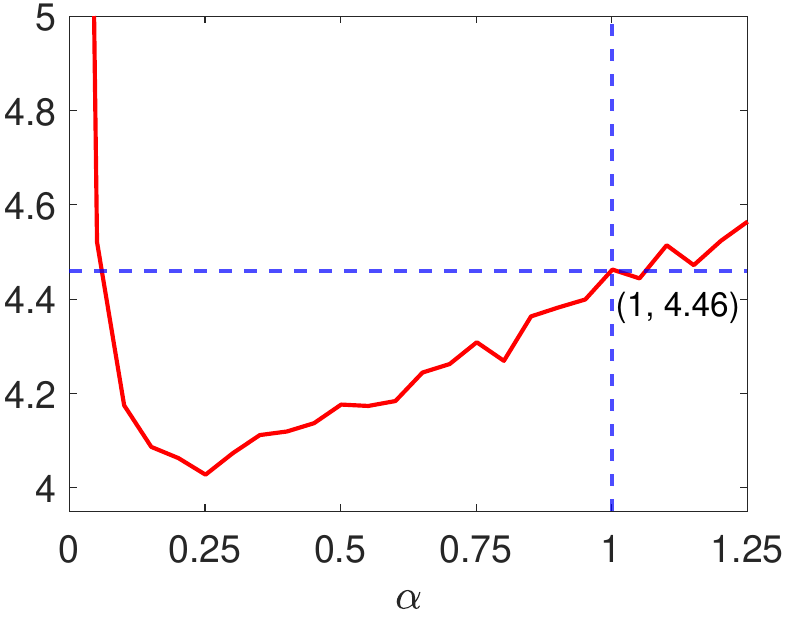}
        \end{minipage}
    }
    \subfigure[100 Particles]{
        \begin{minipage}[htbp]{0.3\linewidth}
            \centering
            \includegraphics[height=3.4cm]{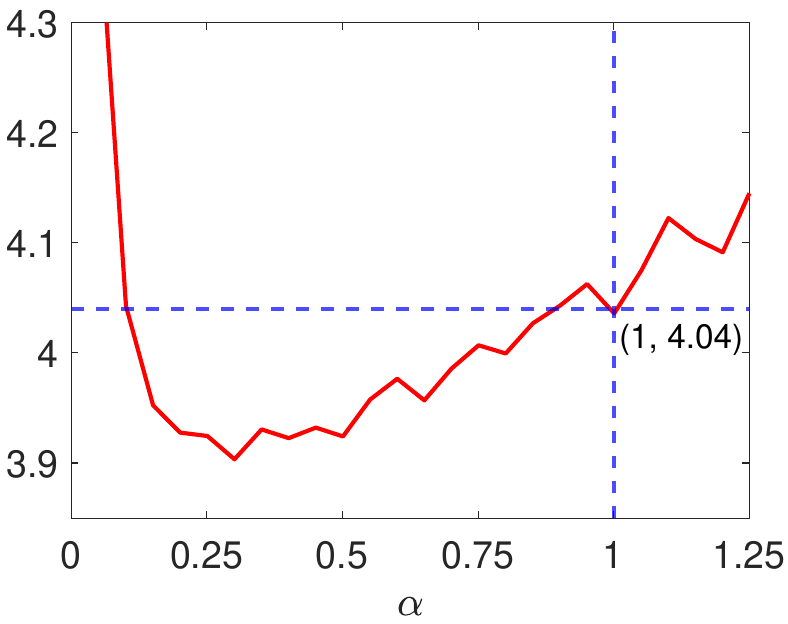}
        \end{minipage}
    }
    \subfigure[200 Particles]{
        \begin{minipage}[htbp]{0.3\linewidth}
            \centering
            \includegraphics[height=3.4cm]{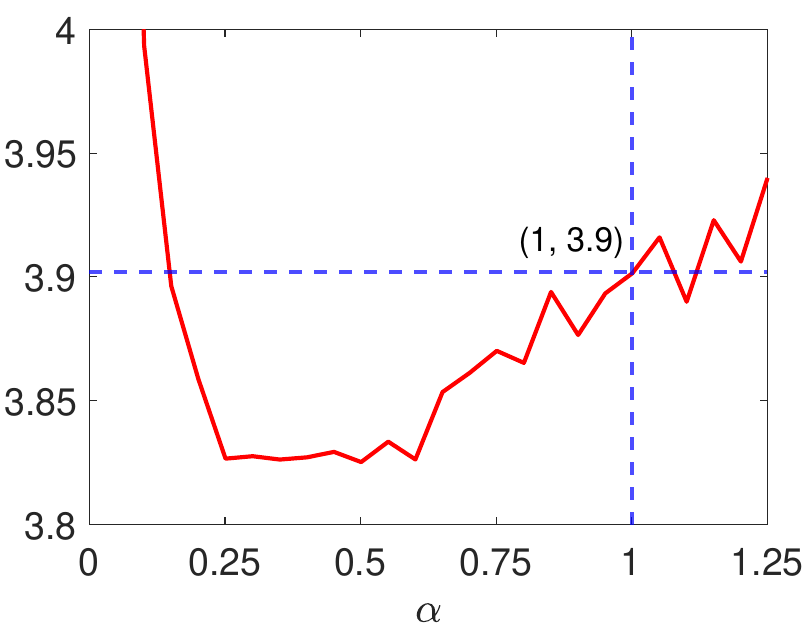}
        \end{minipage}
    }

    \caption{Averaged-RTAMSE performance of the $(\alpha,1)$-uncertainty-aware particle filter, against the value of $\alpha$. The larger the number of particles, the smaller the averaged RTAMSE. For a given number of particles, the $(\alpha,1)$-uncertainty-aware particle filter with $\alpha \approx 0.25$ works best.}
    \label{fig:UAPF}
\end{figure}

As we can see from Figure \ref{fig:UAPF}, when $\alpha < 1$ but $\alpha$ is not close to zero, the $(\alpha,1)$-uncertainty-aware particle filter can outperform the conventional particle filter: This is because when there exists model uncertainty in the measurement equation, the likelihood distribution of particles at each time step tends to be unreliable, and therefore, we need to reduce the concentration (i.e., reduce the trust level and improve the entropy) of the likelihood distribution to cope with the uncertainty. In addition, the results suggest that the superiority of the $(\alpha,1)$-uncertainty-aware particle filter tends to be more significant as the number of particles decreases: This implies that the $(\alpha,1)$-uncertainty-aware particle filter has the innate ability to fight against the particle degeneracy issue. When $\alpha$ is overly small (i.e., close to zero), the $(\alpha,1)$-uncertainty-aware particle filter almost ignores information from the measurements, and therefore, the filter diverges: This indicates that uncertain information is at least better than no information.

\subsubsection{Uncertainty-Aware Interactive Multiple Model Filter}
The interactive multiple model (IMM) filter is standard to handle the state estimation problem of jump linear systems \citep{blom1988interacting,bar2005imm}. By employing the $(\alpha, \beta)$-posterior in \eqref{eq:alpha-beta-posterior}, the $(\alpha, \beta)$-uncertainty-aware IMM filter can be straightforwardly obtained. One may just imagine the models' prior weights as a prior distribution and the models' likelihoods (evaluated at a given measurement) as a likelihood distribution; the aim is to infer the models' posterior distribution and then compute the posterior state estimate; see Example \ref{ex:Bayesian-model-averaging}.

As an illustration, we consider two real-world one-dimensional multi-model target tracking problems; as claimed in \citet{li2003survey}, focusing on only one coordinate does not lose the generality because the motions of a dynamic target in different coordinates can be independently tracked. Mathematically, it is a state estimation problem of a jump linear system model \citep{jilkov2004online,zhao2016recursive,wang2023distributionally}
\[
\left\{
\begin{array}{l}
\rvec x_{k} 
 = 
\left[
\begin{array}{cc}
    1 & T \\
    0 & 1
\end{array}
\right]
\rvec x_{k-1} 
+
\left[
\begin{array}{c}
     {T^2}/{2}  \\
     T
\end{array}
\right]
[\rscl a_{\rscl j_{k-1},k-1} + \rscl w_{k-1}] \\
\rscl y_k = \rscl p_k + \rscl v_{k},
\end{array}
\right.
\]
where the state vector $\rvec x_k$ is defined as
$\rvec x_k \defeq [\rscl p_{k}, \rscl s_{k}]^\top$; 
$\rscl p_k \in \R$ denotes the position, $\rscl s_k \in \R$ the speed, and $\rscl a_{\rscl j_k, k} \in \R$ the maneuvering acceleration of a moving target at time $k$; the positive-integer-valued discrete random variable $\rscl j_k$ denotes the system's operating mode at time $k$; $T > 0$ is the sampling time between two discrete time indices; $\rscl w_k \in \R$ is the acceleration modeling noise and $\rscl v_k \in \R$ the sensor's observation noise at time $k$. At time $k$, the maneuvering acceleration may randomly take any one of the following values
\begin{equation}\label{eq:IMM-model}
\rscl a_{j_k,k} = \left\{
\begin{array}{ll}
   0,  &  j_k = 1, \\
   10,  & j_k = 2, \\
   -10, & j_k = 3,
\end{array}
\right.
\end{equation}
i.e., the random variable $\rscl j_k$ randomly jumps; the diagonal elements of the transition probability matrix are set to $0.8$s and the non-diagonal ones are set to $0.1$s. Therefore, in state estimation for this jump linear system, at every time $k$, we need to jointly estimate the unknown state $\rvec x_k$ and the system's unknown operating mode $\rscl j_k$ based on the past measurements $\{y_1, y_2, \ldots, y_k\}$. 

We reuse the real-world data and experimental setups from \cite{wang2023distributionally}: In short, data from usual GPS are to be processed to obtain higher-accuracy target positions and velocities in real time, while data from RTK serve as ground truth. The RTAMSE, for the total $K$ time steps, is computed as
$
\sqrt{\frac{1}{K} \sum^{K}_{k = 1} (p_k - \hat p_k)^2 + (s_k - \hat s_k)^2}, 
$
where $p_k$ (resp. $s_k$) denotes the true value of position (resp. velocity) at the time $k$, and $\hat p_k$ (resp. $\hat s_k$) denotes its estimate. Note that $(p_k, s_k)$ for every $k$ is provided by RTK. (Some authors prefer to independently report the RTAMSEs of position estimates $\{\hat p_k\}_{\forall k \in [K]}$ and velocity estimates $\{\hat s_k\}_{\forall k \in [K]}$, respectively, because the two dimensions have different numerical scales. However, the author's experiments suggest that it introduces no essential influence on the main claims of this paper. One may use the shared source codes to verify this point.)

\textbf{Track A Slowly-Maneuvering Car}: 
We first track a slowly-maneuvering car that travels on a road in Beijing, China. The car and its trajectory are shown in Figure \ref{fig:car}.
\begin{figure}[!htbp]
    \centering
    \subfigure[car]{
        \begin{minipage}[htbp]{0.46\linewidth}
            \centering
            \includegraphics[height=3.1cm]{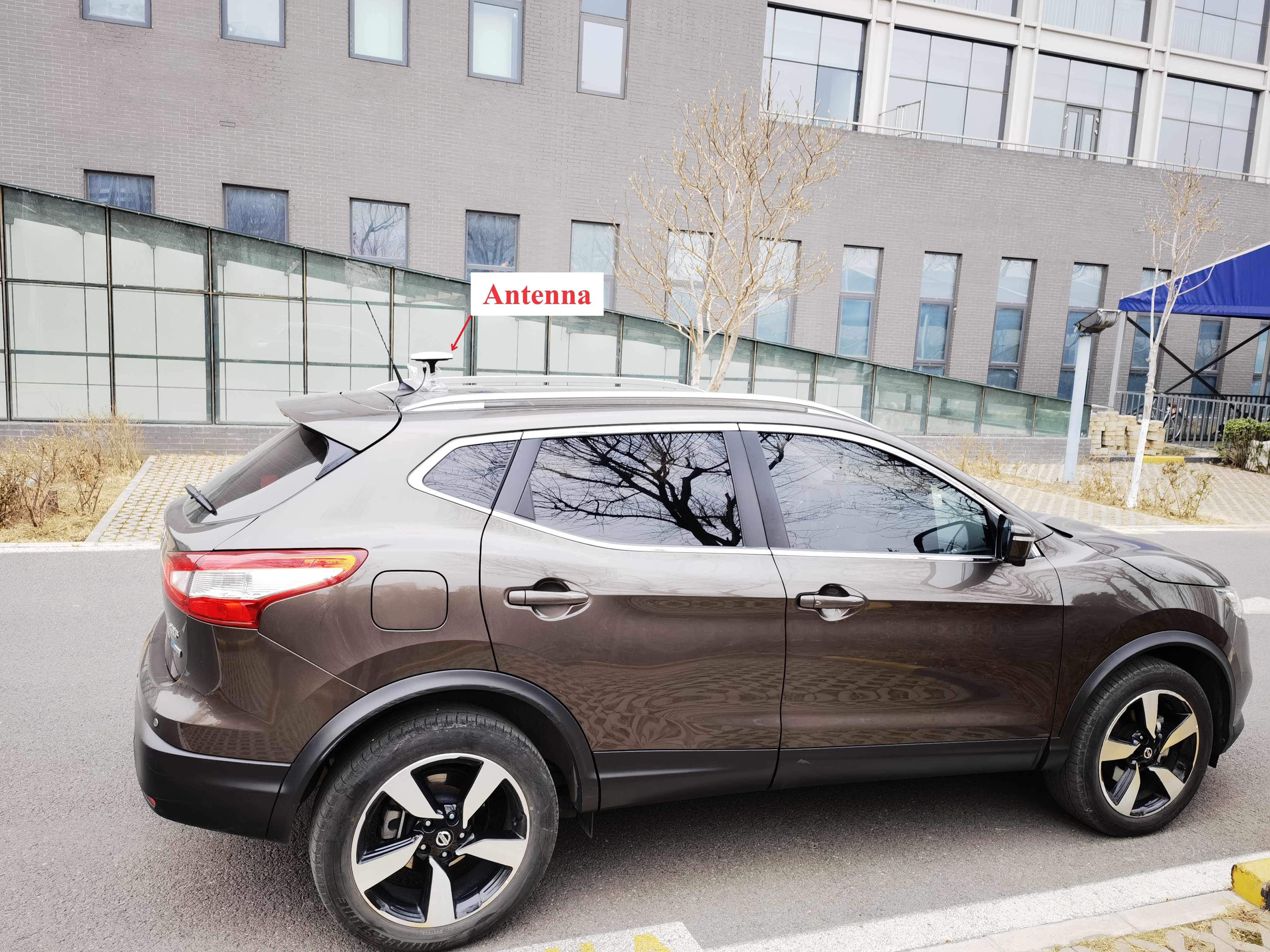}
        \end{minipage}
    }
    \subfigure[trajectory on a road]{
        \begin{minipage}[htbp]{0.46\linewidth}
            \centering
            \includegraphics[height=3.1cm]{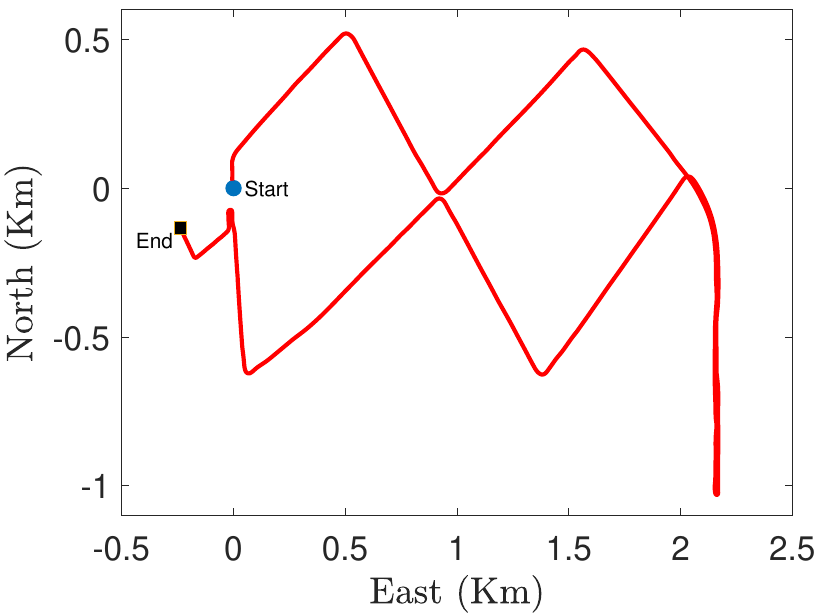}
        \end{minipage}
    }
     
    \caption{Car and its trajectory (data credit: UniStrong Co., Ltd. and \cite{wang2023distributionally}).}
    \label{fig:car}
\end{figure}

The east axis (in the east-north-up coordinate) is investigated \citep{wang2023distributionally}. The RTAMSE performance against the values of $(\alpha, \beta)$ is shown in Figure \ref{fig:car-results}.
\begin{figure}[!htbp]
    \centering
    \includegraphics[height=4.0cm]{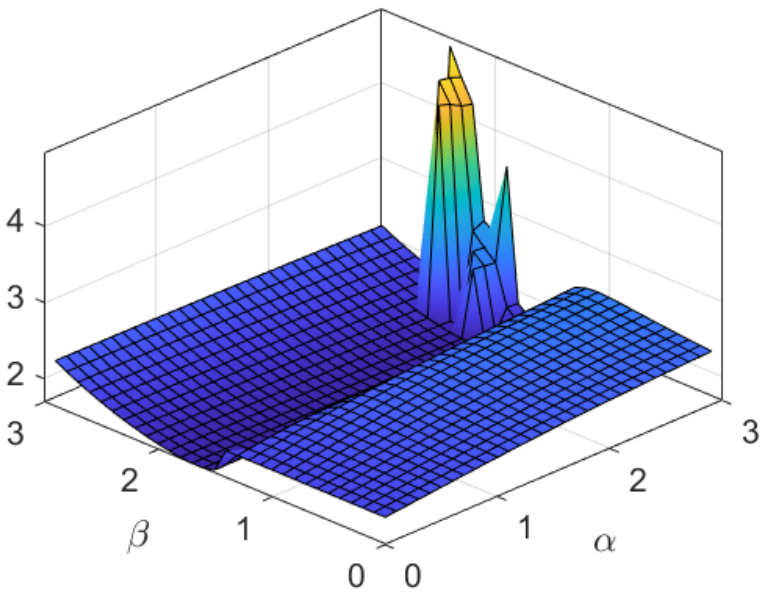}
    \caption{Averaged-RTAMSE performance of the $(\alpha,\beta)$-uncertainty-aware IMM filter, against the values of $(\alpha, \beta)$. The $(\alpha, \beta)$-uncertainty-aware IMM filter with $(\alpha, \beta) = (2, 1.7)$ works best.}
    \label{fig:car-results}
\end{figure}

The results suggest that the $(\alpha, \beta)$-uncertainty-aware IMM (IMM-UA) filter with $(\alpha, \beta) = (2, 1.7)$ works best, under which the tracking results (RTAMSE) are shown in Table \ref{tab:car}.
\begin{table}[!htbp]
\centering
\caption{Tracking Results of The Car: $(\alpha, \beta) = (2, 1.7)$}
\label{tab:car}
\begin{tabular}{lrr|lrr}
\bottomrule[0.8pt]
Filter  & RTAMSE & Avg Time & Filter  & RTAMSE & Avg Time \\
\specialrule{0.7pt}{0pt}{0pt}
IMM  & 2.30 & 1.05e-05 & IMM-UA  & 1.69 & 1.07e-05 \\
\toprule[0.8pt]
\multicolumn{6}{l}{\tabincell{l}{\footnotesize 
{Avg Time}: Average Execution Time at each time step (unit: seconds).
}}
\end{tabular}
\end{table}

From Table \ref{tab:car}, we can see that the value pair of $(\alpha, \beta) = (2, 1.7)$ significantly reduces the tracking errors for the car. This value pair concentrates both the prior distribution (i.e., prior model weights) and the likelihood distribution (i.e., model likelihoods), which means that one of the three models in \eqref{eq:IMM-model} dominates the rest two models most of the time. This implication coincides with our intuition from Figure \ref{fig:car}(b) that most of the time $k$, the model with $\rscl a_{1, k} = 0$ (i.e., constant velocity and straight-line trajectory) dominates the motion of the car.

\textbf{Track A Highly-Maneuvering Drone}: 
We next track a highly-maneuvering drone that flies following round trajectories in the air over an open playground, with a flying speed of about $6$m/s during data collection. The drone and parts of its trajectory are shown in Figure \ref{fig:drone}.
\begin{figure}[!htbp]
    \centering
    \subfigure[drone]{
        \begin{minipage}[htbp]{0.46\linewidth}
            \centering
            \includegraphics[height=3.0cm]{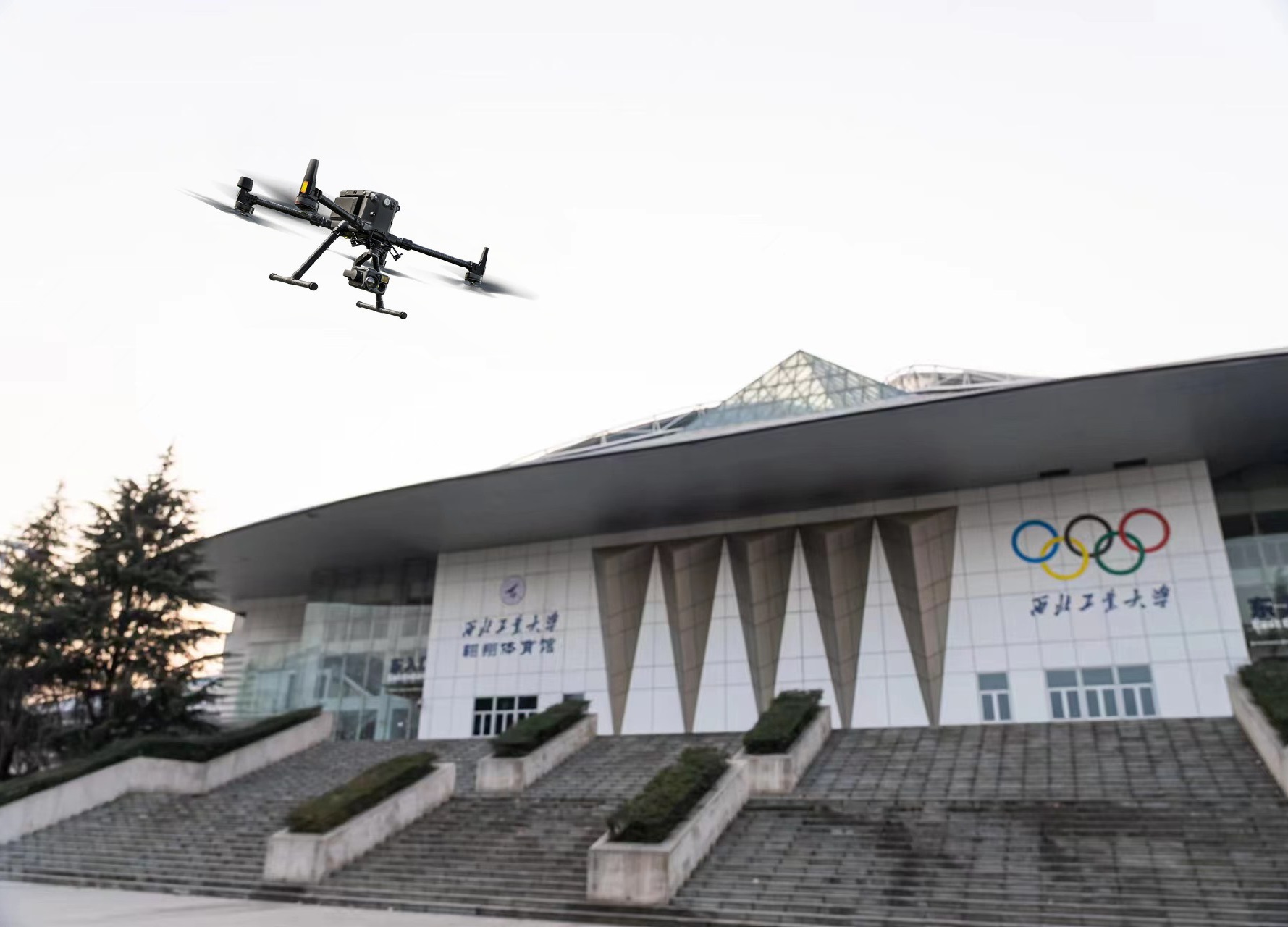}
        \end{minipage}
    }
    \subfigure[trajectory (part)]{
        \begin{minipage}[htbp]{0.46\linewidth}
            \centering
            \includegraphics[height=3.0cm]{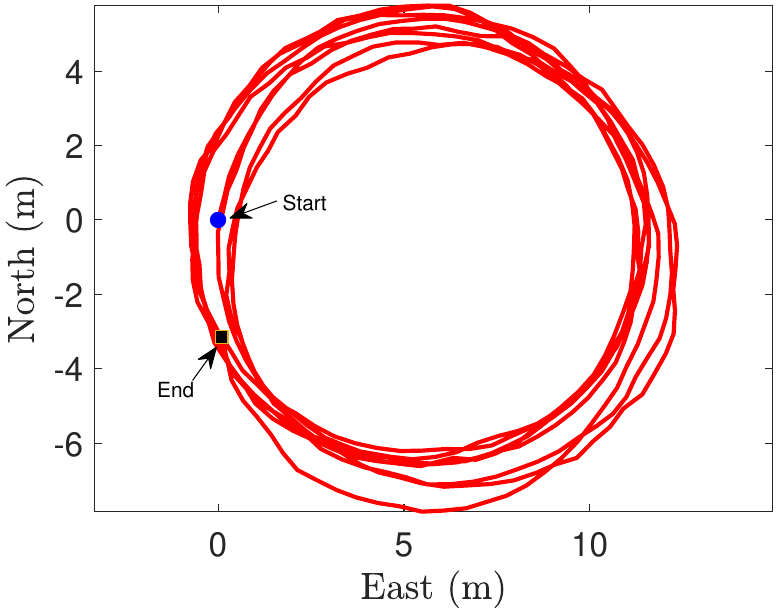}
        \end{minipage}
    }
     
    \caption{Drone and its flying trajectory, starting from $10s$ and ending at $60s$ (data credit: Northwestern Polytechnical University and \cite{wang2023distributionally}).}
    \label{fig:drone}
\end{figure}

The east axis (in the east-north-up coordinate) is investigated \citep{wang2023distributionally}. The RTAMSE performance against the values of $(\alpha, \beta)$ is shown in Figure \ref{fig:drone-results}.
\begin{figure}[!htbp]
    \centering
    \includegraphics[height=4.0cm]{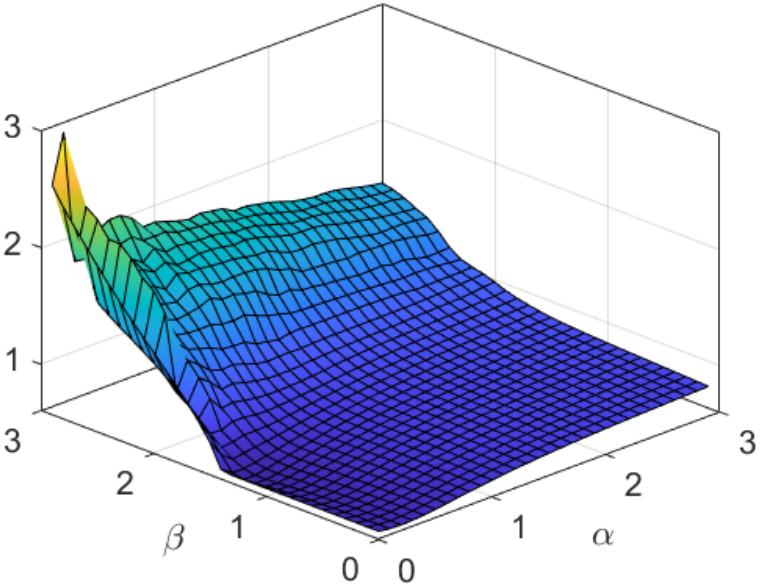}
    \caption{Averaged-RTAMSE performance of the $(\alpha,\beta)$-uncertainty-aware IMM filter, against the values of $(\alpha, \beta)$. The $(\alpha, \beta)$-uncertainty-aware IMM filter with $(\alpha, \beta) = (0.2, 1.2)$ works best.}
    \label{fig:drone-results}
\end{figure}

The results suggest that the $(\alpha, \beta)$-uncertainty-aware IMM filter with $(\alpha, \beta) = (0.2, 1.2)$ works best, under which the tracking results (RTAMSE) are shown in Table \ref{tab:drone}.
\begin{table}[!htbp]
\centering
\caption{Tracking Results of The Drone: $(\alpha, \beta) = (0.2, 1.2)$}
\label{tab:drone}
\begin{tabular}{lrr|lrr}
\bottomrule[0.8pt]
Filter  & RTAMSE & Avg Time & Filter  & RTAMSE & Avg Time \\
\specialrule{0.7pt}{0pt}{0pt}
IMM  & 0.77 & 2.71e-05 & IMM-UA  & 0.60 & 2.85e-05 \\
\toprule[0.8pt]
\multicolumn{6}{l}{\tabincell{l}{\footnotesize 
{Avg Time}: Average Execution Time at each time step (unit: seconds).
}}
\end{tabular}
\end{table}

From Table \ref{tab:drone}, we can see that the value pair of $(\alpha, \beta) = (0.2, 1.2)$ significantly reduces the tracking errors for the drone. This value pair improves the entropy (i.e., the spread) of the likelihood distribution (i.e., the model likelihoods) and almost does not influence the prior distribution (i.e., the prior model weights because $\beta = 1.2 \approx 1$), which means that none of the three models in \eqref{eq:IMM-model} dominates the rest two models and the model set in \eqref{eq:IMM-model} is not complete (viz., more candidate values for $j_k$ and $\rscl a_{j_k, k}$ are expected; only three values are not sufficient).\footnote{For example, $\rscl a_{j_k, k} \defeq \{0, -2.5, -5, -7.5, -10, 2.5, 5, 7.5, 10\}$ should be better; however, this introduces much more computational loads in IMM filter. For extensive reading on this point, see \citet[Subsec.~VI.A.2.;~Tab.~III]{wang2023distributionally}.} This implication coincides with our intuition from Figure \ref{fig:drone}(b) that, since the drone highly maneuvers with acceleration quickly switching between positive values and negative values, there is no model that dominates the motion of the drone.

\subsubsection{Remarks for State Estimation Applications}\label{subsec:remarks-expt-parameter-tuning-sig-proc}
From our experiments, we found that the grid search method with the step size of $0.01$ is an empirical golden rule. As for $\tau$ such that $(\alpha, \beta) \in [0, \tau]^2$, experiments suggest that $\tau \le 3$ is practically sufficient.

\end{document}